\newif\ifbookmargins 
\newif\ifcoverpage 
\newif\ifoneside 

\newif\ifdraft 

\newif\ifdoubleview 
\newif\iffullscreen 

\documentclass[12pt, a4paper]{article}

\usepackage{etoolbox} 
\usepackage{verbatim} 
\usepackage{amsmath}
\usepackage{amsthm} 
\usepackage{amsfonts} 
\usepackage{amssymb} 
\usepackage{dsfont} 
\usepackage{bibentry}
\usepackage{latexsym}
\usepackage{amsfonts}
\usepackage{amstext}
\usepackage{enumerate}

\usepackage{graphicx}
\usepackage{float}
\usepackage{caption}
\usepackage{subcaption}
\usepackage{wrapfig}

\usepackage{anyfontsize}
\usepackage{fancyhdr} 
\usepackage[calcwidth,noindentafter]{titlesec} 

\usepackage[T1]{fontenc}
\usepackage[utf8]{inputenc}
\usepackage[english]{babel}

\usepackage[unicode=true]{hyperref} 

\newcommand*\parttitle{BUG} 

\titleformat{\part}[block]
  {\LARGE\bf\filleft}
  {\sc{\partname{}\ }\thepart\vspace{5mm}}
  {0pt}
  {\titleline*[r]{\titlerule[.7pt]}\vspace{5mm}\fontsize{35pt}{1.5em}\selectfont}

\makeatletter
\let\titlesec@part\part
\renewcommand*{\part}[2][]{
\hypersetup{bookmarksopen=true,bookmarksopenlevel=1}
   \cleardoublepage
   \ \\[0.2\textheight]
   \begin{flushright}
  \begin{minipage}{0.7\textwidth} 
   \ifx\\#1\\
      \titlesec@part{#2}
      \end{minipage}
      \end{flushright}
      \renewcommand*\parttitle{#2}
      \sectionmark[*]{\textsc{\partname{}\ }\thepart.\ #2}
   \else
      \titlesec@part[#1]{#2}
      \end{minipage}
      \end{flushright}
      \renewcommand*\parttitle{#1}
      \sectionmark[*]{\textsc{\partname{}\ }\thepart.\ #1}
   \fi
   \cleardoublepage
}
\makeatother

\numberwithin{equation}{subsection}

\newcommand\RR{\mathds{R}}
\newcommand\ZZ{\mathds{Z}}
\newcommand\CC{\mathds{C}}

\newcommand\one{\mathds{1}}

\newcommand{\lotimes}{\mathop{\otimes}\limits}
\newcommand{\loplus}{\mathop{\oplus}\limits}

\newcommand{\nph}{{\rm nph}}
\renewcommand{\d}{{\operatorname{d}}}
\newcommand{\e}{\operatorname{e}}
\newcommand{\Exc}{{\rm Exc}}
\newcommand{\exc}{{\epsilon}}
\newcommand{\0}{\mathbf{0}}
\renewcommand{\sp}{{\operatorname{sp}}}
\newcommand{\cl}{{\operatorname{cl}}}
\newcommand{\Ran}{{\operatorname{Ran}}}
\renewcommand{\a}{{\operatorname{a}}}
     \renewcommand{\d}{{\operatorname{d}}}
\newcommand{\s}{{\operatorname{s}}}
\newcommand{\sah}{{\operatorname{\varsigma}}}
\newcommand{\crit}{{\operatorname{cr}}}
\newcommand{\ess}{{\rm ess}}
\newcommand{\fr}{{\rm fr}}
\newcommand{\sgn}{{\operatorname{sgn}}}
\newcommand{\vecsp}{\overrightarrow{\operatorname{sp}}}
\newcommand{\hc}{{\rm hc}}
\newcommand{\ext}{{\operatorname{ext}}}

\newcommand\x{\mathbf{x}}

\newcommand\y{\mathbf{y}}
\newcommand\kk{\mathbf{k}}
\newcommand\p{\mathbf{p}}
\newcommand\q{\mathbf{q}}

\newcommand\zz{\mathrm{z}}

\newcommand\ii{\mathrm{i}}

\newcommand\fG{\mathfrak{G}}

\newcommand\cH{\mathcal{H}}

\newcommand\cQ{\mathcal{Q}}
\newcommand\cK{\mathcal{K}}

\newcommand\cZ{\mathcal{Z}}

\newcommand\bep{\begin{proposition}}
\newcommand\eep{\end{proposition}}
\newcommand\ber{\begin{remark}}
\newcommand\eer{\end{remark}}

\newcommand{\Bog}{{\operatorname{Bog}}}
\def\sa{{\rm s/a}}
\def\fG{{\mathfrak G}}

     \newtheorem{thm}{Theorem}[section]
     \newtheorem{prop}[thm]{Proposition}
     \newtheorem{lemma}[thm]{Lemma}
     \newtheorem{cor}[thm]{Corollary}

     \newtheorem{remark}[thm]{Remark}

\let\theparentequation\theequation
\patchcmd{\theparentequation}{equation}{parentequation}{}{}

\renewenvironment{subequations}[1][]{
  \refstepcounter{equation}%
  \setcounter{parentequation}{\value{equation}}
  \setcounter{equation}{0}
  \let\parentlabel\label
  \ifx\\#1\\\relax\else\label{#1}\fi
  \ignorespaces
}{%
  \setcounter{equation}{\value{parentequation}}
  \ignorespacesafterend
}

\newcommand*{\nextParentEquation}[1][]{
  \refstepcounter{parentequation}
  \setcounter{equation}{0}
  \ifx\\#1\\\relax\else\parentlabel{#1}\fi
}

\newcommand\longtitle{Excitation spectrum and quasiparticles in quantum gases. A rigorous approach.}
\newcommand\authors{Marcin Napiórkowski}
\newcommand\headtitle{Excitation spectrum and quasiparticles in quantum gases.} 

\hypersetup{
  pdfauthor={\authors},
  pdftitle={\longtitle},
  pdfsubject={spectral analysis},
  pdfkeywords={Bogoliubov approximation},
 }

 \hypersetup{
   pdfstartview={FitH},
   pdffitwindow=false,
   pdfborder={0 0 0},
   colorlinks=true,
   linktoc=page,
   linkcolor=blue,
   urlcolor=blue,
   hypertexnames=false,
   pdfdisplaydoctitle=true,
   pdfduplex=DuplexFlipLongEdge,
   ,
 }

\ifcoverpage\hypersetup{pdfstartview=FitV}\fi

\ifdoubleview\hypersetup{%
   pdfstartview={Fit}, 
   pdfpagemode=UseOutlines,
   pdfpagelayout=TwoPageRight,
   ,
   ,
}\fi

\iffullscreen\hypersetup{%
   pdfstartview={Fit}, 
   pdfnonfullscreenpagemode=UseOutlines, 
   pdfpagemode=FullScreen, 
}\fi

\ifoneside\hypersetup{pdfduplex=Simplex}\fi

\ifdraft\hypersetup{draft=true}\fi

\ifbookmargins
 \def\marginasymm{3.5mm}
\else
 \def\marginasymm{0mm}
\fi

\titlespacing{\section}{0pt}{0em}{1.8em}
\titlespacing{\subsection}{0pt}{1.6em}{1.2em}

\begin{document}
\pagestyle{fancy}
\fancyhf{} 
\renewcommand{\sectionmark}[2][]{%
\ifx\\#1\\
    \markright{\ifnum\value{section}=0{}\else{\thesection.{\ }}\fi#2}
 \else{
   \markright{#2}
 }
\fi
}
\renewcommand{\subsectionmark}[1]{}
\ifoneside{
  \lhead{\nouppercase{\textsf{\rightmark}}}
  \rhead{\textsf{\thepage}}
}
\else
  \fancyhead[RE]{\textsf{\ifnum\value{part}=0{\headtitle}\else{\textsc{\partname{}{\ }}\thepart.{\ }\parttitle}\fi}}
  \fancyhead[LO]{\nouppercase{\textsf{\rightmark}}}
  \fancyhead[LE,RO]{\textsf{\thepage}}
\fi
\setcounter{page}{-1} 
\pagestyle{empty} 
\begin{center}
  \begin{figure}[h]
  \centering\includegraphics[height=8em]{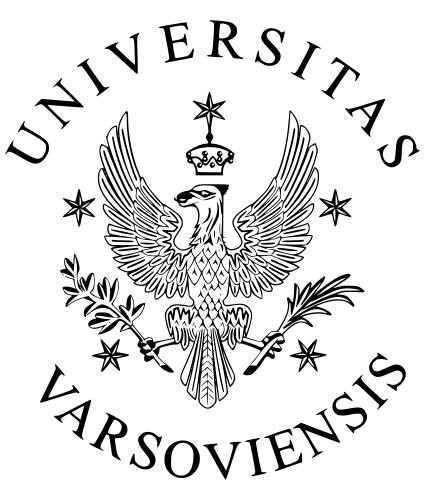}
  \end{figure}
\begin{minipage}[c]{0.8\textwidth}
  \centering{
  \LARGE{\ \\\textsc{University of Warsaw\\Faculty of Physics\\[12mm]}}
  \normalsize{\textsc{Ph.D. thesis }\\[2mm]}
  \rule{\linewidth}{0.5mm}\\[0.7em]
  \Large{\textbf{\longtitle}}\\
  \rule{\linewidth}{0.5mm}\\[2mm]
  \ \\
  }
 \end{minipage}
\begin{minipage}{0.37\textwidth}
 \begin{flushleft}
  \emph{Author:}\\
  Marcin \textsc{Napiórkowski}
 \end{flushleft}
\end{minipage}
\begin{minipage}{0.37\textwidth}
  \begin{flushright}
  \emph{Supervisor:}\\
  Prof. Jan \textsc{Dereziński}
  \end{flushright}
\end{minipage}
\vfill
A dissertation submitted for the degree of\\Doctor of Philosophy\\
at the University of Warsaw\\[1cm]
\begin{minipage}{0.21\textwidth}
 \begin{flushleft}
  \emph{\today}
 \end{flushleft}
\end{minipage}
\begin{minipage}{0.21\textwidth}
  \begin{flushright}
  \includegraphics[height=5em]{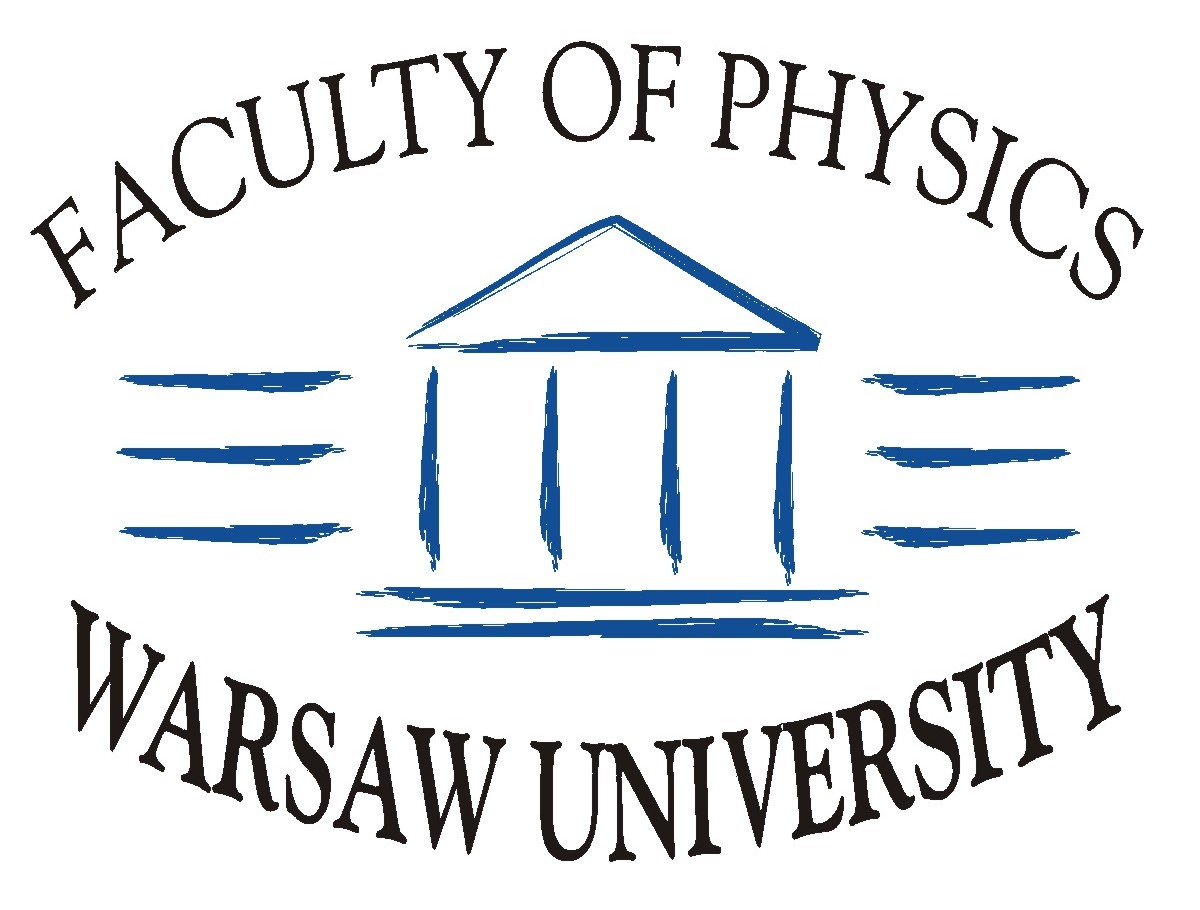}
  \end{flushright}
\end{minipage}
\end{center}
\cleardoublepage
\pagestyle{fancy}
\vspace*{\fill} 

\newpage \section*{Abstract}

This thesis is devoted to a rigorous study of interacting quantum gases. The main objects of interest are the closely related concepts of \textit{excitation spectrum} and \textit{quasiparticles}. The immediate motivation of this work is to propose a spectral point of view concerning these two concepts.

In the first part of this thesis we discuss the concepts of excitation spectrum and quasiparticles. We provide an overview of physical motivations and based on that we propose a spectral and Hamiltonian-based approach towards these terms. Based on that, we formulate definitions and propositions related to these concepts.

In the second part we recall the Bogoliubov and Hartree--Fock--Bogoliubov approximations, which in the physics literature are used to obtain the quasiparticle picture. We show how these two approaches fit into a universal scheme which allows us to arrive at a quasiparticle picture in a more general setup. This scheme is based on the minimization of Hamiltonians over the so-called \textit{Gaussian states}. Its abstract formulation is the content of \textit{Beliaev's Theorem}.

In the last part we present a rigorous result concerning the justification of the Bogoliubov approximation. This justification employs the concept of the mean-field and infinite-volume limit. We show that for a large number of particles, a large volume and a sufficiently high density, the low-lying energy-momentum spectrum of the homogeneous Bose gas is well described by the Bogoliubov approximation. This result, which is formulated in the form of a theorem, can be seen as the main result of this thesis.

\newpage
\thispagestyle{empty}
\mbox{}

\newpage \section*{Streszczenie}

Niniejsza rozprawa poświęcona jest ścisłej analizie gazów kwantowych, a w szczególności zbadaniu pojęć \textit{kwaziczątek} i \textit{widma wzbudzeń} w takich układach. Przedstawione w pracy podejście do tych terminów oparte jest o spektralne własności układów hamiltonowskich.

Rozprawa składa się z trzech części. Pierwsza z nich poświęcona jest dyskusji pojęć kwazicząstek i widma wzbudzeń. Przedstawiony został fizyczny kontekst pojawienia się tych terminów. Posłużył on za motywację do wprowadzenia innych - opartych o formalizm hamiltonowski i własności spektralne układów - definicji tych pojęć. Przedstawione zostały bezpośrednie konsekwencje zaproponowanych definicji. 

Druga część pracy poświęcona jest przybliżeniom Bogoliubowa oraz Hartree--Focka--Bogoliubowa. Nazwy te odnoszą się do schematów, które pozwalają na przybliżone wyznaczenie widma wzbudzeń gazu oddziałujących bozonów i fermionów. Pokazane zostało jak procedury te wpisują się w pewien ogólny schemat rachunkowy, który pozwala na otrzymanie obrazu kwazicząsteczkowego. Ten uniwersalny schemat związany jest z minimalizacją hamiltonianów w klasie tak zwanych \textit{stanów gaussowskich}. Jego abstrakcyjne sformułowanie ujęte jest w postaci \textit{twierdzenia Beliaewa}.

W trzeciej części rozprawy przedstawione zostały ścisłe wyniki dotyczące przybliżenia Bogoliubowa. Udowodniono poprawność tego przybliżenia w ramach teorii pola średniego, w tak zwanej granicy dużej gęstości i słabego sprzężenia. Wynik ten - sformułowany w postaci twierdzenia - uznać można za główny rezultat tej rozprawy.

\newpage
\thispagestyle{empty}
\mbox{}

\newpage \section*{Acknowledgements}
I express my deep gratitude to my advisor Jan Dereziński for giving me the opportunity to work in this interesting field of mathematical physics. This work would not be possible without his constant support and advice.

Part of this thesis is based on a joint work with Jan Philip Solovej who hosted me in Copenhagen in spring 2012. I am grateful for his warm hospitality, patience and many hours of educating discussions.

I would like to thank Peter Pickl and his group for their warm hospitality during my stay in Munich in spring 2013. Special thanks go to my colleagues from the Chair of Mathematical Methods in Physics at the University of Warsaw, especially Przemek Majewski and Paweł Kasprzak. 

Last, but certainly not least, I thank my family for their constant support during my Ph.D. studies.
\newpage
\thispagestyle{empty}
\mbox{}

\newpage\tableofcontents
\newpage
\thispagestyle{empty}
\mbox{}
\section[Introduction]{Introduction}
\sectionmark{Introduction}

One of the main goals of many-body quantum mechanics is to determine how properties of a system of interacting particles differ from those of  non-interacting particles. One can look at different features, e.g., the ground-state energy, excited energy levels, various temperature-dependent quantities and phenomena like free energy, specific heat, phase transitions, etc. These issues, due to the complexity of the system, are usually difficult to analyze rigorously. 

In general, the concept of quasiparticles in a translation invariant quantum system refers to the following idea: as a consequence of the interaction between particles, the single-particle motion becomes considerably modified. As a particle moves along, it drags some particles with it, repels others, etc.  Because of that interaction, one observes a different dispersion relation (energy versus momentum dependence) than for a freely moving particle. The expectation is that, for very low temperatures, the system may be regarded as made up of \textit{quasiparticles} - that is new, independent particles with a specific dispersion relation that depends on many features of the system \cite{Pines1962}.

As a consequence, one expects that quasiparticles provide a description of low-lying excited states. This makes the concepts of \textit{quasiparticles} and \textit{excitation spectrum} - by which we mean the joint energy-momentum spectrum with the ground-state energy subtracted - closely related. 

Historically, the idea of quasiparticles appeared first in the context of condensed matter problems, both in reference to bosonic \cite{Bogoliubov1947} and fermionic \cite{Landau1957,Landau1959} systems. It subsequently spread to other fields of physics, among others to the theory of atomic nuclei and plasma theory.\\
\textit{}\\
One of the possible ways to put the concept of quasiparticles on mathematical basis, is the introduction of the \textit{single-particle Green's function}. In this approach, which will be described in Section \ref{quasigreen}, the energy and lifetime of a quasiparticle are determined by a pole in the analytic continuation of the Green's function. This description has several advantages. In particular, it allows for the \textit{macroscopic} limit (we will use the terms macroscopic and \textit{infinite-volume} interchangeably) without referring to any kind of macroscopic limit of the Hamiltonian of the system. Moreover, one should \textit{not} expect that there exists a Hamiltonian corresponding to the Green's function in the infinite-volume limit.

On the other hand, one could argue that, since quantum theories based on Hamiltonians are in some sense more fundamental, a purely spectral approach to the concept of quasiparticles would be desirable. This was the main motivation for this thesis. In particular, in the first part of the thesis we propose an approach towards the concept of quasiparticles which does not refer to the notion of Green's functions. It is based on spectral properties of Hamiltonians and will be the keynote of this work. 

There exist many papers that study the energy spectrum of systems consisting of interacting fermions and bosons. In particular, there exist interesting works that study the Bogoliubov and Hartree--Fock--Bogoliubov approximations. In our approach, we would like to stress the significance of  \textit{translation invariance} of the systems considered. This enables us to ask questions about the excitation spectrum,  which we expect to have interesting properties. Some of these features play an important role in explaining various phenomena in condensed matter theory such as \textit{superfluidity} in bosonic systems \cite{Landau1941} and \textit{superconductivity} in fermionic ones \cite{BCS1957}.

Several papers discuss the way in which the models based on the concept of quasiparticles approximate the properties of models based on realistic Hamiltonians. However, only relatively crude features are  considered in essentially all these papers. Typically, they study the energy or the free energy per volume in the thermodynamic limit. We are particularly interested in the excitation spectrum. There are rather few rigorous results on the excitation spectrum of an interacting quantum system. One of them has been derived by the author and is presented in the third part of this dissertation. \\
\textit{}\\
The thesis is organized into three parts. They are related, but can be read independently.    

The first part of the thesis consists of Sections \ref{quasigreen} and \ref{excandquasi}. In Section \ref{quasigreen}, based on the books of Pines (\cite{Pines1962}) and Fetter-Walecka (\cite{FetterWalecka1971}), we present an approach towards the concept of quasiparticles that one encounters in many textbooks. It is based on field-theoretical techniques in many-body quantum mechanics and the concept of Green's functions. In Section \ref{excandquasi}, we present a different approach towards quasiparticles. It is formulated in terms of spectral properties of Hamiltonian systems. We would like to stress that the approach presented in this part of the thesis is quite general and applies to any translation invariant system. Section \ref{excandquasi} is based mainly on the article \textit{On the energy-momentum spectrum of a homogeneous Fermi gas} (Annales Henri Poincar\'e 14, 1-36 (2013))   written by Jan Derezi\'nski, Krzysztof A. Meissner and the author. 

Sections \ref{bosegas}, \ref{fermigassec} and \ref{beliaevthmsec} form the second part of this thesis in which we introduce approximate methods that lead to the derivation of the excitation spectrum. In Section \ref{bosegas}, we present the \textit{Bogoliubov approximation} and the so-called \textit{improved Bogoliubov method}. This section is based on the papers \cite{Bogoliubov1947} and \cite{CDZ2009}. Although these articles are not co-authored by the author of this thesis, we include some parts thereof in our presentation because they provide bosonic counterparts of the \textit{Hartree--Fock--Bogoliubov approximation} presented in Section \ref{fermigassec}. The latter section is also based on the article \cite{DMN2013} mentioned in the previous paragraph. In Section \ref{beliaevthmsec}, the last one in the second part of the thesis, we show how the above approximate methods fit into a universal scheme which allows us to arrive at a quasiparticle picture in a more general setup. This scheme is based on the minimization of Hamiltonians over the so-called \textit{Gaussian states}. Section \ref{beliaevthmsec} is based on the article \textit{On the minimization of Hamiltonians over pure Gaussian states} which appeared in the book "Complex Quantum Systems. Analysis of Large Coulomb Systems" (World Scientific 2013), written by Jan Derezi\'nski, Jan Philip Solovej and the author.

In the third and last part of the thesis, we present a rigorous justification of the Bogoliubov approximation. This justification employs the concept of the mean-field and infinite-volume limits. We show that for a large number of particles, a large volume and a sufficiently high density, the low-lying energy-momentum spectrum of the homogeneous Bose gas is well described by the Bogoliubov approximation. In particular, we give explicit bounds on the error terms. In Section \ref{s1}, we formulate that statement in a precise way, while in Section \ref{proofba} we present the proof of that statement. These sections are based on the article \textit{Excitation spectrum of interacting bosons in the mean-field infinite-volume limit} written by Jan Derezi\'nski and the author of this thesis. It will appear in Annales Henri Poincar\'e (DOI 10.1007/s00023-013-0302-4). 

We close the thesis with a short summary and present possible directions of future research.

\part{Excitation spectrum and quasiparticles}
\let\oldsection\section
\renewcommand\section{\clearpage\oldsection}

\section[Quasiparticles via Green's function]{Quasiparticles via Green's function}\label{quasigreen}
One of the possible ways to put the notion of quasiparticles on a reasonable mathematical basis, is the introduction of a concept used in quantum field theory, namely the single-particle propagator or \textit{single-particle Green's function}. Below, we give a brief presentation of that approach. It is based on the classical books of Fetter--Walecka \cite{FetterWalecka1971} and Pines \cite{Pines1962}. These techniques will not be used later in this thesis. 

\subsection{Single-particle Green's function}
Consider a spatially homogeneous (i.e. translation invariant) system which is described by the time-independent, many-body Hamiltonian $H$. Let $\psi_{0}$ represent its ground state with energy $E_{0}$. 

We define the \textit{single-particle Green's function}, or \textit{propagator}, as 
\begin{gather}
G(\p,\tau):=-\ii \langle \psi_{0}|T\{c_{\p}(\tau)c^{\dagger}_{\p}(0)\}|\psi_{0}\rangle,  \label{propagator}
\intertext{where}
c_{\p}(\tau)=\e^{\ii H \tau}c_{\p}\e^{-\ii H \tau} \nonumber
\end{gather}
and its hermitian conjugate are annihilation and creation operators in the Heisenberg representation (here and henceforth we assume $\hbar=1$) . $T$ is the Dyson chronological operator, which orders earlier times to the right. Thus, using \eqref{propagator}, we have 
\begin{subequations}[propagatory]
\begin{align}
G(\p,\tau)=-\ii \langle \psi_{0}|c_{\p}\e^{-\ii H \tau}c^{\dagger}_{\p}|\psi_{0}\rangle\e^{\ii E_{0} \tau} \quad \quad \tau>0 \\
  G(\p,\tau)=\pm\ii \langle \psi_{0}|c^{\dagger}_{\p}\e^{\ii H \tau}c_{\p}|\psi_{0}\rangle\e^{-\ii E_{0} \tau} \quad \quad \tau<0, 
\end{align}
\end{subequations}
where the $(+)$ sign corresponds to fermions and the $(-)$ sign to bosons. Note that we adopt the convention (as usual in physics) that the inner product is linear in the second argument. (In this section we shall suppress the spin index on fermion operators, unless needed to prevent ambiguity.)

$G(\p,\tau)$ represents the probability amplitude that if there is a particle with momentum $\p$ at $t=0$, then it will be found in that state at time $\tau>0$. In other words, $G(\p,\tau)$ describes the propagation in time of a particle in that momentum state.

We also define $G(\p,\epsilon)$ by the following relation:
\begin{eqnarray}
G(\p,\tau)=:\frac{1}{2\pi}\int_{\RR}G(\p,\epsilon)\e^{-\ii \epsilon \tau}\d\epsilon. \label{gpe}
\end{eqnarray}
$G(\p,\tau)$ is defined in the momentum representation. Of course, one can also define the space-dependent single-particle Green's function
$$ G(\x,\tau)=-\ii \langle \psi_{0}|T\{a_{\x}(\tau)a_{\x}^{\dagger}(0)\}|\psi_{0}\rangle.$$
Here $a_{\x}(\tau)$ and $a_{\x}^{\dagger}(0)$ are the second quantized field operators in the Heisenberg representation. Clearly, $G(\p,\tau)$ is the Fourier transform of $G(\x,\tau)$.

\subsection{Lehmann representation}
Assume there is a complete set of eigenstates $|\psi_{n}\rangle$, that is, $(N+1)$-particle states with momentum $\p$ such that $H\psi_{n}=E_{n}\psi_{n}$ and $\sum_{n} |\psi_{n}\rangle\langle\psi_{n}|=\one$. By introducing these states into the definition \eqref{propagatory}, we find for $\tau>0$
\begin{eqnarray*}
G(\p,\tau) &=& -\ii\sum_{n}\langle\psi_{0}|c_{\p}|\psi_{n}\rangle\langle\psi_{n}|\e^{-\ii H \tau}c^{\dagger}_{\p}|\psi_{0}\rangle \e^{\ii E_{0}\tau}\\
&=& -\ii\sum_{n}|\big(c^{\dagger}_{\p}\big)_{n0}|^2\e^{-\ii(E_{n}-E_{0})\tau}, 
\end{eqnarray*}
where $\big(c^{\dagger}_{\p}\big)_{n0}=\langle\psi_{n}|c^{\dagger}_{\p}|\psi_{0}\rangle$. We now write
\begin{eqnarray}
E_{n}(N+1)-E_{0}(N) &=& E_{n}(N+1)-E_{0}(N+1)+E_{0}(N+1)-E_{0}(N)\\
&=&\omega_{n0}+\mu  \label{wzbudzenia}
\end{eqnarray} 
where $\omega_{n0}$ is the excitation energy in the $(N+1)$-particle system and $\mu=E_{0}(N+1)-E_{0}(N)$. 

Note, that we allowed for states with different numbers of particles and thus it is natural to consider the problem in the grand-canonical approach where the number of particles is not fixed. In the grand-canonical setting one fixes the \textit{chemical potential}. From statistical mechanics we know that, for a system at zero temperature and of constant volume, the chemical potential describes the change of the energy of the system upon adding a particle to the system. Thus $\mu$ in \eqref{wzbudzenia} has the natural interpretation of the chemical potential.

   Doing the same for $\tau<0$ leads to
\begin{eqnarray*}
G(\p,\tau) = \pm\ii\sum_{n}|\big(c_{\p}\big)_{n0}|^2\e^{\ii(E_{n}-E_{0})\tau}, 
\end{eqnarray*}
where the eigenstates states now correspond to states of momentum $-\p$ of the $(N-1)$-particle system. Again one can write
$$E_{n}(N-1)-E_{0}(N)=\omega_{n0}'-\mu'$$
where $\omega_{n0}'$ is an excitation energy in the $(N-1)$-particle system and $\mu'$ is the chemical potential for going from $N-1$ to $N$ particles. In the large $N$ limit one can expect that 
$$\omega_{n0}=\omega_{n0}'\quad\quad \text{and} \quad\quad \mu=\mu'$$
to an accuracy of order $1/N$. With this assumption
\begin{subequations}[propagatory2]
\begin{align}
G(\p,\tau)=-\ii\sum_{n}|\big(c^{\dagger}_{\p}\big)_{n0}|^2\e^{-\ii(\omega_{n0}+\mu)\tau}, \quad \quad \tau>0 \\
  G(\p,\tau)= \pm\ii\sum_{n}|\big(c_{\p}\big)_{n0}|^2\e^{\ii(\omega_{n0}-\mu)\tau}, \quad \quad \tau<0
\end{align}
\end{subequations}
where the excitation energies $\omega_{n0}$ are necessarily positive. By introducing the \textit{spectral functions} 
\begin{eqnarray*}
A(\p,\omega)&=&\sum_{n}|\big(c^{\dagger}_{\p}\big)_{n0}|^2\delta(\omega-\omega_{n0}),\\
B(\p,\omega)&=&\sum_{n}|\big(c_{\p}\big)_{n0}|^2\delta(\omega-\omega_{n0})
\end{eqnarray*}
we can rewrite \eqref{propagatory2} in the following form

\begin{subequations}[propagatoryspektralne]
\begin{align}
G(\p,\tau)=-\ii\int_{0}^{\infty}A(\p,\omega)\e^{-\ii(\omega+\mu)\tau}\d\omega, \quad \quad \tau>0 \label{propagatoryspektralneczasdodatni}\\
  G(\p,\tau)= \pm\ii\int_{0}^{\infty}B(\p,\omega)\e^{\ii(\omega-\mu)\tau}\d\omega, \quad \quad \tau<0.
\end{align}
\end{subequations}
Furthermore, using contour integration, we obtain
\begin{eqnarray}
G(\p,\epsilon)=\int_{0}^{\infty}\left[\frac{A(\p,\omega)}{\epsilon-(\omega+\mu)+\ii\eta}+\frac{B(\p,\omega)}{\epsilon+\omega-\mu-\ii\eta}\right]\d\omega. \label{lehmann}
\end{eqnarray}
Here $\pm\ii\eta$ is an infinitesimal quantity which specifies the position of the pole in the complex $\epsilon$ plane. 

Formula \eqref{lehmann} is known as the Lehmann representation in quantum field theory \cite{Lehmann1954}.

\subsection{Quasiparticles as poles of the propagator}\label{polepropagator}
So far we have assumed that all eigenvalues are discrete. This is true for finite systems. With this assumption \eqref{wzbudzenia} and \eqref{lehmann} imply that the discrete poles of the function $G(\p,\epsilon)$ yield the excitation energies of an interacting system corresponding to total momentum $\p$.


This situations changes if one considers the macroscopic limit. In that case, the spectral function $A(\p,\omega)$ of an interacting system will no longer be a combination of $\delta$ functions. Its behaviour will be modified - one can expect it will become a continuous function. In particular, it will no longer correspond to discrete excited states.

As a consequence, the analytic structure of $G(\p,\epsilon)$ will change in the infinite-volume limit. Before the macroscopic limit was taken, $G(\p,\epsilon)$ was a \textit{meromorphic} function. After taking the infinite-volume limit the energy spectrum becomes continuous and $G(\p,\epsilon)$ has a branch cut along the real axis. One can expect that $G(\p,\epsilon)$ has an analytic continuation onto the second Riemann sheet. In particular it may have complex poles there. Thus there exists a qualitative difference between the analytic structure of $G(\p,\epsilon)$ before and after taking the macroscopic limit.  

Motivated by this argument assume that the spectral function has a peak which is smeared out (in contrast to a "pure" $\delta$ function) in the vicinity of $\omega=\tilde{\epsilon}(\p)-\mu$ in the following way
\begin{eqnarray*}
A(\p,\omega)=\frac{\ii}{2\pi}\frac{Z_{\p}}{\omega-(\tilde{\epsilon}(\p)-\mu)+\ii \Gamma_{\p}}+\text{c.c.},
\end{eqnarray*}
where c.c. means complex conjugate. Using this form of the spectral function in \eqref{propagatoryspektralneczasdodatni} leads by contour integration (a contour consisting of the positive real axis, negative imaginary axis and a curve in the lower right quadrant) to the following result  
\begin{eqnarray}
G(\p,\tau)=-\ii Z_{\p}\e^{-\ii \tilde{\epsilon}(\p)\tau}\e^{-\Gamma_\p\tau}+\,\,\text{correction terms}. \label{propagatorenergiatlumienie}
\end{eqnarray}
If one assumes that the correction terms are negligible, then one can see that $G(\p,\tau)$ describes the propagation of a state of energy $\tilde{\epsilon}(\p)$ and lifetime of $\tau\backsim 1/\Gamma_{\p}$. Thus the state has a finite lifetime and its propagation is damped. 

Thus, one arrives at the following conclusion: \textit{a pole in the analytic continuation of $G(\p,\epsilon)$ yields the energy and lifetime of a quasiparticle with momentum $\p$.}

\subsection{Remarks on quasiparticles via Green's function}
In the previous section we gave a brief presentation of an approach to quasiparticles based on analytic properties of Green's functions. It has several advantages (for more details see, e.g., \cite{FetterWalecka1971}). The biggest one is probably a natural way to detect quasiparticles with finite lifetime. Note that this procedure involves taking some kind of macroscopic limit.

Another advantage of using Green's function relies on the fact that properties derived in the infinite-volume limit \textit{do not} require an application of the concept of a macroscopic limit to the Hamiltonian - whatever that would mean. One should rather expect that there is \textit{no} Hamiltonian corresponding 
to the Green's  function after the infinite-volume limit.  

On the other hand, one could argue that quantum theories based on Hamiltonians are in some sense more fundamental, especially in the non-relativistic case. With that assumption, it would be desirable to have a different approach to quasiparticles - based on spectral properties of Hamiltonians of systems under consideration. This is done in the next section.

\section[Quasiparticles and quasiparticle-like excitation spectrum]{Quasiparticles and quasiparticle-like excitation spectrum}\label{excandquasi}
In this section we attempt to give a number of interpretations of the term \textit{quasiparticle} based on spectral properties of Hamiltonians. In particular, we will not a priori refer to field-theoretic techniques. We will also discuss spectral properties of
quantum systems that can be described in terms of quasiparticles.
The discussion of this section will be rather general and is based on \cite{DMN2013}.

\subsection{Translation invariant quantum systems}\label{transinvquansys}

From a spectral point of view, it is not completely obvious, how to describe properties of translation invariant quantum systems in {\em macroscopic limit}. There are at least two approaches that can be used to describe such systems.

In the first approach one starts with a construction of  a
 system in finite volume, using  $\Lambda=[-L/2,L/2]^d$, the $d$-dimensional cubic box of side length $L$, as the configuration space.
It is convenient, although somewhat unphysical, to impose the periodic boundary conditions. The system is described by its Hilbert space $\cH^L$, Hamiltonian $H^L$  and momentum $P^L$. The spectrum of the momentum is discrete and coincides with $\frac{2\pi}{L}\ZZ^d$. After computing appropriate quantities (such as the infimum of the excitation spectrum, which will be defined later) one tries to take the limit $L\to\infty$.

Sometimes a different approach is possible. One can try to construct a Hilbert space $\cH$, a Hamiltonian $H$ and a momentum $P$ that describe the system  on $\RR^d$. This may be not easy. It requires the use of refined techniques \cite{GlimmJaffe1987,BraRob19872} and is, probably, not always  possible. Note that in this case the spectrum of the momentum is expected to be absolutely continuous, with the exception of the ground state. 

The latter approach seems conceptually more elegant. Throughout most of this section we will adopt it. In many situations this will allow us to formulate some of the physical concepts in a concise manner.  

In the next two parts of this thesis we adopt the former approach. It is more down-to-earth and, probably, more useful when specific calculations need to be done. In this approach, only a family $(H^L,P^L)$ for finite $L$ will be defined.

To sum up, throughout  most of this section 
by a {\em translation invariant quantum system} we  will mean
$d+1$ commuting self-adjoint operators $(H,P_1,\dots,P_d)$ on a Hilbert space $\cH$. $H$ has the interpretation of a {\em Hamiltonian} and $P=(P_1,\dots, P_d)$ describes the {\em momentum}.

\subsection{Excitation spectrum}
\label{Excitation spectrum}

The joint  spectrum of the operators $(H,P)$  (which is a subset of $\RR^{1+d}$) will be denoted by $\sp(H,P)$ and called the  {\em energy-momentum spectrum of $(H,P)$}.

We will often assume that $H$ is bounded from below. If it is the case, we can define the {\em ground state energy} as $E:=\inf\sp H$. 
We will also often assume that $H$ possesses translation invariant ground state $\Phi$, which is a unique joint eigenvector of $H$ and $P$. In particular, $H\Phi=E\Phi$ and $P\Phi=0$.

Under these assumptions, by subtracting the ground state energy from the  energy-momentum spectrum we obtain the {\em excitation spectrum of $(H,P)$}, that is, $\sp(H-E,P)$.
We can also introduce the {\em strict excitation spectrum} as the joint spectrum of restriction of $(H-E,P)$ to the orthogonal complement of $\Phi$:
\begin{equation}
\Exc:=\sp\Big((H-E,P)\Big|_{\{\Phi\}^\perp}\Big).\label{exci-1}
\end{equation}
Thus if $(E,\0)$ is an isolated simple eigenvalue of $(H,P)$, then 
$$\Exc=\sp(H-E,P)\setminus (0,\0).$$
We introduce also a special notation for the infimum of $\Exc$:
\begin{eqnarray*}
\exc(\kk)&:=&\inf\{e\ :\ (e,\kk)\in\Exc\}.
\end{eqnarray*}

The following two parameters have interesting
physical implications.  The first is  the {\em energy gap}, defined as
\[\varepsilon:=\inf\Big(\sp (H-E)\Big|_{\{\Phi\}^\perp}\Big)=\inf\{\exc(\kk)\ :\ \kk\in\RR^d\}.\]
Another quantity of physical interest is the {\em critical velocity}:
\[c_{\rm{crit}}:=\inf_{\kk\neq\0}\frac{\epsilon(\kk)}{|\kk|}.\]

Physical properties of a system are especially interesting if the energy gap $\varepsilon$ is strictly positive. In such a case, the ground state energy
 is separated from the rest of the energy spectrum, and hence the ground state is {\em stable}.

Positive critical velocity is also very interesting.  Physically,  a positive critical velocity is closely related to the phenomenon
of  {\em superfluidity}, see e.g. a discussion in  \cite{CDZ2009}.

\subsection{Essential excitation spectrum}
\label{Essential excitation spectrum}

One expects that in a macroscopic limit most of a typical excitation spectrum is absolutely continuous with respect to the Lebesgue measure on $\RR^{d+1}$. However, it  may also contain  isolated shells continuously depending on the momentum. In this subsection we attempt to define  the part of the excitation spectrum that corresponds to such a situation. 


We say that $(e,\kk)\in \RR^{d+1}$ belongs to $\Exc_\d$, called the
 {\em discrete excitation spectrum}, if there exists
 $\delta>0$ such that the operator $P$ has an absolutely continuous spectrum of uniformly  finite multiplicity when restricted to 
 $$\Ran\, \one(|H-E-e|<\delta)\one(|P-\kk|<\delta).$$
The {\em essential excitation spectrum} is defined as $\Exc_\ess:=\Exc\setminus \Exc_\d$. 

Here, we use an obvious notation for spectral projections of self-adjoint operators $H$ and $P$: e.g. $\one(|H-e|<\delta)$ denotes the spectral projection of $H$ onto $]e-\delta,e+\delta[$.

We introduce also a special notation for the bottom of  $\Exc_\ess$:
\begin{eqnarray*}
\exc_\ess(\kk)&:=&\inf\{e\ :\ (e,\kk)\in\Exc_\ess\}.
\end{eqnarray*}

Obviously,
\[\Exc\supset
\Exc_\ess,\]
\[\exc(\kk)\leq\exc_\ess(\kk),\ \ \kk\in\RR^d.\]

One can expect that, typically,   $\Exc_\d$ consists of 
a finite number of  shells separated by
lacunas.

Assume now, that abstract theory allows us to
 represent the Hilbert space $\cH$ as the direct integral over $\RR^d$ given by the spectral decomposition of $P$ (see, e.g., Section 4.4.1. in  \cite{BraRob19871}). Suppose, in addition, that  this direct integral can be taken with respect to the Lebesgue measure, so that we can write
\begin{equation}
H=\int\limits_{\kk\in\RR^d}^\oplus H(\kk)\d\kk.\label{pw0}
\end{equation}
Then  it is tempting to claim that
\begin{eqnarray}
\sp(H-E,P)&=&\bigcup\limits_{\kk\in\RR^d}\sp\big( H(\kk)-E\big)\times\{\kk\},\label{pw1}\\
\Exc_\ess&=&\Big(\bigcup\limits_{\kk\in\RR^d}\sp_\ess \big(H(\kk)-E\big)\times\{\kk\}\Big)^\cl,\label{pw2}
\end{eqnarray}
where $\sp_\ess$ denotes the essential spectrum and the subscript $\cl$ denotes the closure.
Unfortunately, at this level of generality there is a problem with (\ref{pw1}) and (\ref{pw2}). First of all, there is no guarantee that we can put the Lebesgue measure in (\ref{pw0}). Secondly, the direct integral representation (\ref{pw0}) is not defined uniquely, but only modulo sets of measure zero.

 In concrete situations (such as quasiparticle systems considered in Sect.
\ref{Concept of a quasiparticle}), however,
the direct integral (\ref{pw0}) has an obvious distinguished realization involving the Lebesgue measure, for which the identities (\ref{pw1}) and (\ref{pw2}) are actually true.

\subsection{Quasiparticle quantum systems}
\label{Concept of a quasiparticle}

Let us now turn to the concept of quasiparticles in a more precise way.


For a Hilbert space $\cZ$, the notation $\Gamma_\s(\cZ)$, resp. $\Gamma_\a(\cZ)$ will stand for the {\em bosonic}, resp. {\em fermionic Fock space} with the {\em 1-particle space} $\cZ$. 

By a  {\em quasiparticle quantum system} we will mean
$(H_\fr,P_\fr)$, where
\begin{eqnarray}
H_\fr &=& \sum_{i\in\cQ}\int_{I_i}\omega_i(\kk)b_i^{\dagger}(\kk)b_i(\kk)\d\kk ,
\label{quasa1} \\
P_\fr&=&\sum_{i\in\cQ}\int_{I_i}\kk b_i^{\dagger}(\kk)b_i(\kk)\d\kk ,
\label{quasa2}
\end{eqnarray}
for some intervals $I_i\subset\RR^d$, real continuous functions $I_i\ni\kk\mapsto\omega_i(\kk)$, and (quasiparticle) {\em creation}, resp. {\em annihilation operators}  $b_i^{\dagger}(\kk)$ and $b_i(\kk)$.
 $\cQ$ is called the set of  {\em quasiparticle species} and it is partitioned into $\cQ_\s$ and $\cQ_\a$ -- bosonic and fermionic quasiparticles. 

 We are using the standard notation of the formalism of 2nd quantization:  $b_i^{\dagger}(\kk)$ and $b_i(\kk)$ satisfy the usual commutation/anticommutation relations. The right hand sides of (\ref{quasa1}) and (\ref{quasa2}) are  well defined as operators on the Fock space
\begin{equation}
\lotimes_{i\in\cQ_\s}\Gamma_\s\left(L^2(I_i)\right) \otimes
\lotimes_{j\in\cQ_\a}\Gamma_\a\left(L^2(I_j)\right). \label{fock}\end{equation}

 For $i\in\cQ$, the set $I_i$ describes the allowed  range of the momentum of a single $i$th quasiparticle and  $\omega_i(\kk)$ is its energy  (dispersion relation) for momentum $\kk\in\RR^d$. Note that $I_i$ can be strictly smaller than $\RR^d$ -- some quasiparticles may exist only for some momenta. This allows us more flexibility and is consistent with  applications to condensed matter physics. 
 It will be convenient to define 
$$I(\kk):=\{i\in\cQ\ :\ \kk\in I_i(\kk)\}$$
(the set of quasiparticles that may have momentum  $\kk\in\RR^d$).

Clearly, if we know the dispersion relations
 $I_i\ni\kk\mapsto \omega_i(\kk)$, $i\in\cQ$, then we can determine 
the energy-momentum spectrum of $(H_\fr,P_\fr)$:
\begin{eqnarray}
\sp
(H_\fr,P_\fr)&=&\{(0,\0)\} \label{sphfrpfr}\\
&& \!\!\cup \ \big\{\big(\omega_{i_1}(\kk_1)+\cdots+\omega_{i_n}(\kk_n),\kk_1+\cdots+\kk_n\big)\ :
\ n=1,2,3\dots\}^\cl . \notag
\end{eqnarray}
Note that there is an obvious direct integral representation of the form (\ref{pw0}) and the relations (\ref{pw1}) and (\ref{pw2}) hold.

\subsection{Properties of the excitation spectrum of quasiparticle systems}

Let $(H,P)$ be a  quasiparticle  system (note we dropped the subscripts fr). The energy-momentum spectrum of such systems has special properties. First, we have 
\begin{eqnarray}
(0,\0)&\in&\sp(H,P),\label{endeq0-}
\end{eqnarray} 
because of the Fock vacuum state, which is a unique joint eigenstate of      $(H,P)$.  Moreover, we have a remarkable addition property
\begin{eqnarray}
\sp(H,P)&=&\sp(H,P)+\sp(H,P).\label{endeq1-}
\end{eqnarray} 

Assume now that
the Hamiltonian (\ref{quasa1}) is bounded from below, or what is equivalent,  assume that all the dispersion relations are non-negative. Then the Fock vacuum is a ground state satisfying $E=0$, so that  the excitation spectrum coincides with the energy-momentum spectrum. Thus we can rewrite (\ref{endeq0-}) and (\ref{endeq1-}) as
\begin{eqnarray}
(0,\0)&\in&\sp(H-E,P),\label{endeq0}\\
\sp(H-E,P)&=&\sp(H-E,P)+\sp(H-E,P).\label{endeq1}
\end{eqnarray}

Given (\ref{endeq0}), (\ref{endeq1}) is equivalent to
\begin{eqnarray}
\Exc & \supset &\Exc+\Exc.\label{endeq1.}
\end{eqnarray}

Another remarkable property holds true if in addition the number of particle species is finite. We have  then
\begin{eqnarray}
\Exc_\ess &=& \big(\Exc+\Exc\big)^\cl.\label{endeq4}
\end{eqnarray}
Indeed, using the continuity of the momentum spectrum, we easily see that only 1-quasiparticle states can belong to the discrete spectrum of the fiber Hamiltonians $H(\kk)$.

Before we proceed, let us introduce some terminology
concerning  real functions that will be useful in our study of quasiparticle-like spectra.
Recall that a function $\RR^d\ni\kk\mapsto\epsilon(\kk)$ is called {\em
  subadditive} 
if
$$\epsilon(\kk_1+\kk_2)\leq \epsilon(\kk_1)+\epsilon(\kk_2).$$

Let 
$\RR^d\supset I\ni\kk\mapsto\omega(\kk)$ be a given function.
Define
\begin{eqnarray}
\sah_\omega(\kk)&=&\inf
\{\omega(\kk_1)+\cdots+\omega(\kk_n)\ :\ \kk_1+\cdots+\kk_n=\kk,
\ n=1,2,3,\dots\}, \nonumber
\\
\sah_{\ess,\omega}(\kk)&=&\inf
\{\omega(\kk_1)+\cdots+\omega(\kk_n)\ :\ \kk_1+\cdots+\kk_n=\kk,
\ n=2,3,\dots\}, \nonumber
\end{eqnarray}
(By definition, the infimum of an empty set is $+\infty$).
 $\sah_\omega$ is known under the name of the {\em 
subadditive hull of $\omega$}. Equivalently, 
 $\sah_\omega$ is the biggest subadditive function less than $\omega$.

Note the relation
\begin{eqnarray}
\sah_\omega(\kk)&=&\min\{\omega(\kk),\sah_{\ess,\omega}(\kk)\}. \label{sahmin}
\end{eqnarray}


For $\kk\in\RR^d$, define
\begin{equation}
  \omega_{\min}(\kk):=\min\{\omega_i\ :\ i\in I(\kk)\}.\label{mino}\end{equation}
Recall the functions  $\exc$ and $\exc_\ess$  and the parameters
$\varepsilon$ and $c_\crit$ that we defined in Subsects 
\ref{Excitation spectrum} and \ref{Essential excitation spectrum}.

\begin{thm} \label{spectralpropquantusys}
Consider the quasiparticle  system given by (\ref{quasa1}) and (\ref{quasa2}) with non-negative dispersion relations. Assume the number of quasiparticle species is finite.
\begin{enumerate}
\item The bottom of the strict excitation spectrum is the subadditive hull of $\omega_{\min}$:
\begin{eqnarray}
\exc(\kk)&=&
\sah_{\omega_{\min}}(\kk),\ \kk\in\RR^d.\nonumber
\end{eqnarray}
\item The energy gap satisfies
\begin{eqnarray*}
\varepsilon&=&\inf_\kk\omega_{\min}(\kk).
\end{eqnarray*}
\item 
The bottom of the essential excitation spectrum satisfies
\begin{eqnarray} \exc_\ess(\kk)&=&\sah_{\ess,\omega_{\min}}(\kk),\ \ \kk\in\RR^d.\nonumber
\end{eqnarray}
\item The critical velocity satisfies
\begin{eqnarray*}
c_\crit &=&\inf_{\kk\neq\0}\frac{\omega_{\min}(\kk)}{|\kk|}
=\inf_{\kk\neq\0}\frac{\epsilon_\ess(\kk)}{|\kk|}.
\end{eqnarray*}
\end{enumerate}
\end{thm}

\begin{proof}
\textit{1.} is a direct consequence of the representation \eqref{sphfrpfr} and the fact that $(0,\0)$ is removed from the strict excitation spectrum.

Let us prove \textit{2}. Assume $\inf_\kk\omega_{\min}(\kk)=\omega_{\min}(\kk_0)>\varepsilon$. By \textit{1.} and the definition of $\varepsilon$ this implies that there exists a $\kk_1$ such that $\sah_{\omega_{\min}}(\kk_1)<\omega_{\min}(\kk_0)$. But since the dispersion relations are non-negative   this means there exists a $\kk_2$ (maybe equal $\kk_1$) such that $\omega_{\min}(\kk_2)\leq\omega_{\min}(\kk_0)$ which is a contradiction with   $\inf_\kk\omega_{\min}(\kk)=\omega_{\min}(\kk_0)$.

\textit{3.} follows from \textit{1.} and the fact that only 1-quasiparticle states belong to the discrete spectrum, as expressed by \eqref{endeq4}.

To prove \textit{4.} note that using \textit{1.} we have
\begin{eqnarray*}
c_\crit &:=&\inf_{\kk\neq\0}\frac{\exc(\kk)}{|\kk|}
=\inf_{\kk\neq\0}\frac{\sah_{\omega_{\min}}(\kk)}{|\kk|}.
\end{eqnarray*}
Using \eqref{sahmin} this can be rewritten as 
\begin{eqnarray*}
c_\crit &:=&\inf_{\kk\neq\0}\frac{\min \{ \omega_{\min}(\kk), \sah_{\ess,\omega_{\min}}(\kk)\}}{|\kk|}=\inf_{\kk\neq\0}\frac{\min \{ \omega_{\min}(\kk), \exc_\ess(\kk)\}}{|\kk|}
\end{eqnarray*}
where the last equality follows by \textit{3.} Obviously 
$$\inf_{\kk\neq\0}\frac{\omega_{\min}(\kk)}{|\kk|}
\geq\inf_{\kk\neq\0}\frac{\epsilon_\ess(\kk)}{|\kk|}.$$
On the other hand, a strict inequality would violate the assumption on non-negative dispersion relation. This ends the proof.
\end{proof}

Note that we assume that the momentum space is $\RR^d$. If we replace the momentum space
$\RR^d$ with $\frac{2\pi}{L}\ZZ^d$ (that is, if we put 
our system on a torus of side length $L$)  and we assume that all quasiparticles are bosonic, then
all statements of this
subsection generalize in an obvious way. However, because of the Pauli
principle, not all of them generalize  in
the  fermionic case.

\subsection{Quasiparticle-like quantum systems}
\label{Approximate versus exact quasiparticles}

One often considers quantum systems of the form
\begin{equation*}
H=H_\fr +V,\ \ \ P=P_\fr,
\end{equation*}
where $(H_\fr,P_\fr)$ is a quasiparticle system (defined in Subsection \ref{Concept of a quasiparticle}) and the perturbation $V$ is in some sense small. A description of physical systems in terms of approximate quasiparticles is very common in condensed matter physics. 

Clearly, there is a considerable freedom in choosing the splitting of $H$ into $H_\fr$ and $V$, and so  quasiparticles of this kind are only vaguely determined.
We will argue that in some cases  a different concept of quasiparticles is useful, which is rigorous and in a way much more interesting. This concept is expressed in the following definition.

Let $(H,P)$ be a translation invariant system on a Hilbert space $\cH$.
We will say that it is a {\em   quasiparticle-like system} if it is unitarily equivalent to a quasiparticle system.

\subsection{Asymptotic quasiparticles}

The above definition has one drawback. In practice we expect that the unitary equivalence mentioned in this definition is in some sense {\em natural} and constructed in the framework of {\em scattering theory}. 

Scattering theory is quite far from the main subject of this thesis, which is mostly concerned with purely spectral questions. However, since it has been mentioned and is related to the concept of a quasiparticle, let us give a brief discussion of this topic.

For a number of many-body systems the basic idea  of  scattering theory can be described as follows. Using the evolution $\e^{\ii tH}$ for $t\to\pm\infty$, we define  two  isometric operators
\begin{equation}
S^\pm:\lotimes_{i\in\cQ_\s}\Gamma_\s\left(L^2(I_i)\right)\otimes
\lotimes_{j\in\cQ_\a}\Gamma_\a\left(L^2(I_j)\right)\to\cH.
\label{paha}
\end{equation}
$S^\pm$ are called the {\em wave} or {\em M{\o}ller operators} and they satisfy
$$HS^\pm=S^\pm H_\fr,\ \ \ PS^\pm=S^\pm P_\fr,$$
where $(H_\fr,P_\fr)$ is  a quasiparticle quantum system. $S:=\left(S^{+}\right)^{\dagger}S^-$ is then called the {\em scattering operator}.

We will say that the system is {\em asymptotically complete} if the wave operators $S^\pm$ are unitary. Clearly, if a system is asymptotically complete, then it is  quasiparticle-like.

 There are at least two  classes of important physical system which possess a natural and rigorous scattering theory of this kind.

The first class
 consists of the 2nd quantization of Schr\"odinger many-body operators with  2-body  short range interactions  \cite{Der1998}. One can show that these systems are asymptotically complete
 (see \cite{Der1993} and references therein).
In this case the system is invariant with respect to the {\em Galileian group} and the dispersion relations have the form $\RR^d\ni\kk\mapsto E+\frac{\kk^2}{2m}$. 
Quasiparticles obtained in this context can be ``elementary'' -- in applications to physics these are typically electrons and nuclei -- as well as ``composite'' -- atoms, ions, molecules, etc.

Another important class of systems where the concept of asymptotic quasiparticles has a rigorous foundation belongs to (relativistic) quantum field theory, as axiomatized by the  Haag-Kastler or Wightman axioms. If we assume the existence of   discrete mass shells,  the so-called {\em Haag-Ruelle theory} allows us to construct  the wave operators, see e.g. \cite{Jost1965}. Note that in this case the system is covariant with respect to the {\em Poincar\'e group} and the dispersion relation has the form $\RR^d\ni \kk\mapsto\sqrt{m^2+\kk^2}$. 
 Here, quasiparticles are the usual stable particles.

Let us stress that both classes of systems  can be interacting in spite of the fact that they are equivalent to  free quasiparticle systems. In particular, their scattering operator can be nontrivial.

The above described  classes of quantum systems 
 are quite special. They are covariant with respect to rather large groups (Galilei or Poincar\'e) and have quite special dispersion relations.

\subsection{Quasiparticles in condensed matter physics}

The concept of a quasiparticle  is useful also in other  contexts, without the Galilei or Poincar\'{e} covariance.

An  interesting system which admits a quasiparticle interpretation is the free Fermi gas with a positive chemical potential. We describe this system in Subsection \ref{Non-interacting Fermi gas}. In this case the scattering theory is trivial:  $S^+=S^-$, and hence $S=\one$.

It seems that condensed matter physicists apply successfully the concept of a quasiparticle also to various interacting  translation invariant  systems.

One class of such systems seems to be the  Bose gas with repulsive interactions at zero temperature and positive density. In  this case, apparently, the system is  typically  well described by a free Bose gas of
 quasiparticles of (at least) two kinds: at low momenta we have {\em phonons} with an approximately linear dispersion relation, and at somewhat higher momenta  we have {\em rotons} (see, e.g., \cite{Griffin1993} and also Figure \ref{rotons} in Section \ref{mainresult}). This idea underlies the famous {\em Bogoliubov approximation} \cite{Bogoliubov1947,Bogoliubov1947a}. The phenomenon of
{\em superfluidity} can be to a large extent explained within this picture.
 The model of free asymptotic phonons seems to work well in real experiments \cite{CowWoo1971,Maris1977}. The Bogoliubov approximation will be discussed in greater details in the second and third part of this thesis.

Another class of strongly interacting systems  that seems to be
successfully  modelled by independent quasiparticles is the 
Fermi gas  with attractive interactions at zero temperature and positive chemical potential. 
By using the {\em Hartree-Fock-Bogoliubov (HFB) approach},
 which is closely related to the original {\em Bardeen-Cooper-Schrieffer (BCS)
approximation} \cite{BCS1957},
one obtains a simple model that can be used to explain
the {\em superconductivity} of the Fermi gas at very low temperatures. The corresponding quasiparticles are sometimes called {\em partiholes}. This will be explained in the second part of this work.

Note that the above two examples -- the interacting Bose and Fermi gas -- are in realistic circumstances neither Galilei nor Poincar\'{e} covariant. This  allows us to consider more general dispersion relations. However, we do not know whether these systems admit a quasiparticle interpretation or possess some kind of scattering theory.
 
Unfortunately, rigorous results in this direction are rather modest. One of such results is presented in the third part of this thesis.

\subsection{Quasiparticle-like excitation spectrum}
\label{Quasiparticle-like excitation spectrum}

The concept of a quasiparticle-like system, as defined in Subsection
\ref{Approximate versus exact quasiparticles},
is probably too strong for many applications. Let us propose a weaker property, which is  more likely to be satisfied in various situations. 

Again, our starting point is a translation invariant system described by its Hamiltonian and momentum $(H,P)$. Let us assume that $H$ is bounded from below, with $E$, as usual, denoting the ground state energy.
 
 We will say that the {\em excitation spectrum of  $(H,P)$ is quasiparticle-like} if it coincides with the excitation spectrum
 of a quasiparticle  system (see (\ref{quasa1}) and (\ref{quasa2})).

Clearly, the excitation spectrum of a quasiparticle-like system with a bounded from below Hamiltonian  is  quasiparticle-like. However, a system may have a quasiparticle-like excitation spectrum without being a quasiparticle-like system.

A quasiparticle-like excitation spectrum has special properties. In particular, it satisfies (\ref{endeq0}) and (\ref{endeq1}).

There exists a heuristic, but, we believe, a relatively convincing general argument why realistic translation invariant
quantum systems in thermodynamic limit  at zero temperature should satisfy 
 (\ref{endeq0}) and (\ref{endeq1}). We present it below. Note in particular that the infinite size of the quantum system plays an important role in this argument.

Consider a quantum gas in a  box of a very large side length $L$, described by $(H^L,P^L)$. For brevity, let us drop the superscript $L$. 
First of all, it seems reasonable to assume that the system possesses a translation invariant ground state, which we will denote by $\Phi$, so that $H\Phi=E\Phi$, $P\Phi=0$. Thus (\ref{endeq0}) holds.

Let $(E_{}+e_i,\kk_i)\in\sp(H,P)$, $i=1,2$. We can find eigenvectors with these
eigenvalues, that is, vectors $\Phi_i$ satisfying $H\Phi_i=(E_{}+e_i)\Phi_i$,
$P\Phi_i=\kk_i\Phi_i$.  Let us make the assumption that  it is possible to find
operators $A_i$ that are polynomials in creation and annihilation operator
smeared with functions well localized in  configuration space such that
$PA_i\approx A_i(P+\kk_i)$,  and which approximately 
create the vectors $\Phi_i$ from the
ground state, that is
 $\Phi_i\approx A_i\Phi_{}$.
 By replacing $\Phi_2$ with $\e^{\ii \y P}\Phi_2$ for
some $\y$ and $A_2$ with $\e^{\ii\y P}A_2 \e^{-\ii\y P}$, we can make sure that the
regions of localization of $A_1$ and $A_2$  are separated by a large distance. Note that here a large size of $L$ plays a role.

 Now consider the vector $\Phi_{12}:=
A_1A_2\Phi_{}$. Clearly, \[P\Phi_{12}\approx (\kk_1+\kk_2)\Phi_{12}.\]
$\Phi_{12}$ looks like the vector $\Phi_i$ in
 the region of localization of $A_i$, elsewhere it looks like $\Phi_{}$.
 The
Hamiltonian $H$  involves only
 expressions of short range (the potential decays in space). Therefore, 
we expect that
\[H\Phi_{12}\approx (E_{}+e_1+e_2)\Phi_{12}.\]
If this is the case, it
 implies
 that $(E_{}+e_1+e_2,\kk_1+\kk_2)\in\sp(H,P)$.
Thus  (\ref{endeq1}) holds.

\subsection{Bottom of a quasiparticle-like excitation spectrum}

Now suppose that $(H,P)$ is an arbitrary translation invariant  system with a bounded from below Hamiltonian. For simplicity, assume that its ground state energy is zero. We assume that we know its excitation spectrum $\sp(H,P)$.
There are two natural questions
\begin{enumerate}
\item Is $\sp(H,P)$ quasiparticle-like?
\item If it is the case, to what extent are its dispersion relations determined uniquely?
\end{enumerate}

In order to give partial answers to the above questions, recall 
the functions $\exc$ and $\exc_\ess$, as well as the sets $\Exc_\d$ and $\Exc_\ess$
that we defined in Subsections 
\ref{Excitation spectrum} and \ref{Essential excitation spectrum}. The following statements immediately result from the definitions of quasiparticle-like excitation spectrum and quasiparticle system. 

These theorems justify also the title of this subsection since they show the bottom of a quasiparticle-like excitation spectrum is the crucial object when giving (at least partial) answers to the above questions.

\begin{thm} \label{twbottom1}
Suppose that the excitation spectrum of $(H,P)$ is quasiparticle-like. Then the following is true:
\begin{enumerate}
\item 
 $\epsilon$ is subadditive.
\item
We can partly reconstruct some of the dispersion relations:
\begin{equation}\Exc_\d=\{(\omega_i(\kk),\kk)\ : \ i\in\cQ,\ \kk\in\RR^d\}\,\setminus\, \Exc_\ess
.\label{paio}\end{equation}
Consequently,
for $\kk$ satisfying $\epsilon(\kk)<\epsilon_\ess(\kk)$,
\[\epsilon(\kk)=\omega_{\min}(\kk),\]
where $\omega_{\min}$ was defined in (\ref{mino}).
\item If the number of quasiparticles species is finite, we can reconstruct 
$\epsilon_\ess$  from $\epsilon$:
\begin{eqnarray}
\epsilon_\ess(\kk)&=&\inf \{\epsilon(\kk_1)+
\epsilon(\kk_2)\ :\ \kk=\kk_1+\kk_2\}.\label{essop}
\end{eqnarray}
\end{enumerate}
\end{thm}

The existential part of the inverse problem has a partial solution:
\begin{thm}\label{twbottom2}
Suppose that  $\RR^d\ni\kk\mapsto\omega(\kk)$  be a given subbadditive function. Consider the translation invariant system
\[
 H_\fr=\int\omega(\kk)b_\kk^*b_\kk\d\kk,\ \ \ 
 P_\fr=\int \kk b_\kk^*
b_\kk \d\kk.\]
Then 
\begin{eqnarray*}\epsilon(\kk)&=&\omega(\kk),\\
\epsilon_\ess(\kk)&=&\inf \{\omega(\kk_1)+
\omega(\kk_2)\ :\ \kk=\kk_1+\kk_2\}.
\end{eqnarray*}
\end{thm}

The answer to the uniqueness part of the inverse problem is negative. The only situation where we can identify
dispersion relations from the spectral information involves
$\Exc_\d$, see (\ref{paio}).
The following example shows that we have quite a lot of freedom in choosing a dispersion relation giving a prescribed excitation spectrum.
For instance, all the Hamiltonians below have the same excitation spectrum and essential excitation spectrum  with  $\epsilon(\kk)=\epsilon_\ess(\kk)=|\kk|$:
\[H=\int_{|\kk|<c}|\kk|(1+d|\kk|^\alpha)b_\kk^*b_\kk\d\kk,\]
where $c>0$, $d\geq0$ and $\alpha>0$ are arbitrary.

\subsection{Translation invariant  systems with  two superselection sectors}
\label{translationtwosector}

We believe that it is relevant to introduce another concept. Suppose that a Hilbert space $\cH$ has a decomposition $\cH=\cH^+\oplus\cH^-$, which can be treated as a {\em superselection rule} (see e.g. \cite{Moretti2013}). This means that all observables decompose into direct sums. In particular, the Hamiltonian and momentum decompose as 
$(H,P)=(H^+,P^+)\oplus( H^-,P^-)$.
Clearly, 
\begin{eqnarray}
\sp(H,P)&=&\sp(H^+,P^+)\cup\sp(H^-,P^-).
\end{eqnarray}

 We will often assume that $H$ is bounded from below and possesses a translation invariant ground state $\Phi$ with energy $E$, which belongs to the sector $\cH^+$. The sector $\cH^+$ will be called {\em even}. The other sector $\cH^-$ will be called {\em odd}.

Under these assumptions we will call $\sp(H^+-E,P^+)$, resp.  $\sp(H^--E,P^-)$ the {\em even}, resp. {\em odd excitation spectrum}.
We introduce also the {\em strict even excitation spectrum}:
\begin{equation}\Exc^+:=\sp\Bigl((H^+-E,P^+)\Big|_{\{\Phi\}^\perp}\Bigr)\label{exci-2}
\end{equation}
The  {\em strict odd excitation spectrum} will coincide with
the full odd excitation spectrum:
\begin{equation}\Exc^-:=\sp(H^--E,P^-).\label{exci-3}
\end{equation}

Finally, we define the {\em even} and {\em odd essential  excitation spectrum}
$\Exc_\ess^\pm$ just as in Subsect. 
\ref{Essential excitation spectrum}, except that we replace $(H,P)$ with $(H^\pm,P^\pm)$.

We introduce also a special notation for the bottom of the sets $\Exc^\pm$ and $\Exc_\ess^\pm$:
\begin{eqnarray*}
\exc^\pm(\kk)&:=&\inf\{e\ :\ (e,\kk)\in\Exc^\pm\},\\
\exc_\ess^\pm(\kk)&:=&\inf\{e\ :\ (e,\kk)\in\Exc_\ess^\pm\}.
\end{eqnarray*}

We have,
\begin{eqnarray}
\sp(H-E,P)&=&\sp(H^+-E,P^+)\cup\sp(H^--E,P^-),\\
\Exc&=&\Exc^+\cup\Exc^-,\\
\Exc_\ess&=&\Exc_\ess^+\cup\Exc_\ess^-,\\
\exc(\kk)&=&\min\{\exc^-(\kk),\exc^+(\kk)\},\\
\exc_\ess(\kk)&=&\min\{\exc_\ess^-(\kk),\exc_\ess^+(\kk)\}
.\end{eqnarray}

\subsection{Quasiparticle  systems with the fermionic  superselection rule}

Consider a quasiparticle  system $(H_\fr,P_\fr)$ on the Fock space (\ref{fock}).
Define the {\em fermionic number operator} as
\[ N_\a=\sum_{i\in\cQ_\a}b_i^*(\kk)b_i(\kk).\]
The {\em fermionic parity}  $(-1)^{N_\a}$ provides a natural superselection rule.
If $\cH=\cH^+\oplus\cH^-$ denotes the corresponding direct sum decomposition, then the  Hamiltonian and momentum  decompose as
 \begin{equation}
(H_\fr,P_\fr)=(H_\fr^+,P_\fr^+)\oplus(H_\fr^-,P_\fr^-).\label{qrq}\end{equation}
(\ref{qrq}) will be called a {\em two-sector quasiparticle  system}.

If we know the dispersion relations
 $I_i\ni\kk\mapsto \omega_i(\kk)$, $i\in\cQ$, then we can determine 
the even and odd energy momentum spectrum of $(H_\fr^+,P_\fr^+)$:
\begin{eqnarray*}
\sp
(H_\fr^+,P_\fr^+)&=&\{(0,\0)\}\\
&&\!\!\cup\ \bigl\{\bigl(\omega_{i_1}(\kk_1)+\cdots+\omega_{i_n}(\kk_n),\kk_1+\cdots+\kk_n\bigr)\ :\\
&&\hskip 5ex
  \hbox{even number of fermions}, n=1,2,3\dots\}^\cl,\\
  \\
\sp
(H_\fr^-,P_\fr^-)&=& \bigl\{\bigl(\omega_{i_1}(\kk_1)+\cdots+\omega_{i_n}(\kk_n),\kk_1+\cdots+\kk_n\bigr)\ :\\
&&\hskip 5ex
  \hbox{odd number of fermions}, n=1,2,3\dots\}^\cl
.\end{eqnarray*}

\subsection{Properties of the excitation spectrum of two-sector quasiparticle systems}\label{twosectorproperties}

Let $(H,P_\fr)=(H^+,P^+)\oplus( H^-,P^-)$ be a  two-sector quasiparticle system.
Clearly, we have
\begin{eqnarray}
(0,\0)&\in&\sp(H^+,P^+)\label{pas1a0-}
\end{eqnarray}
because of the Fock vacuum. Here are the  properties of the even and odd excitation spectrum:
\begin{eqnarray}
\sp(H^+,P^+)&=&\sp(H^+,P^+)+\sp(H^+,P^+)\label{pas1a0}
\\
&\supset&
\sp(H^-,P^-)+\sp(H^-,P^-),\label{pas20}\\
\sp(H^-,P^-)&=&\sp(H^-,P^-)+\sp(H^+,P^+).\label{pas30}
\end{eqnarray}

Assume now that
the Hamiltonian is bounded from below.
 Then the Fock vacuum is a translation invariant ground state satisfying $E=0$, so that  the excitation spectrum coincides with the energy-momentum spectrum.
Thus we can rewrite (\ref{pas1a0-})-(\ref{pas30}) as
\begin{eqnarray}
(0,\0)&\in&\sp(H^+-E,P^+)\label{pas1a00},\\
\sp(H^+-E,P^+)&=&\sp(H^+-E,P^+)+\sp(H^+-E,P^+)\label{pas1a00.}
\\
&\supset&
\sp(H^--E,P^-)+\sp(H^--E,P^-),\label{pas200}\\
\sp(H^--E,P^-)&=&\sp(H^--E,P^-)+\sp(H^+-E,P^+).\label{pas300}
\end{eqnarray}

Given (\ref{pas1a00}), (\ref{pas1a00.})-(\ref{pas300}) are equivalent to
\begin{eqnarray}
\Exc^+&\supset&\bigl(\Exc^++\Exc^+\bigr)\cup\bigl(\Exc^-+\Exc^-\bigr),
\label{pas1a00a}
\\
\Exc^-&\supset&\Exc^-+\Exc^+.\label{pas300.}
\end{eqnarray}

If in addition the number of particle species is finite, then
\begin{eqnarray}
\Exc_\ess^+&=&\bigl(\Exc^++\Exc^+\bigr)^\cl\cup \bigl(\Exc^-+\Exc^-\bigr)^\cl,
\label{pas1a000}
\\
\Exc_\ess^-&=&\bigl(\Exc^-+\Exc^+\bigr)^\cl.\label{pas3000}
\end{eqnarray}

\subsection{Two-sector quasiparticle-like spectrum}

Consider now an arbitrary translation invariant  system with two superselection sectors
$(H,P)=(H^+,P^+)\oplus( H^-,P^-)$. We will assume that $H$ is bounded from below and the ground state with energy $E$ is translation invariant and belongs to the sector $\cH^+$.

We will say that the excitation spectrum of the system $(H^+,P^+)\oplus( H^-,P^-)$ {\em is two-sector quasiparticle-like} if it coincides with the excitation spectrum of a two-sector quasiparticle  system. Such an excitation spectrum has special properties. In particular, it satisfies (\ref{pas1a00})-(\ref{pas300}).

There exists a heuristic 
general argument why realistic translation invariant
quantum systems in thermodynamic limit should satisfy
 (\ref{pas1a00})-(\ref{pas300}).
 It is an obvious modification of the argument given in Subsect.
\ref{Quasiparticle-like excitation spectrum}. 

 Indeed, we need to notice what follows.
 $(-1)^{N_\a}$ is always a superselection rule for  realistic quantum system.
In particular, if we assume that the ground state is nondegenerate, it has to be either bosonic or fermionic. We make an assumption that it is bosonic.

The eigenvectors $\Phi_1$ and $\Phi_2$, discussed in
Subsect.
\ref{Quasiparticle-like excitation spectrum},
 can be chosen to be purely bosonic or fermionic. Using the fact that the ground state is purely bosonic, we see that we can chose the operators $A_1$ and $A_2$ 
to be purely bosonic or fermionic. (That means, they either commute or anticommute with $(-1)^{N_\a}$). Consequently, we have the following possibilities:
\begin{itemize}
\item Both $\Phi_1$ and $\Phi_2$ are bosonic. Then $\Phi_{12}$ is bosonic.
\item Both $\Phi_1$ and $\Phi_2$ are fermionic. Then $\Phi_{12}$ is bosonic.
\item One of $\Phi_1$ and $\Phi_2$  is bosonic, the other is fermionic. Then $\Phi_{12}$ is fermionic. \end{itemize}
This implies (\ref{pas1a00a}) and
(\ref{pas300.}).

\subsection{Bottom of a two-sector quasiparticle-like excitation spectrum}

Suppose again that 
$(H,P)=(H^+,P^+)\oplus(H^-,P^-)$ is a translation invariant 
 system with two superselection sectors.
 We assume that we know  its excitation spectrum.
We would like to describe some criteria to verify whether
it is two-sector quasiparticle-like. 

These criteria will involve the properties of the bottom of the even and odd excitation spectrum. The following theorem follows directly from the properties described in Subsection \ref{translationtwosector} and the definition two-sector quasiparticle-like quantum system. It is in some sense analogous to Theorems \ref{twbottom1} and \ref{twbottom2}.

\begin{thm} Suppose that the excitation spectrum of $(H^+,P^+)\oplus( H^-,P^-)$ is two-sector quasiparticle-like.
\begin{enumerate}
\item We have the following subadditivity properties:
\begin{eqnarray*}
\exc^-(\kk_1+\kk_2)\leq
\exc^-(\kk_1)+\exc^+(\kk_2),\\
\exc^+(\kk_1+\kk_2)\leq
\exc^-(\kk_1)+\exc^-(\kk_2),\\
\exc^+(\kk_1+\kk_2)\leq
\exc^+(\kk_1)+\exc^+(\kk_2).
\end{eqnarray*}
\item If the number of species of quasiparticles is finite, then we can reconstruct 
$\exc_\ess^-$ and $\exc_\ess^+$ from $\exc^-$ and
$\exc^+$:
\begin{eqnarray*}
\exc_\ess^-(\kk)&=&\inf \{\exc^-(\kk_1)+
\exc^+(\kk_2)\ :\ \kk=\kk_1+\kk_2\},\\
\exc_\ess^+(\kk)&=&\inf \{\exc^+(\kk_1)+
\exc^+(\kk_2),\ \ \ \exc^-(\kk_1)+
\exc^-(\kk_2)\ :\ \kk=\kk_1+\kk_2\}.
\end{eqnarray*}\end{enumerate}
\end{thm}

\subsection{Non-interacting Fermi gas}
\label{Non-interacting Fermi gas}
As an example of the introduced concepts, let us give a brief discussion of the free Fermi gas with chemical potential $\mu$ in $d$
dimensions. For simplicity, we will
assume that particles have no internal degrees of freedom such as spin.

The Hilbert space of $N$ fermions equals $\Gamma_\a^N\left(L^2(\RR^d)\right)$ 
(antisymmetric square integrable functions on $(\RR^d)^N$).
Let $\Delta_{(i)}$ denote  the Laplacian $\Delta$ acting on
the $i$th variable. Then the  Hamiltonian  equals
\begin{equation}
H_{N}=
\sum_{ i=1}^N(-\Delta_{(i)}-\mu).
\label{sch4}\end{equation}
It commutes with the  momentum operator 
\[P_{N}:=\sum_{i=1}^N -\ii\nabla_{(i)}.\]

It is convenient to put together various $N$-particle sectors in a single Fock
space 
 \[\Gamma_\a(L^2(\RR^d)):=\mathop{\oplus}\limits_{N=0}^\infty 
\Gamma_\a^N\left(L^2(\RR^d)\right).\]
Then the basic observables are the
 Hamiltonian, the total momentum  and the particle number operator:
\begin{eqnarray}\nonumber
H&=&\loplus_{N=0}^\infty H_N=\int a_\x^{\dagger}(-\Delta-\mu)a_\x\d\x,\\
 P&=&\loplus_{N=0}^\infty P_N=-\ii\int a_\x^{\dagger}\nabla_\x a_\x\d\x
,\label{postu1}\\
 N&=&\loplus_{N=0}^\infty N=\int a_\x^{\dagger} a_\x\d\x
,\nonumber
\end{eqnarray}
where
 $a_\x^{\dagger}$/ $a_\x$ are the usual fermionic creation/annihilation operators.

The three operators in (\ref{postu1}) describe only a finite number of particles in an
infinite space. We would like to
investigate  homogeneous Fermi gas at a positive density in the thermodynamic limit.
Following the accepted, although somewhat unphysical tradition, 
we first consider our system on
$\Lambda=[-L/2,L/2]^d$, the $d$-dimensional
 cubic box of side length $L$, with periodic boundary conditions.
 Note that the spectrum of the momentum becomes
$ \frac{2\pi }{L}\ZZ^d$. At the end we let
 $L\to\infty$. The Fock space is now $\Gamma_\a(L^2(\Lambda))$.

It is convenient to pass to the momentum representation:
\begin{eqnarray}\label{2B}
H_L 
&=&\sum_{\kk}(\kk^2-\mu) a^{\dagger}_{\kk}
a_{\kk}\nonumber 
\\
P_L&=&\sum_{\kk}\kk a^{\dagger}_{\kk}a_\kk,\label{2C}\\
N_L&=&\sum_{\kk}a^{\dagger}_{\kk}a_\kk,\nonumber\end{eqnarray}
where we used (\ref{postu1})
and $a_\x=L^{-d/2}\sum_\kk \e^{\ii \kk\x}a_\kk$. We sum over
$\kk\in \frac{2\pi }{L}\ZZ^d$.

It is natural to change  the representation of canonical anticommutation
relations
and replace the usual fermionic creation/annihilation operators
by new ones, which
 kill the ground state of the Hamiltonian:
\begin{eqnarray*}
b_\kk^{\dagger}:&=&a_\kk^{\dagger},\ b_\kk:=a_\kk,\ \kk^2>\mu,\\
b_\kk^{\dagger}:&=&a_\kk,\ b_\kk:=a_\kk^{\dagger},\ \kk^2\leq\mu.
\end{eqnarray*}
Then,
\begin{eqnarray*}
 H_L&=&\sum_\kk |\kk^2-\mu|b_\kk^{\dagger}b_\kk+E_L,\\
 P_L&=&\sum_\kk \kk b_\kk^{\dagger}
b_\kk ,\\
 N_L&=&\sum_\kk \sgn(\kk^2-\mu) b_\kk^{\dagger}
b_\kk +C_L
,\label{ham2}\end{eqnarray*}
where 
\begin{eqnarray*}
E_L&=&\sum_{\kk^2\leq\mu} (\kk^2-\mu),\\
C_L&=&\sum_{\kk^2\leq\mu} 1.
\end{eqnarray*}
It is customary to drop the constants $E_L$ and $C_L$.

Set $\omega(\kk)=|\kk^2-\mu|$.
In the case of an infinite space, the above analysis suggests that it is natural to postulate
\begin{eqnarray}
 H&=&\int\omega(\kk)b_\kk^{\dagger}b_\kk\d\kk,\label{ham3}\\
 P&=&\int \kk b_\kk^{\dagger}
b_\kk \d\kk,\label{ham3a}\\
 N&=&\int \sgn(\kk^2-\mu) b_\kk^{\dagger}
b_\kk \d\kk,
\label{ham4}\end{eqnarray}
 as the Hamiltonian, total momentum and number operator of the free Fermi gas from the beginning, instead of (\ref{postu1}).

The operators $b_\kk^{\dagger}/b_\kk$ can be called quasiparticle creation/annihilation operators and the function $\kk\mapsto\omega(\kk)$ the quasiparticle dispersion relation.  Thus  a quasiparticle is a true particle above the Fermi level and a  hole below the Fermi level.
\\
\\
\\
\\


\subsection{Energy-momentum spectrum of non-interacting Fermi gas}\label{noninteractinggraphics}
The analysis in the previous subsection implies that the energy-momentum spectrum of a non-interacting Fermi gas is described by 
(\ref{ham3}) and (\ref{ham3a}) with the
dispersion relation $\omega(\kk)=|\kk^2-\mu|$. Below we present present
diagrams representing the energy-momentum spectrum. 

In the  full and the odd cases, that is $\sp (H,P)$ and $\sp(H^-,P^-)$, the dispersion relation
$\omega$ is a singular part of the spectrum and it 
is be denoted by a solid line. In the even case, $\sp(H^+,P^+)$, the dispersion relation is denoted by a dotted line.

\begin{figure}[!h]
\centering
\includegraphics{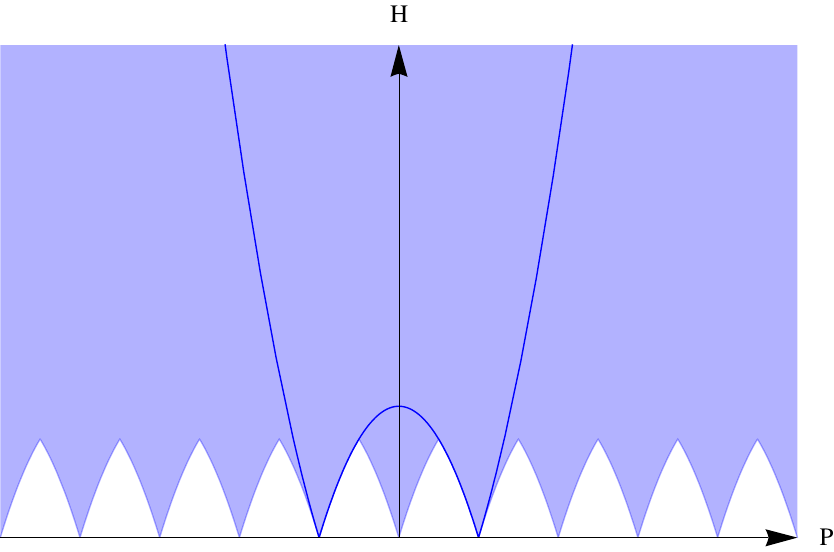}
\label{fig:1}
\caption{$\sp(H,P)$ in the non-interacting case, $d=1$.}
\end{figure}

\begin{figure}[ht]
\centering
\includegraphics{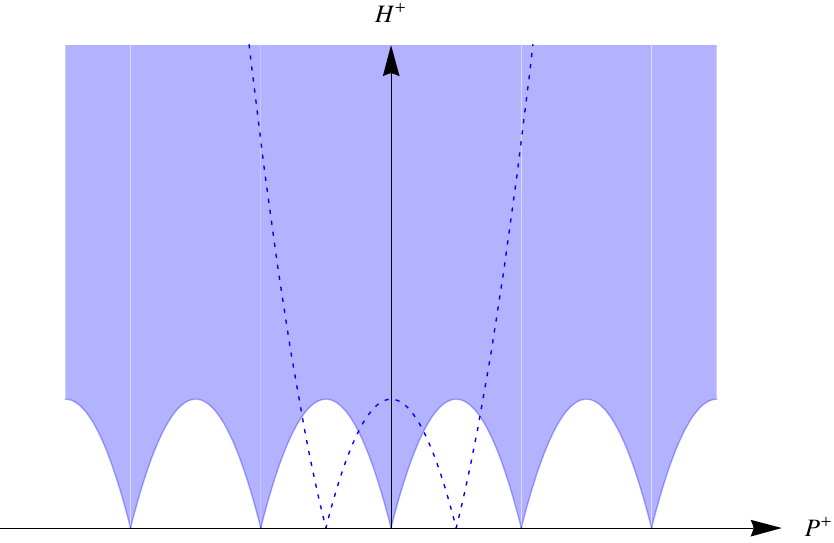}
\label{fig:2}
\caption{$\sp(H^+,P^+)$ in the non-interacting case, $d=1$.}
\end{figure}

\begin{figure}[ht]
\centering
\includegraphics{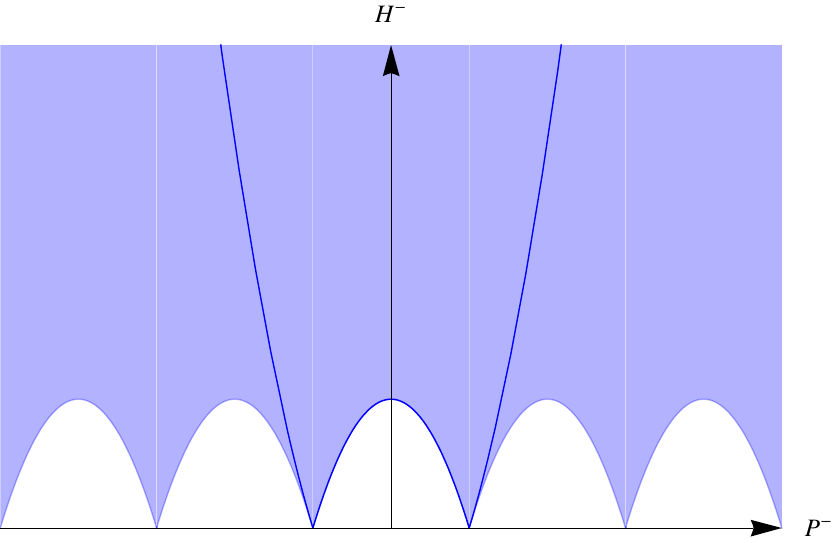}
\label{fig:3}
\caption{$\sp(H^-,P^-)$ in the non-interacting case, $d=1$.}
\end{figure}
\newpage
For $d\geq2$ the energy-momentum spectrum is rather boring:
\begin{figure}[!h]
\centering
\includegraphics{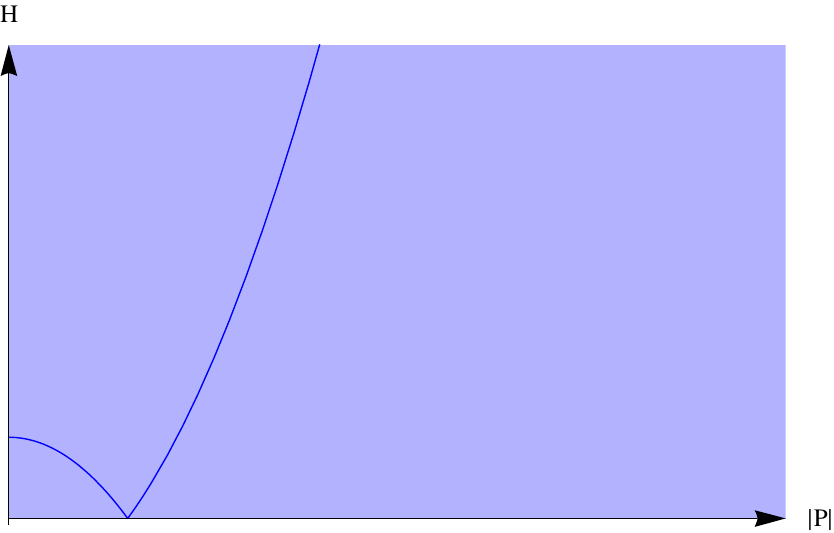}
\label{fig:4}
\caption{$\sp(H,P)$, $\sp(H^+,P^+)$, $\sp(H^-,P^-)$ in the \\ non-interacting case, $d\geq2$.}
\end{figure}

In the next parts of the thesis we shall also present diagrams representing the energy-momentum spectrum of interacting Bose and Fermi gases, which one can obtain using the Bogoliubov (in the bosonic case) and Hartree--Fock--Bogoliubov (in the fermionic one) approximations. 


\part{Calculation of the excitation spectrum - approximate methods}

\section[Energy-momentum spectrum of a homogeneous Bose gas]{Energy-momentum spectrum of a homogeneous Bose gas}\label{bosegas}
In this section we shall consider a homogeneous Bose gas. After defining the model, we will present the so-called \textit{Bogoliubov approximation}. This approximation leads to a quantitative description of the low-lying energy-momentum spectrum. This section is based on \cite{Bogoliubov1947} and \cite{CDZ2009}.

\subsection[The model]{The model} \label{themodelbose}
 \subsubsection{Bose gas in canonical approach}
Consider a system of $N$ bosons interacting via a 2-body potential. In the canonical approach one assumes that the number of particles is fixed. 

Suppose that the 2-body potential of an interacting Bose gas is described by real function
$\mathbb{R}^{d}\ni \x\mapsto v
(\x)$, with its Fourier transform defined  by
\[\hat{v}(\p):=\int_{\mathbb{R}^{d}}v(\x)\e^{-\ii\p\x} \d \x.\]
 We assume that $v(\x)=v(-\x)$, and that $v$ decays sufficiently fast at infinity. 
 
 We also suppose  that the Fourier transform of the potential is positive, i.e.
$$\hat{v}(\p)\geq 0, \,\,\, \p \in \mathbb{R}^{d}.$$
Such potentials are sometimes called \textit{repsulsive}.

A homogeneous Bose gas is described by the Hilbert space  $L^{2}_{\s}((\RR^d)^{N})$ (symmetric
square integrable functions on $(\RR^d)^{N})$ , the $N$-body Schr\"odinger Hamiltonian
\begin{equation}
H_{N}=-\sum_{i=1}^{N}\Delta_{i}+\sum_{1\leq i<j\leq N} v(\x_{i}-\x_{j}), \label{hamiltonianbosepart2}
\end{equation}
and the momentum operator
\begin{equation}
P_N=\sum_{i=1}^{N}-\ii\nabla_{i}. \label{momentum}
\end{equation}
Clearly \eqref{hamiltonianbosepart2} and \eqref{momentum} commute, thus they define a \textit{translation invariant quantum system} (recall Subsection \ref{transinvquansys}).

As in the case of the Fermi gas, we
want to investigate homogeneous Bose gas at \textit{positive} density. Therefore, we will consider the Bose gas on a torus, that is in the box $\Lambda=[-L/2,L/2]^{d}$ with periodic boundary conditions. 

The original potential $v$ is then replaced by its 
{\em periodized} version
\[v^{L}(\x):=\frac{1}{L^d}\sum_{\p\in(2\pi/L)\mathbb{Z}^{d}}\e^{\ii\p\x}\hat{v}(\p).\]
Here, $\p\in(2\pi/L)\mathbb{Z}^{d}$ is the discrete momentum variable.
 Note that $v^{L}$ is periodic with respect to the domain $\Lambda$ and that $v^{L}(\x)\rightarrow v(\x)$ as $L\to \infty$. 
 
 The homogeneous Bose gas is thus described by the Hamiltonian 
\begin{equation}
H_N^L=-\sum_{i=1}^{N}\Delta^{L}_{i}+\sum_{1\leq i<j\leq N} v^{L}(\x_{i}-\x_{j}) \label{hamiltonianboseboxpart2}
\end{equation}
acting on the space $L^{2}_{\s}(\Lambda^{N})$ (the symmetric subspace
of  $L^{2}(\Lambda^{N})$). The Laplacian is assumed to have periodic boundary conditions.

Accordingly, the momentum operator is given by
\begin{equation}
P_N^L=\sum_{i=1}^{N}-\ii\nabla^L_{i}, \label{momentumbox}
\end{equation}
where $\nabla^L_{i}$ denotes the gradient acting on the $i$-th variable, with periodic boundary conditions. 

\subsubsection{Grand-canonical Hamiltonian of the Bose gas}
\label{grandcanonicalbose}
Instead of studying the Bose gas in the canonical
formalism, fixing the number of particles, it is mathematically more convenient to use the grand-canonical formalism and fix the chemical potential $\mu$. One can pass from the chemical potential to the density by the Legendre transformation.

In the grand-canonical approach one allows the system to have an arbitrary number of particles. Thus, it is convenient to put all $N$-particle spaces into a single bosonic Fock space
$$\Gamma_\mathrm{s}(L^2(\Lambda)):=\mathop{\oplus}\limits_{N=0}^\infty 
\Gamma_\mathrm{s}^N\left(L^2(\Lambda)\right)$$
with the Hamiltonian 
\begin{eqnarray*}
H_\mu^L&:=&\mathop{\oplus}\limits_{N=0}^\infty (H^{L}_N -\mu N)\\
&=&\int a_{\x}^{\dagger} (-\Delta_{\x}^L-\mu) a_{\x}\d
\x 
+\frac12\int\int a_{\x}^{\dagger} a_{\y}^{\dagger} 
v^L(\x-\y)a_{\y} a_{\x}\d
\x\d\y, 
\end{eqnarray*}
where $a_{\x}$, $a_{\x}^{\dagger}$ are the usual bosonic annihilation and creation operators. The second quantized momentum and number operators are defined as
\begin{eqnarray*}
P^L&:=&
\mathop{\oplus}\limits_{N=0}^\infty P^{L}_N=-\ii\int
a_{\x}^{\dagger} \nabla_\x^La_{\x} \d\x,\\
N^L&:
=&\mathop{\oplus}\limits_{N=0}^\infty N=\int a_{\x}^{\dagger} a_{\x}\d\x
.\end{eqnarray*}

Due to periodic boundary conditions it is convenient to pass to momentum representation. Then
\begin{eqnarray}
H_\mu^L 
&=&\sum_{\kk}(\frac12 \kk^2-\mu) a^{\dagger} _{\kk}
a_{\kk}\nonumber \\
&&+\frac{1}{2L^d}\sum_{\kk_1,\kk_2,\kk_3,\kk_4}\delta(\kk_1+\kk_2-\kk_3-\kk_4)\hat{v}(\kk_2-\kk_3)
a^{\dagger}_{\kk_1}a^{\dagger}_{\kk_2}a_{\kk_3}a_{\kk_4}
,\label{hamboseboxpedpart2}
\\
P^L&=&\sum_{\kk}\kk a^{\dagger} _{\kk}a_\kk,\nonumber
\\
N^L&=&\sum_{\kk}a^{\dagger} _{\kk}a_\kk.\nonumber
\end{eqnarray}
Here $a_{\x}=L^{-d/2}\sum_{\kk}\e^{\ii\kk\x}a_{\kk}$ and the summation over momenta runs over $\kk \in \frac{2\pi}{L}\mathbb{Z}^d$.

\subsection{The Bogoliubov approximation}
In this subsection we would like to introduce the famous Bogoliubov approximation. This scheme was introduced by Bogoliubov in 1947 in his famous paper "On the Theory of Superfluidity" \cite{Bogoliubov1947}. His goal was to derive a microscopic theory of superfluidity. In particular, Bogoliubov wanted to show that the homogeneous Bose gas meets Landau's criterion for superfluidity. This criterion says that a system can behave in a superfluid manner only if its low-lying excited states depend linearly on the total momentum of the system \cite{Landau1941}.

Below we will describe the original ideas of Bogoliubov. Later, based on \cite{CDZ2009}, we would like to present a modification of that scheme. This improved method fits into a more general scheme which will be described in Section \ref{beliaevthmsec}.

\subsubsection{The original Bogoliubov approximation}\label{The original Bogoliubov approximation}
In his original approach \cite{Bogoliubov1947,BogLectures1}, Bogoliubov considered the canonical setup with the Hamiltonian
\begin{eqnarray}
H_0^L 
&=&\sum_{\kk}\frac12 \kk^2 a^{\dagger} _{\kk}
a_{\kk}\nonumber \\
&&+\frac{1}{2L^d}\sum_{\kk_1,\kk_2,\kk_3,\kk_4}\delta(\kk_1+\kk_2-\kk_3-\kk_4)\hat{v}(\kk_2-\kk_3)
a^{\dagger}_{\kk_1}a^{\dagger}_{\kk_2}a_{\kk_3}a_{\kk_4}.
\label{hambogoriginal}
\end{eqnarray}
From now on we will drop the subscript $0$ and superscript $L$.

Bogoliubov starts with the observation that the operators $N_{0}=a^{\dagger}_0 a_0$ and $N_{0}+1=a_0 a^{\dagger}_0$ enter the Hamiltonian \eqref{hambogoriginal} only via the ratios 
$$\frac{a^{\dagger}_0 a_0}{L^{d}},\,\,\,\frac{a_0 a^{\dagger}_0}{L^{d}}.$$
The difference of these ratios equals $1/L^d$, which goes to zero when taking the \textit{thermodynamic limit}, that is when 
\begin{eqnarray}
N\rightarrow \infty, \quad L\rightarrow \infty \quad \,\,\, \text{with} \,\,\, \quad \frac{N}{L^d}=\rho=\rm{const.}\label{thermlimit}
\end{eqnarray}
Thus, according to Bogoliubov, one can neglect the non-commutativity of the operators $a^{\dagger}_0$ and $a_0$ when deriving an approximate expression for the Hamiltonian \eqref{hambogoriginal}, and replace them by (complex) numbers. This is the so-called  \textit{$c$-number substitution} \cite{LSY05}.

In the next step Bogoliubov assumes the considered system is \textit{weakly interacting}. With this assumption one can expect that  the ground state of this interacting is not too different from the ground state of the non-interacting system, i.e. most particles have zero momentum. Thus, when looking at the low-lying states, the quantity 
$$\sum\nolimits '\frac{N_{\p}}{N}=\frac{N-N_0}{N}$$
is small. Here, the prime next to sum symbol means we are summing over momenta $\p\in(2\pi/L)\mathbb{Z}^{d}\setminus\{\0\}.$ Since $N_{\p}=a^\dagger_\p a_\p$ we conclude that for $\p\neq \0$ the operators (or rather operator amplitudes) $a_\p$ and $a^\dagger_\p$ will be small compared with $a_\0$ and $a^\dagger_\0$ respectively. Thus, in the approximate expression for   \eqref{hambogoriginal} Bogoliubov neglects the terms with more than two $a_\p$ and $a^\dagger_\p$ ($\p\neq \0$).

With these assumptions one arrives at the following expression which is an approximation for  \eqref{hambogoriginal}:
\begin{eqnarray}
H&=&\sum\nolimits '\frac{\p^2}{2}a^\dagger_\p a_\p+\hat{v}(\0)\left(\frac{N_\0^2}{2L^d}+\frac{N_\0}{L^d}\sum\nolimits 'a^\dagger_\p a_\p\right)+\frac{N_\0}{L^d}\sum\nolimits '\hat{v}(\p)a^\dagger_\p a_\p \nonumber \\
&+& \frac{a_\0^2}{2L^d}\sum\nolimits '\hat{v}(\p)a^\dagger_\p a^\dagger_{-\p} +\frac{a_\0^{\dagger,2}}{2L^d}\sum\nolimits '\hat{v}(\p)a_\p a_{-\p}. \label{firstbogapprox}
\end{eqnarray}
(For simplicity, we will keep denoting the Hamiltonians after successive approximations also by $H$). Since
$$\sum\nolimits ' N_\p=N-N_\0,$$
then in the considered approximation
$$\frac{N_\0^2}{2L^d}+\frac{N_\0}{L^d}\sum\nolimits ' a^\dagger_\p a_\p=\frac{N_\0^2+2N_\0(N-N_\0)}{2L^d} \approx \frac{N^2}{2L^d}$$
up to terms with four creation or annihilation operators with non-zero momentum which, due to our assumptions, can be dropped.

Therefore, using \eqref{firstbogapprox}, we obtain
\begin{align}
H=\hat{v}(\0)\frac{N^2}{2L^d}+\sum\nolimits '\left(\frac{\p^2}{2}+\frac{N_\0}{L^d}\hat{v}(\p)\right)a^\dagger_\p a_\p +\sum\nolimits '\hat{v}(\p)\left(\frac{a_\0^2}{2L^d}a^\dagger_\p a^\dagger_{-\p} +\frac{a_\0^{\dagger,2}}{2L^d}a_\p a_{-\p}\right). \label{secondbogapprox}
\end{align}
Let us now introduce the operators ($\p\neq \0$)
\begin{eqnarray}
b_{\p}=a_\0^{\dagger}N_\0^{-1/2}a_\p, \,\,\,b^{\dagger}_{\p}=a_\0 N_\0^{-1/2}a^{\dagger}_\p.  \label{boperators}
\end{eqnarray}
Note that
$$[b_{\p},b^{\dagger}_{\q}]=\delta_{\p,\q}-\frac{a_\p^{\dagger}a_\q}{N_\0}$$
and thus the assumption about the macroscopic occupation of the zero momentum mode yields in the thermodynamic limit bosonic commutation relations for the operators $b_\p$ and $b_\p^{\dagger}$.

Using \eqref{boperators} we can rewrite \eqref{secondbogapprox} as
\begin{eqnarray}
H=\hat{v}(\0)\frac{N^2}{2L^d}+\sum\nolimits '\left(\frac{\p^2}{2}+\frac{N_\0}{L^d}\hat{v}(\p)\right)b^\dagger_\p b_\p +\sum\nolimits '\hat{v}(\p)\left(\frac{N_\0}{2L^d}b^\dagger_\p b^\dagger_{-\p} +\frac{N_\0}{2L^d}b_\p b_{-\p}\right).\quad \label{thirdbogapprox}
\end{eqnarray}
This expression can be looked upon as a quadratic Hamiltonian in terms of the operators $b_\p$ and $b_\p^{\dagger}$ and it can be diagonalized by introducing new operators
\begin{eqnarray}
\xi_\p=\frac{b_\p-A_\p b^\dagger_{-\p}}{\sqrt{1-A_\p^2}}, \quad \xi^{\dagger}_\p=\frac{b^{\dagger}_\p-A_\p b_{-\p}}{\sqrt{1-A_\p^2}}, \label{diagbogoperators}
\end{eqnarray}
where
\begin{eqnarray}
A_\p=\frac{L^{d}}{N_\0\hat{v}(\p)}\left(\tilde{E}(\p)-\frac{\p^2}{2}-\frac{N_\0}{L^d}\hat{v}(\p)\right)\quad \text{and} \quad \tilde{E}(\p)=\sqrt{\frac{N_\0}{L^d}\p^2\hat{v}(\p)+\frac{\p^4}{4}}. \label{diagbogopcoeff}
\end{eqnarray}
The operators $\xi^{\dagger}_\p$ and $\xi_\p$ also satisfy the (approximate) canonical commutation relations. Inverting \eqref{diagbogoperators} we have 
\begin{eqnarray}
b_\p=\frac{\xi_\p+A_\p \xi^\dagger_{-\p}}{\sqrt{1-A_\p^2}}, \quad b^{\dagger}_\p=\frac{\xi^{\dagger}_\p+A_\p \xi_{-\p}}{\sqrt{1-A_\p^2}}. \label{diagbogoperatorsinverted}
\end{eqnarray}
Substituting these expressions in the Hamiltonian \eqref{thirdbogapprox} we obtain 
\begin{eqnarray*}
H=\frac{N^2}{2L^d}\hat{v}(\0)+\frac12 \sum\nolimits '\left(\tilde{E}(\p)-\frac{\p^2}{2}-\frac{N_\0}{L^d}\hat{v}(\p)\right)+\sum\nolimits ' \tilde{E}(\p)\xi^{\dagger}_\p \xi_\p. 
\end{eqnarray*}
In the approximation we are considering, we can replace $N_\0/L^d$ by the density $\rho$. Then, finally,
\begin{eqnarray}
H=\frac{N\rho}{2}\hat{v}(\0)+\frac12 \sum\nolimits '\left(E(\p)-\frac{\p^2}{2}-\rho\hat{v}(\p)\right)+\sum\nolimits ' E(\p)\xi^{\dagger}_\p \xi_\p \label{fourthbogapprox}
\end{eqnarray}
where 
$$E(\p)=\sqrt{\p^2\rho\hat{v}(\p)+\frac{\p^4}{4}}.$$
Thus, the Bogoliubov approximation predicts that the low-lying spectrum of a weakly interacting Bose gas is quasiparticle-like in the sense of the definition in Subsection \ref{Quasiparticle-like excitation spectrum}. 

Since 
$$\inf_{\p\neq \0}\frac{E(\p)}{|\p|}>0,$$
it also predicts a \textit{positive critical velocity} and \textit{no energy gap} (recall the definitions in Subsection \ref{Excitation spectrum}).

\subsubsection{The improved Bogoliubov method}\label{improvedbogmethod}
Let us now present a version of the Bogoliubov approximation adapted to the grand-canonical setting with arbitrary chemical potential $\mu$. This presentation is based on \cite{CDZ2009}, which, as mentioned before, is not co-authored by the author of this thesis. We include this presentation nevertheless, because it was the immediate motivation for the work presented in the next sections and it fits perfectly into the structure of this thesis. 

We start by defining two operators. For $\alpha\in \CC$, we define the \textit{displacement} or \textit{Weyl operator} of the mode $\kk=\0$:
\begin{eqnarray}
W_\alpha :=\e^{-\alpha a^{\dagger}_\0+\bar{\alpha}a_\0}. \label{weyloperator}
\end{eqnarray} 
If $\Omega$ denotes Fock vacuum, then we define the \textit{coherent vector} by
\begin{eqnarray}
\Omega_{\alpha}:= W_{\alpha}^{\dagger}\Omega. \label{coherent}
\end{eqnarray}
 Then using the Lie formula
\begin{gather*}
\e^{-X}B\e^X=\sum_{j=0}^\infty\frac{(-1)^j}{j!}\underset{\hbox{ $j$ times\hskip 3ex}}{[X,...[X,B]\dots]}
\end{gather*} 
we obtain
\begin{eqnarray}
\tilde{a}_{\0}:=W_{\alpha}^{\dagger}a_\0 W_{\alpha}=a_\0-\alpha \quad \text{and} \quad \tilde{a}_{\kk}:=W_{\alpha}^{\dagger}a_\kk W_{\alpha}=a_\kk \quad \text{if} \quad \kk\neq \0. \label{bogtranslation}
\end{eqnarray}
Such a transformation is sometimes called the \textit{Bogoliubov translation}. Note that the operators with and without tildes satisfy the same commutation relations. In addition, the annihilation operators with tildes destroy the "translated vacuum" $\Omega_\alpha$. 

Now, let $\frac{2\pi}{L}\mathbb{Z}^{d}\ni \kk\mapsto \theta_\kk\in \CC$ be a square summable sequence with $\theta_\kk=\theta_{-\kk}$. For such a sequence let us define the unitary operator
\begin{eqnarray}
U_\theta := \prod_{\kk} \e^{-\frac12\theta_\kk a_\kk^{\dagger}a_{-\kk}^{\dagger}+\frac12\bar{\theta}_\kk a_\kk a_{-\kk}}. \label{ubogtransf}
\end{eqnarray}
 We then have 
 \begin{eqnarray}
 U^{\dagger}_\theta a_\kk U_\theta = \cosh |\theta_\kk | a_\kk-\frac{\theta_\kk}{|\theta_\kk |}\sinh|\theta_\kk |a_{-\kk}^{\dagger}=: b_\kk. \label{bogrotation}
 \end{eqnarray}
 The transformation above is called \textit{Bogoliubov rotation}. Furthermore, introducing 
\begin{eqnarray}
c_\kk := \cosh |\theta_\kk | \quad \text{and} \quad s_\kk :=-\frac{\theta_\kk}{|\theta_\kk |}\sinh|\theta_\kk |, \label{cksk}
\end{eqnarray}
we also have
\begin{eqnarray}
a^{\dagger}_\kk=c_\kk b^{\dagger}_\kk-\bar{s}_\kk b_{-\kk} \quad \text{and} \quad a_\kk=c_\kk b_\kk-s_\kk b^{\dagger}_{-\kk}. \label{invbogrot}
\end{eqnarray}
Both the Bogoliubov translation and the Bogoliubov rotation are special cases of the more general \textit{Bogoliubov transformations}. They will be discussed in more detail in Section \ref{beliaevthmsec}. The operator
$$U_{\alpha,\theta}:=W_{\alpha}U_{\theta}$$
is the general form of a Bogoliubov transformation commuting with the total momentum operator $P^L$. The vector
$$\Omega_{\alpha,\theta}:=U_{\alpha,\theta}^{\dagger} \Omega$$
is called a \textit{pure Gaussian vector} or \textit{squeezed vector}.

We shall now look for a pure Gaussian vector that minimizes the expectation value of $H^L_\mu$  (recall \eqref{hamboseboxpedpart2}). Clearly, if
\begin{eqnarray}
\tilde{b}_{\kk}:=U_{\alpha,\theta}^{\dagger}a_\kk U_{\alpha,\theta}, \label{tildeb} 
\end{eqnarray}
then
\begin{eqnarray*}
\tilde{b}_{\kk}\Omega_{\alpha,\theta}=0. 
\end{eqnarray*}
Thus, to calculate the expectation value mentioned above, it is useful to express $H^L_\mu$ in terms of $\tilde{b}_{\kk}$. 
We start by performing a Bogoliubov translation and expressing the Hamiltonian in terms of $\tilde{a}_\kk$. By \eqref{bogtranslation} 
we have 
\begin{eqnarray}\label{bogtranslationham}
H^L_{\mu} &=& -\mu|\alpha|^2+\frac{\hat{v}(\0)}{2L^d}|\alpha|^{4} \\
&+& \left(\frac{\hat{v}(\0)}{L^d}|\alpha|^{2}-\mu\right)(\bar{\alpha}a_\0+\alpha a_\0^{\dagger}) \nonumber \\
&+& \sum_{\kk}\left(\frac12 \kk^2-\mu+\frac{(\hat{v}(\0)+\hat{v}(\kk))}{L^d}|\alpha|^2 \right)a_\kk^{\dagger}a_\kk \nonumber \\
&+& \sum_{\kk}\frac{\hat{v}(\kk)}{2L^d}(\bar{\alpha}^2 a_\kk a_{-\kk}+\alpha^2 a^{\dagger}_\kk a^{\dagger}_{-\kk}) \nonumber \\
&+& \sum_{\kk,\kk'}\frac{\hat{v}(\kk)}{L^d}(\bar{\alpha}a^{\dagger}_{\kk+\kk'} a_\kk a_{\kk'}+\alpha a^{\dagger}_\kk a^{\dagger}_{\kk'}a_{\kk+\kk'}) \nonumber \\
&+&\frac{1}{2L^d}\sum_{\kk_1,\kk_2,\kk_3,\kk_4}\delta(\kk_1+\kk_2-\kk_3-\kk_4)\hat{v}(\kk_2-\kk_3)
a^{\dagger}_{\kk_1}a^{\dagger}_{\kk_2}a_{\kk_3}a_{\kk_4} \nonumber
\end{eqnarray}
where we dropped the tildes for notational simplicity. 

Now we perform a Bogoliubov rotation \eqref{invbogrot}. After this substitution the Hamiltonian in the Wick ordered form equals
\begin{eqnarray}
H &=& B +C b^{\dagger}_{\0}+\bar{C}b_{\0} \nonumber \\
&+& \frac12\sum_{\kk}O(\kk)b^{\dagger}_{\kk}b^{\dagger}_{-\kk}+\frac12\sum_{\kk}\bar{O}(\kk)b_{\kk}b_{-\kk}+\sum_{\kk}D(\kk)b^{\dagger}_{\kk}b_{\kk} \nonumber \\
&+& \text{terms higher order in $b$'s.} \label{rotbogham}
\end{eqnarray}
Above, we have dropped the tildes, superscript $L$ and subscript $\mu$. From now on, if not needed, we will not use them. In \eqref{rotbogham} B and C are given by
\begin{eqnarray*}
B&=& -\mu|\alpha|^2+\frac{\hat{v}(\0)}{2L^d}|\alpha|^{4} \\
 &+& \sum_{\kk}\left(\frac12 \kk^2-\mu+\frac{(\hat{v}(\0)+\hat{v}(\kk))}{L^d}|\alpha|^2 \right)|s_\kk|^2\\
 &-& \sum_{\kk}\frac{\hat{v}(\kk)}{2L^d}(\bar{\alpha}^2s_\kk c_\kk+\alpha^2\bar{s}_\kk c_\kk)\\
 &+& \sum_{\kk,\kk'}\frac{\hat{v}(\kk-\kk')}{2L^d}c_\kk s_\kk c_{\kk'}\bar{s}_{\kk'}\\
 &+&  \sum_{\kk,\kk'}\frac{\hat{v}(\0)+\hat{v}(\kk-\kk')}{2L^d}|s_\kk|^2|s_{\kk'}|^2
\end{eqnarray*}
and
\begin{eqnarray*}
C&=& \left(\frac{\hat{v}(\0)}{L^d}|\alpha|^{2}-\mu+ \sum_{\kk}\frac{\hat{v}(\0)+\hat{v}(\kk)}{L^d}|s_\kk|^2\right)(\alpha c_\0-\bar{\alpha}s_\0)\\
 &+&\sum_{\kk}\frac{\hat{v}(\kk)}{L^d}(\alpha s_\0 c_\kk \bar{s}_{\kk}-\bar{\alpha} c_\0 c_\kk s_\kk).
\end{eqnarray*}
Furthermore, introducing 
\begin{eqnarray}
f_\kk := \frac{\kk^2}{2}-\mu+|\alpha|^2\frac{\hat{v}(\0)+\hat{v}(\kk)}{L^d}+\sum_{\kk'}\frac{\hat{v}(\0)+\hat{v}(\kk'-\kk)}{L^d}|s_{\kk'}|^2 \label{f_k}
\end{eqnarray}
and
\begin{eqnarray}
g_{\kk}:=\alpha^2\frac{\hat{v}(\kk)}{L^d}-\sum_{\kk'}\frac{\hat{v}(\kk'-\kk)}{L^d}s_{\kk'}c_{\kk'}\label{g_k}
\end{eqnarray}
one can express $O(\kk)$ and $D(\kk)$ by
\begin{eqnarray}
O(\kk)&=&-2c_{\kk}s_{\kk}f_{\kk}+s_{\kk}^2\bar{g}_{\kk}+c^2_{\kk}g_{\kk}, \label{O(k)}\\
D(\kk)&=&f_{\kk}(c^2_\kk +|s_\kk|^2)-c_\kk(s_\kk\bar{g}_{\kk}+\bar{s}_\kk g_\kk). \label{D(k)}
\end{eqnarray}
Note that $f_\kk$ is real. 

Recall we want to minimize $H$ over $\Omega_{\alpha, \theta}$. By \eqref{rotbogham} we have
$$\langle \Omega_{\alpha, \theta}|H \Omega_{\alpha, \theta}\rangle =B.$$
Thus we demand that $B$ attains a minimum. We shall start by minimizing $B$ over $\alpha$. Since $\alpha$ is a complex parameter we can minimize independently with respect to $\alpha$ and $\bar{\alpha}$. The derivatives are
\begin{eqnarray*}
\partial_{\alpha}B&=&\left(\frac{\hat{v}(\0)}{L^d}|\alpha|^{2}-\mu+ \sum_{\kk}\frac{\hat{v}(\0)+\hat{v}(\kk)}{L^d}|s_\kk|^2\right)\bar{\alpha}-\sum_{\kk}\frac{\hat{v}(\kk)}{L^d}\alpha c_\kk \bar{s}_{\kk}, \\
\partial_{\bar{\alpha}}B&=&\left(\frac{\hat{v}(\0)}{L^d}|\alpha|^{2}-\mu+ \sum_{\kk}\frac{\hat{v}(\0)+\hat{v}(\kk)}{L^d}|s_\kk|^2\right)\alpha-\sum_{\kk}\frac{\hat{v}(\kk)}{L^d}\bar{\alpha} c_\kk s_{\kk}.
\end{eqnarray*} 
It follows that
$$C=c_\0 \partial_{\bar{\alpha}}B-s_\0 \partial_{\alpha}B,$$
and thus the condition
\begin{eqnarray}
\partial_{\bar{\alpha}}B=\partial_{\alpha}B=0 \label{alphaminimizationB}
\end{eqnarray}
entails $C=0$.
By \eqref{alphaminimizationB} we obtain 
\begin{eqnarray}
\mu=\frac{\hat{v}(\0)}{L^d}|\alpha|^{2}+\sum_{\kk}\frac{\hat{v}(\0)+\hat{v}(\kk)}{L^d}|s_\kk|^2-\frac{\alpha^2}{|\alpha|^2}\sum_{\kk}\frac{\hat{v}(\kk)}{L^d}\bar{s}_{\kk'}c_{\kk'}. \label{muminimizing}
\end{eqnarray}
Eliminating $\mu$ from $f_\kk$ we obtain
\begin{eqnarray}
f_\kk=\frac{\kk^2}{2}+|\alpha|^2\frac{\hat{v}(\kk)}{L^d}+\sum_{\kk'}\frac{\hat{v}(\kk'-\kk)-\hat{v}(\kk')}{L^d}|s_{\kk'}|^2+ \frac{\alpha^2}{|\alpha|^2}\sum_{\kk}\frac{\hat{v}(\kk)}{L^d}\bar{s}_{\kk'}c_{\kk'}.\label{f_kbezmu}
\end{eqnarray}
Let us now derive the conditions arising from the minimization of the energy over $\theta_\kk$. By \eqref{cksk} we see that instead of minimizing over $\theta_\kk\in \CC$ we can choose $s_\kk$  and $\bar{s}_\kk$ as the independent parameters. \eqref{cksk} implies 
\begin{eqnarray*}
s_\kk \bar{s}_\kk=c_\kk^2-1
\end{eqnarray*}
and thus
\begin{eqnarray*}
\partial_{s_\kk}c_\kk=\frac{\bar{s}_\kk}{2c_\kk}, \quad\quad \partial_{\bar{s}_\kk}c_\kk=\frac{s_\kk}{2c_\kk}.
\end{eqnarray*}
Using that we obtain
\begin{eqnarray}
\partial_{s_\kk}B&=& f_\kk \bar{s}_\kk -\frac{\bar{g}_\kk}{2}\left(c_\kk+\frac{|s_\kk|^2}{2c_\kk}\right)-g_\kk \frac{\bar{s}_\kk^2}{4c_\kk}, \label{partialsB}\\
\partial_{\bar{s}_\kk}B&=& f_\kk s_\kk -\frac{g_\kk}{2}\left(c_\kk+\frac{|s_\kk|^2}{2c_\kk}\right)-\bar{g}_\kk \frac{s_\kk^2}{4c_\kk}. \label{partialsbarB}
\end{eqnarray}
We then calculate that 
\begin{eqnarray*}
O(\kk)=\left(-2c_\kk+\frac{|s_\kk|^2}{c_\kk}\right)\partial_{\bar{s}_\kk}B-\frac{s_\kk^2}{c_\kk}\partial_{s_\kk}B.
\end{eqnarray*}
Thus, the condition
\begin{eqnarray}
\partial_{s_\kk}B=\partial_{\bar{s}_\kk}B=0 \label{sminimizationB}
\end{eqnarray}
entails $O(\kk)=0$.

Note also, that \eqref{partialsB} and \eqref{partialsbarB} imply
$$s_\kk \partial_{s_\kk}B-\bar{s}_\kk \partial_{\bar{s}_\kk}B= \frac{c_\kk}{2}(g_\kk \bar{s}_\kk-\bar{g}_\kk s_\kk),$$
and hence \eqref{sminimizationB} implies
\begin{eqnarray}
g_\kk \bar{s}_\kk =\bar{g}_\kk s_\kk. \label{gs=gs}
\end{eqnarray}
Using that we obtain
\begin{eqnarray}
D(\kk)&=&(c^2_\kk+|s_\kk|^2)f_\kk-2 s_\kk c_\kk\bar{g}_\kk, \label{D(k)dwa}\\
O(\kk)&=&-2 s_\kk c_\kk f_\kk+ (c^2_\kk+|s_\kk|^2)g_\kk. \label{O(k)dwa} 
\end{eqnarray}
If we assume $O(\kk)=0$ and $f_\kk\neq 0$, then
$$D(\kk)=\frac{c_\kk+|s_\kk|^2}{f_\kk}\left(f_\kk^2-|g_\kk|^2\right)=\frac{2c_\kk s_\kk}{g_\kk}\left(f_\kk^2-|g_\kk|^2\right).$$
We obtain the solution
\begin{eqnarray}
D(\kk)&=& (\sgn f_\kk) \sqrt{f_\kk^2-|g_\kk|^2}, \label{D(k)sol}\\
S_\kk  &=& \frac{g_\kk}{D(\kk)}, \label{Sksol}\\
C_\kk &=& \frac{f_\kk}{D(\kk)}, \label{Cksol}
\end{eqnarray}
where we introduced
$$S_\kk := 2s_\kk c_\kk, \quad \quad C_\kk=c_\kk^2+|s_\kk|^2.$$
If we also set $\e^{\ii \tau}:=\frac{\alpha}{|\alpha|}$,  then we can write
\begin{eqnarray}
f_\kk &=& \frac{\kk^2}{2}+|\alpha|^2\frac{\hat{v}(\kk)}{L^d} \nonumber \\
&+&\sum_{\kk'}\frac{\hat{v}(\kk'-\kk)-\hat{v}(\kk')}{2L^d}(C_{\kk'}-1)+ \sum_{\kk'}\frac{\hat{v}(\kk')}{2L^d}\e^{2\ii \tau}\bar{S}_{\kk'}, \label{fktau} \\
g_\kk &=& \alpha^2\frac{\hat{v}(\kk)}{L^d}-\sum_{\kk'}\frac{\hat{v}(\kk'-\kk)}{2L^d}S_{\kk'}. \label{gktau}
\end{eqnarray}

Let us now recap what we have presented above. We started with the Hamiltonian \eqref{hamboseboxpedpart2} and we minimized its expectation value with respect to pure Gaussian states $\Omega_{\alpha,\theta}$. To this end we expressed the Hamiltonian in terms of creation and annihilation operators $b_{\kk}^{\dagger}$ and $b_{\kk}$ such that $b\Omega_{\alpha,\theta}=0$. 

Then, after some tedious calculations, we noticed that the minimizing conditions imply $C=0$ and $O(\kk)=0$ in \eqref{rotbogham}. We would like to stress that this turns out to be a special case of a more general fact which will be described in abstract terms in Section \ref{beliaevthmsec}. Putting $C=0$ and $O(\kk)=0$ in \eqref{rotbogham} yields
\begin{eqnarray}
H=B+\sum_{\kk}D(\kk)b^{\dagger}_{\kk}b_{\kk}.
\end{eqnarray} 
Clearly $B$ is a rigorous upper bound for the ground state energy of \eqref{hamboseboxpedpart2}. 

We shall now look at $D(\kk)$ given by \eqref{D(k)sol}. First notice that the case of $D(\kk)<0$ seems physically irrelevant. Thus, we can assume that $f_\kk>0$ and thus 
$$D(\kk)=   \sqrt{f_\kk^2-|g_\kk|^2}$$
where $f_\kk$ is given by \eqref{fktau} and $g_\kk$ is given by \eqref{gktau}. 

We shall now try to find parameters $\alpha$ and $S_\kk$ in \eqref{fktau} and \eqref{gktau} which satisfy the minimization condition. This is of course a difficult task and one could try to do it by iterations.

 A natural starting point seems to be $S_\kk=0$  (and thus also $s_\kk=0$). Then by \eqref{muminimizing} we have 
 $$\mu=\hat{v}(\0)\kappa,$$
where we put $\alpha=\sqrt{V\kappa}\e^{\ii \tau}$. The reason we introduced the fixed parameter $\kappa>0$ is that it has the interpretation of the density of the condensate. We shall not elaborate on that.

Then, after one iteration, we obtain
\begin{eqnarray*}
f_\kk &=& \frac{\kk^2}{2}+\kappa\hat{v}(\kk), \\
g_\kk &=& \kappa \hat{v}(\kk)
\end{eqnarray*}
and thus
\begin{eqnarray*}
D(\kk) = \sqrt{(\kk^2/2)^2+\kk^2\kappa\hat{v}(\kk)}.
\end{eqnarray*}

We thus reconstruct the quasiparticle-like excitation spectrum of the Bogoliubov approximation presented in Section \ref{The original Bogoliubov approximation}. In particular we reconstruct the same dispersion relation of the quasiparticles as in  \eqref{diagbogopcoeff}, that is before replacing the condensate density by the full density of the considered system. For further advantages and consequences of this approach we refer to \cite{CDZ2009}.

\section[Energy-momentum spectrum of a homogeneous Fermi gas]{Energy-momentum spectrum of a homogeneous Fermi gas} \label{fermigassec}
We shall now turn our attention to the description of the low-lying excitation spectrum of a homogeneous Fermi gas. This description can be obtained through the \textit{Hartree--Fock--Bogoliubov} approximation which can be looked upon as the fermionic counterpart of the bosonic Bogoliubov approximation described in the previous section. 

Internal degrees of freedom of particles, such as spin, play an important  role in fermionic systems. In particular, they are crucial in the BCS approach (\cite{BCS1957}). Therefore, we will take them into account. This, however, leads to a more general form of the kinetic and potential energy than in the case of spinless bosons. This will be discussed in the next subsection. In Subsection \ref{HFBBCS} the Hartree-Fock-Bogoliubov approximation will be applied to a general, spin-dependent Hamiltonian. 

This section is based mainly on \cite{DMN2013}.

\subsection{The model}
\subsubsection{Kinetic energy}

We assume that the internal degrees of freedom are described
by a finite dimensional Hilbert space $\CC^m$. Thus the one-particle space of
the system is $L^2(\RR^d,\CC^m)$.

The kinetic energy of one particle including its chemical potential is given by  a self-adjoint operator $T$ on $L^2(\RR^d,\CC^m)$.
We use the following notation for its integral kernel: for $\Phi\in
L^2(\RR^d,\CC^m)$, 
$$(T\Phi)_{i_1}(\x_{1})=
\sum_{i_{2}}\int T_{i_{1},i_{2}}(\x_{1},\x_{2}) \Phi_{i_{2}}(\x_{2})\d \x_{2}.$$
 We assume that $T$ is a self-adjoint and translation invariant 1-body operator. Then,
\begin{eqnarray*}
T_{i_1,i_2}(\x_1,\x_2)&=&\overline{ T_{i_2,i_1}(\x_2,\x_1)}\\
&=&T_{i_1,i_2}(\x_1+\y,\x_2+\y).\end{eqnarray*}
The first identity expresses the hermiticity of $T$ while the second the
translation invariance of $T$.

We will sometimes assume that $T$ is {\em real}, that is, invariant with respect to the complex conjugation. This means that $T_{i_1,i_2}(\x_1,\x_2)$ are real.
An example of a real 1-particle energy is
\[T_{ij}=\big(-\frac{1}{2m_i}\Delta-\mu_i\big)\delta_{i,j},\] where the $i$th
``spin'' has the mass $m_i$  and the chemical
potential $\mu_i$.

If the operator $T$ has the form 
$$T_{i,j}(\x_i,\x_j)=t(\x_i,\x_j)\delta_{i,j},$$
for some function $t$ satisfying

\begin{eqnarray*}
t(\x_1,\x_2)&=&\overline{t(\x_2,\x_1)}\\
&=&t(\x_1+\y,\x_2+\y),
\end{eqnarray*}

then we will say that $T$ is {\em spin-independent}.

Clearly, the 1-particle energy can be written as
$$T_{i,j}(\x_{1},\x_{2})=(2\pi)^{-d}\int\tau_{i,j}(\kk)\e^{\ii\kk(\x_{1}-\x_{2})}\d\kk.$$

If it is real, then $\tau_{i,j}(\kk)=\tau_{i,j}(-\kk).$

If it is spin independent, then $\tau_{i,j}(\kk)=\tau(\kk)\delta_{i,j}.$

In the real spin-independent case we have $\tau(\kk)=\tau(-\kk)$.

\subsubsection{Interaction}

The interaction in the Fermi gas will be described by a
2-body operator $V$. It acts on the antisymmetric 2-particle
space as
\[(V\Phi)_{i_1,i_2}(\x_1,\x_2)=
\sum_{i_3,i_4}\int\int
V_{i_1,i_2,i_3,i_4}(\x_1,\x_2,\x_3,\x_4)\Phi_{i_4,i_3}(\x_4,\x_3)\d \x_3\d \x_4,\]
where  $\Phi\in \Gamma_{\mathrm{a}}^{2}\left( L^2(\RR^d,\CC^m)\right)$.
We will assume that it is self-adjoint and translation invariant. Its  integral kernel satisfies 
\begin{eqnarray*}
V_{i_1,i_2,i_3,i_4}(\x_1,\x_2,\x_3,\x_4)&=&
-V_{i_2,i_1,i_3,i_4}(\x_2,\x_1,\x_3,\x_4)
\\
&=&-V_{i_1,i_2,i_4,i_3}(\x_1,\x_2,\x_4,\x_3)\\
&=&\overline{V_{i_4,i_3,i_2,i_1}(\x_4,\x_3,\x_2,\x_1)}\\
&=&V_{i_1,i_2,i_3,i_4}(\x_1+\y,\x_2+\y,\x_3+\y,\x_4+\y).
\end{eqnarray*}
The first two identities express the antisymmetry of the interaction, the
third --  its hermiticity and the fourth -- its translation invariance.
 We also  assume 
that $V(\x_1,\x_2,\x_3,\x_4)$ decays for large differences of its
 arguments sufficiently fast.

We will sometimes assume that $V$  are {\em real}, that means, they are invariant with respect to
the complex conjugation. This means
 $V_{i_1,i_2,i_3,i_4}(\x_1,\x_2,\x_3,\x_4)$ is real.

 We will say 
that the operator $V$ is {\em spin-independent}
if there exists a function
 $v(\x_1,\x_2,\x_3,\x_4)$
such that \begin{eqnarray*}
&&V_{i_1,i_2,i_3,i_4}(\x_1,\x_2,\x_3,\x_4)\\&=&
\frac12\bigl(v(\x_1,\x_2,\x_3,\x_4)
\delta_{i_1,i_4}\delta_{i_2,i_3}
-v(\x_1,\x_2,\x_4,\x_3)\delta_{i_1,i_3}\delta_{i_2,i_4}\bigr),
\end{eqnarray*}
Note that
\begin{eqnarray*}
v(\x_1,\x_2,\x_3,\x_4)&=&
v(\x_2,\x_1,\x_4,\x_3)\\
&=&\overline{v(\x_4,\x_3,\x_2,\x_1)}\\
&=&v(\x_1+\y,\x_2+\y,\x_3+\y,\x_4+\y).
\end{eqnarray*}

It will be convenient to write the Fourier transform of $V$ as follows
\begin{eqnarray*}
&&V(\x_1,\x_2,\x_3,\x_4)\\&=&
(2\pi)^{-4d}
\int\e^{\ii \kk_1\x_1+\ii\kk_2\x_2-\ii \kk_3\x_3-\ii\kk_4\x_4}
 Q(\kk_1,\kk_2,\kk_3,\kk_4)\\&&\times\delta
(\kk_1+\kk_2-\kk_3-\kk_4)\d\kk_1\d\kk_2\d\kk_3\d\kk_4,\end{eqnarray*}
where
$ Q(\kk_1,\kk_2,\kk_3,\kk_4)$ is a function defined on the subspace
$\kk_1+\kk_2=\kk_3+\kk_4$. (Thus we could drop, say, $\kk_4$ from its
arguments; we do not do it for the sake of the symmetry of formulas).
Clearly,
\begin{eqnarray*}
Q_{i_1,i_2,i_3,i_4}(\kk_1,\kk_2,\kk_3,\kk_4)&=&
-Q_{i_2,i_1,i_3,i_4}(\kk_2,\kk_1,\kk_3,\kk_4)
\\
&=&-Q_{i_1,i_2,i_4,i_3}(\kk_1,\kk_2,\kk_4,\kk_3)\\
&=&\overline{Q_{i_4,i_3,i_2,i_1}(\kk_4,\kk_3,\kk_2,\kk_1)}.
\end{eqnarray*}

If we assume that the interaction is real, then
\begin{eqnarray*}
 Q_{i_1,i_2,i_3,i_k}(\kk_1,\kk_2,\kk_3,\kk_4)&=&\overline{ Q_{i_1,i_2,i_3,i_k}
(-\kk_1,-\kk_2,-\kk_3,-\kk_4)}.
\end{eqnarray*}

If  we assume that the interaction is spin-independent, then 
\begin{eqnarray*} 
Q_{i_{1}i_{2}i_{3}i_{4}}(\kk_1,\kk_2,\kk_3,\kk_4) & = &\frac12\bigl(
q(\kk_1,\kk_2,\kk_3,\kk_4)\delta_{i_{1},i_{4}}\delta_{i_{2},i_{3}}
-q(\kk_1,\kk_2,\kk_4,\kk_3)\delta_{i_{1},i_{3}}\delta_{i_{2},i_{4}}\bigr),
\end{eqnarray*}
for some function $q$ defined on $\kk_1+\kk_2=\kk_3+\kk_4$
satisfying
\begin{eqnarray*}
q(\kk_1,\kk_2,\kk_3,\kk_4)&=&
q(\kk_2,\kk_1,\kk_4,\kk_3)\\
&=&\overline{q(\kk_4,\kk_3,\kk_2,\kk_1)}.
\end{eqnarray*}

In the real spin-independent case we have in addition
$$q(\kk_1,\kk_2,\kk_3,\kk_4)=\overline{q(-\kk_1,-\kk_2,-\kk_3,-\kk_4)}.$$

For example,
 a 2-body potential $V(\x)$ such that $V(\x)=V(-\x)$ corresponds to the real spin-independent interaction with
\begin{eqnarray*}
v(\x_1,\x_2,\x_3,\x_4)
&=& V(\x_1-\x_2)
\delta(\x_1-\x_4)\delta(\x_2-\x_3),\\
q(\kk_1,\kk_2,\kk_3,\kk_4)&=&\int\d\q
\hat V(\q)\delta(\kk_1-\kk_4-\q)\delta(\kk_2-\kk_3+\q).
\end{eqnarray*}

\subsubsection{$N$-body Hamiltonian}

\quad The $N$-body  Hamiltonian of the homogeneous Fermi gas 
acts on  the Hilbert space $\Gamma_\mathrm{a}^N\left(L^2(\RR^d,\CC^m)\right)$ 
(antisymmetric square integrable functions on $(\RR^d)^N$ with values in $(\CC^m)^{\otimes N}$).
Let $T_{i}$ denote  the operator $T$ acting on
the $i$th variable and $V_{ij}$ denote the  operator $V$ acting on the $ij$th pair of variables. The full $N$-body Hamiltonian  equals
\begin{equation}
H_{N}=
\sum_{1\leq i\leq N}T_{i}+
\sum_{1\leq i<j\leq N}V_{ij}. \label{Nbodyham}
\end{equation} 
Recall we have assumed that the kinetic energy $T$ and interaction $V$ are translation invariant. Thus $H_N$ commutes with the total momentum operator 
$$P_{N}:=\sum_{i=1}^N-\ii\nabla_{i}.$$

\subsubsection{Putting system in a box}

As before, since we want to investigate homogeneous Fermi gas at positive density, we restrict \eqref{Nbodyham}
to $\Lambda=[-L/2,L/2]^d$ , the $d$-dimensional cubic box of side length $L$ with periodic boundary conditions. 

This means in particular that
the kinetic energy is replaced by
$$T^L(\x_1,\x_2)=\frac{1}{L^d}
\sum_{\kk \in \frac{2\pi}{L}\mathbb{Z}^d}\e^{\ii \kk\cdot( \x_1-\x_2)}\tau(\kk) ,$$
and the  potential $V$ is replaced  by
\begin{eqnarray*}
&&V^L(\x_1,\x_2,\x_3,\x_4)=\\
&=&\frac{1}{L^{4d}}\sum_{\substack{\kk_1,\dots,\kk_4
\in \frac{2\pi}{L}\mathbb{Z}^d, \\ \kk_1+\kk_2=\kk_3+\kk_4}}
\e^{\ii\kk_1 \x_1+\ii\kk_2\x_2-\ii\kk_3\x_3-
\ii\kk_4\x_4} Q(\kk_1,\kk_2,\kk_3,\kk_4).
\end{eqnarray*}
Note that $V^L$ is periodic with respect to the domain
$\Lambda$, and $V^L(\x)\to V(\x)$ as $L\to\infty$. 
 The system  on a torus  is described by the Hamiltonian
\begin{equation}
H^{L}_N=
\sum_{1\leq i\leq N}T_{i}^L+
\sum_{1\leq i<j\leq N}V_{ij}^L
 \label{boxNbodyham}
 \end{equation}
acting on the space $\Gamma_\mathrm{a}^N\left(L^2(\Lambda,\CC^m)\right)$.

\subsubsection{Grand-canonical Hamiltonian of the Fermi gas}
As before, it is convenient to put all the $N$-particle spaces into a single fermionic Fock space
 $$\Gamma_\mathrm{a}(L^2(\Lambda,\CC^m)):=\mathop{\oplus}\limits_{N=0}^\infty 
\Gamma_\mathrm{a}^N\left(L^2(\Lambda,\CC^m)\right)$$
with the Hamiltonian 
\begin{eqnarray*}
H^L&:=&\mathop{\oplus}\limits_{N=0}^\infty H^{L}_N \nonumber\\
&=&\int a_{\x,i_1}^{\dagger} T_{i_1,i_2}^L(\x_{i_1}-\x_{i_2}) a_{\x,i_2}\d
\x_1\d\x_2 \nonumber\\
&&+\frac12\int\int a_{\x_1,i_1}^{\dagger} a_{\x_2,i_2}^{\dagger} 
V_{i_1,i_2,i_3,i_4}^L(\x_1,\x_2,\x_3,\x_4) a_{\x_3,i_3} a_{\x_4,i_4}\d
\x_1\d\x_2\d\x_3\d\x_4, 
\end{eqnarray*}
where $a_{\x,i}$, $a_{\x,i}^{\dagger}$ are the usual fermionic annihilation and
creation operators. The second quantized
momentum and number operators are defined as
\begin{eqnarray*}
P^L&:=&
\mathop{\oplus}\limits_{N=0}^\infty P^{L}_N=-\ii\int
a_{\x,i}^{\dagger} \nabla_\x^La_{\x,i} \d\x,\\
N^L&:
=&\mathop{\oplus}\limits_{N=0}^\infty N=\int a_{\x,i}^{\dagger} a_{\x,i}\d\x
.\end{eqnarray*}
Above we use the summation convention. 

In the momentum representation (with the indices omitted),
\begin{eqnarray}
H^L 
&=&\sum_{\kk}\tau(\kk)a^{\dagger} _{\kk}
a_{\kk}\nonumber \\
&&+\frac{1}{2L^{d}}\sum_{\kk_1+\kk_2=\kk_3+\kk_4}Q(\kk_1,\kk_2,\kk_3,\kk_4)
a^{\dagger} _{\kk_1}a^{\dagger} _{\kk_2}a_{\kk_3}a_{\kk_4},\label{ham2kwantped}
\\
P^L&=&\sum_{\kk}\kk a^{\dagger} _{\kk}a_\kk,\nonumber
\\
N^L&=&\sum_{\kk}a^{\dagger} _{\kk}a_\kk.\nonumber
\end{eqnarray}
Here the summation over momenta runs over $\kk
\in \frac{2\pi}{L}\mathbb{Z}^d$. In the spin-independent case, the interaction equals
\[\frac{1}{2L^{d}}\sum_{\kk_1+\kk_2=\kk_3+\kk_4}q(\kk_1,\kk_2,\kk_3,\kk_4)
a^{\dagger} _{\kk_1,i}a^{\dagger} _{\kk_2,j}a_{\kk_3,j}a_{\kk_4,i}
\]
In the case of a (local) potential, it is
\[\frac{1}{2L^d}\sum_{\kk,\kk',\q}\hat V(\q)
a^{\dagger} _{\kk+\q,i}a^{\dagger} _{\kk'-\q,j}a_{\kk',j}a_{\kk,i}.
\]

\subsection{The Hartree-Fock-Bogoliubov approximation applied to a homogeneous Fermi gas}\label{HFBBCS}
We shall now present how one can try to compute the excitation spectrum of the homogeneous interacting Fermi gas by approximate methods. Historically, the
first computation of this sort is due to Bardeen-Cooper-Schrieffer in 1957 (\cite{BCS1957}). In its original version, the BCS method involved a replacement of quadratic fermionic operators with bosonic ones. We will use an approach based on a Bogoliubov rotation of fermionic variables, which is commonly called the Hartree-Fock-Bogoliubov method. Its main idea is the same as in Subsection \ref{improvedbogmethod}, that is to minimize the energy in the so-called fermionic Gaussian
 states -- states obtained by a Bogoliubov rotation from the fermionic Fock vacuum. The minimizing state will define new creation/annihilation operators. We express the Hamiltonian in the new creation/annihilation operators and drop all higher order terms. This defines a new Hamiltonian, that we expect to give an approximate description of low-energy part of the excitation spectrum.

\subsubsection{The rotated Hamiltonian}
We start the HFB method with a rotation of the fermionic creation and annihilation operators. 
Similarly to the bosonic case, for any $\kk$, this corresponds to a substitution
\begin{eqnarray}
&a_\kk^{\dagger}=c_\kk b_\kk^{\dagger}+s_\kk b_{-\kk},&
a_\kk=\bar{c}_\kk b_\kk+\bar{s}_\kk b_{-\kk}^{\dagger},\label{rota}\end{eqnarray}
where $c_\kk$ and  $s_\kk$ are matrices on $\CC^m$ satisfying
\begin{eqnarray}\label{susti1}
c_\kk c_\kk^{\dagger}+s_\kk s_\kk^{\dagger}&=&1,\\ \label{susti2}
c_\kk s_{-\kk}^{\text{T}}+s_\kk c_{-\kk}^{\text{T}}&=&0.
\end{eqnarray}
($\dagger$ denotes the hermitian conjugation,
$\cdot^{\text{T}}$ denotes the transposition and $\overline{\cdot}$ denotes the complex conjugation). 



For a sequence
$\frac{2\pi}{L}{\ZZ}^d\ni \kk\mapsto \theta_\kk$ 
with values in matrices on $\CC^m$ such that
 $\theta_\kk=\theta_{-\kk}$, set
\begin{equation}
U_{\theta}:=\prod_\kk\e^{-\frac12
\theta_\kk a_\kk^{\dagger} a_{-\kk}^{\dagger}+\frac12
    \theta_\kk^{\dagger} a_\kk a_{-\kk}}.\label{uthe}\end{equation}
It is well known that for an appropriate sequence $\theta$
we have
\[U_\theta^{\dagger}
a_\kk
U_\theta=b_\kk,\ \
\ \ 
U_\theta^{\dagger}
a_\kk^{\dagger}U_\theta
=b_\kk^{\dagger}.\]
As in the bosonic case $U_{\theta}$
is the general form of an even
 Bogoliubov transformation commuting with $P^L$.

Note that \eqref{susti1} guarantees that $[a_\kk^{\dagger},a_\kk]_+=1$,
\eqref{susti2} guarantees that $[a_\kk^{\dagger},a_{\kk}^{\dagger}]_+=0$ while 
 $[a_\kk^{\dagger},a_{\kk'}]_+=0$,  $[a_\kk^{\dagger},a_{\kk'}^{\dagger}]_+=0$ for $\kk\neq\kk'$
 are satisfied automatically.

In this subsection  we drop the superscript $L$, writing, e.g., $H$ for $H^L$.
The Hamiltonian \eqref{ham2kwantped} after the substitution (\ref{rota}) and
 the Wick ordering equals
\begin{align}\nonumber
   H 
&=
B
\\&+\frac12\sum_\kk O(\kk)b_\kk^{\dagger}b_{-\kk}^{\dagger}+
\frac12\sum_\kk\bar O(\kk)b_{-\kk} b_{\kk} +\sum_\kk
D(\kk)b_\kk^{\dagger}b_\kk\nonumber
\\&
+\hbox{terms  higher order in {\it b}'s}\label{non}.
\end{align}
 A tedious computations leads to the following explicit formulas for $B$,
 $D(\kk)$ and $O(\kk)$:
\begin{eqnarray}
B&=&
\sum_\kk\tau(\kk) s_\kk \bar s_\kk
\nonumber\\&&
+\frac{1}{2L^d}\sum_{\kk,\kk'}
 Q(\kk,-\kk,-\kk',\kk') s_{\kk}c_{-\kk}\bar c_{-\kk'}\bar s_{\kk'}\nonumber\\
&&+\frac{1}{L^d}
\sum_{\kk,\kk'}Q(\kk,\kk',\kk',\kk)s_{\kk}s_{\kk'}
\bar s_{\kk'}\bar s_{\kk}
;\nonumber
\end{eqnarray}
\begin{eqnarray}
O(\kk)&=&
2\tau(\kk) c_\kk\bar s_\kk\nonumber\\
&&+\frac{1}{L^d}\sum_{\kk'}
Q(\kk',-\kk',-\kk,\kk)s_{\kk'}c_{-\kk'}\bar s_{-\kk}\bar
s_{\kk}\nonumber\\ 
 &&+\frac{1}{L^d}\sum_{\kk'}
Q(\kk,-\kk,-\kk',\kk')c_{\kk}c_{-\kk}\bar c_{-\kk'}\bar s_{\kk'}
\nonumber\\
&&+\frac{4}{L^d}\sum_{\kk'} Q(\kk,\kk',\kk',\kk)c_{\kk}s_{\kk'}
\bar s_{\kk'}\bar
s_{\kk}\nonumber\\ 
D(\kk)&=&\tau(\kk)
c_\kk\bar c_\kk-\bigl(\tau(\kk) s_{-\kk} \bar s_{-\kk}\bigr)^{\rm T}\nonumber\\
&&+\frac{1}{L^d}\sum_{\kk'}
 Q(\kk',-\kk',-\kk,\kk)s_{\kk'}c_{-\kk'}\bar s_{-\kk}\bar
c_{\kk}\nonumber\\ 
 &&+\frac{1}{L^d}\sum_{\kk'}
 Q(\kk,-\kk,-\kk',\kk')c_{\kk}s_{-\kk}\bar c_{-\kk'}\bar s_{\kk'}
\nonumber\\
&&+\frac{2}{L^d}\sum_{
\kk'}
 Q(\kk,\kk',\kk',\kk)
c_{\kk}s_{\kk'}\bar s_{\kk'}\bar
c_{\kk}
\nonumber\\
&&-\frac{2}{L^d}\sum_{
\kk'}\bigl(
 Q(-\kk,\kk',\kk',-\kk)
 s_{-\kk}s_{\kk'}\bar s_{\kk'}
\bar s_{-\kk}\bigr)^{\rm T}.
\nonumber
\end{eqnarray}

Note that the       formulas for $B$, $O(\kk)$ and $D(\kk)$ are
written in a special notation, whose aim is to avoid putting a big
number of internal indices.
 The matrices $c_\kk$ and $s_\kk$ have two internal indices: right and left. 
 We sum over the right internal indices, whenever we sum over the
 corresponding momenta. The left internal indices are contracted with the
 corresponding indices of $\tau$ or $ Q$.
The superscript $\rm T$ stands for the transposition (swapping
the indices).

\subsubsection{Minimization over Gaussian states}
\label{Minimization over Gaussian states}
Let $\Omega$ denote the  vacuum vector. $\Omega_{\theta}:=U_{\theta}^{\dagger}\Omega$
is the general form of an even fermionic Gaussian  vector of zero momentum.
Clearly,
\begin{eqnarray}
(\Omega_{\theta}|H \Omega_{\theta})&=&B,\label{min1}
\\ (b_\kk^{\dagger}\Omega_{\theta}|H 
b_\kk^{\dagger}\Omega_{\theta})
&=&B+D(\kk).\label{min2}\end{eqnarray}

We would like to find a fermionic Gaussian  vector that minimizes $B$ -- the expectation value of
$H$.
We assume that there exists a stationary point $(\tilde{s}_\kk,\tilde{c}_\kk)$
of $B$ considered as a function of $c$ and $s$. 
 Bogoliubov transformations form a group,
hence the neighbourhood of the stationary point can be expressed in the
following way: 
\begin{equation}
\left[ \begin{array}{cc}
  c_\kk & s_\kk \\
  \bar s_\kk & \bar c_\kk\\       
     \end{array}  \right]=
\left[ \begin{array}{cc}
  \tilde{c}_\kk & \tilde{s}_\kk \\
  \bar{\tilde s}_\kk & \bar{\tilde c}_\kk\\       
     \end{array}  \right]
\left[ \begin{array}{cc}
  c_\kk' & s'_\kk \\
  \bar s'_\kk & \bar c_\kk'\\       
     \end{array}  \right].
\end{equation}
This means (including internal indices) that
\begin{eqnarray*}
c_{il, \kk}&=&\tilde{c}_{im, \kk}c'_{ml, \kk}+\tilde{s}_{im, \kk}\bar{s'}_{ml,
  \kk}, \nonumber \\
\bar{c}_{il, \kk}&=&\bar{\tilde{s}}_{im, \kk}s'_{ml, \kk}+\bar{\tilde{c}}_{im,
  \kk}\bar c'_{ml, \kk}, \nonumber \\
s_{il, \kk}&=&\tilde{c}_{im, \kk}s'_{ml, \kk}+\tilde{s}_{im, \kk}\bar c'_{ml,
  \kk}  \nonumber \\
\bar{s}_{il, \kk}&=&\bar{\tilde{s}}_{im, \kk}c'_{ml, \kk}+\bar{\tilde{c}}_{im, \kk}\bar{s'}_{ml, \kk}. \nonumber
\end{eqnarray*}
We enter the above formulas into the expressions for $B, O(\kk)$ and
$D(\kk)$. 

We can always multiply $c_\kk$ and $s_\kk$ by a unitary matrix without changing
the Gaussian state. Hence, we can assume that
\begin{eqnarray}
c_\kk'&=&\sqrt{1-(s_\kk')^{\dagger}s_\kk'}.
\label{pas1}\end{eqnarray}
Since $s'$ is a complex function we can treat $s'$ and $\bar{s}'$ as
independent variables.
$c_\kk=\tilde c_\kk$, $s_\kk= \tilde s_\kk$ corresponds to
$s'=0$, $\bar s'=0.$
Because of (\ref{pas1}), we have
\begin{eqnarray*}
\frac{\partial}{\partial s'_{\kk}}c_\kk'\Big|_{\substack{s'=0
    \\ \bar{s'}=0}}&=&0,\\ 
\frac{\partial}{\partial \bar s'_{\kk}} c_\kk'\Big|_{\substack{s'=0
    \\ \bar{s'}=0}}&=&0.
\end{eqnarray*}
Then, for
example, taking the first term of $B$ one gets 
\begin{align*}
\frac{\partial}{\partial s'_{rt, \kk'}}\sum_{\kk}\tau_{\alpha\beta,
  \kk}s_{\alpha\alpha', \kk}\bar{s}_{\beta \alpha',
  \kk}\Big|_{\substack{s'=0 \\ \bar{s'}=0}}
=\tau_{\alpha\beta, \kk'}\tilde{c}_{\alpha r,
  \kk'}\bar{\tilde{s}}_{\beta t, \kk'} ,
\end{align*} 
which equals the first term of $O(\kk)$ at $c=\tilde{c}$ and
$s=\tilde{s}$. Calculating other terms of $B$ one finally gets 
\begin{equation}
\frac{\partial B}{\partial s'}\Big|_{\substack{s'=0
    \\ \bar{s'}=0}}=\frac12O(\kk)|_{\substack{c=\tilde{c}
    \\ s=\tilde{s}}}. \label{zerowanie} 
\end{equation}
Thus the minimizing procedure is equivalent to $O(\kk)=0$ in exactly the same way as it happened in the improved Bogoliubov method in Subsection \ref{improvedbogmethod}.

As in the bosonic case described in the previous section , it turns out this result is a special case of a more general fact discussed in the next section.

Thus, if we choose the Bogoliubov transformation according to the minimization procedure, the Hamiltonian equals
\begin{align}
   H 
&=
B+\sum_\kk
D(\kk)b_\kk^{\dagger}b_\kk
+\hbox{terms  higher order in {\it b}'s}\label{non-}.
\end{align}

In the case of the model interaction considered by Bardeen-Cooper-Schrieffer, described in many texts, e.g. in \cite{FetterWalecka1971},
the minimization of $B$ yields  a dispersion relation that has a positive energy gap and a positive critical velocity uniformly as $L\to\infty$, that is,
\begin{eqnarray}\inf_\kk D(\kk)>0,
&&\inf_{\kk\neq\0}\frac{D(\kk)}{|\kk|}>0.\label{resa4}
\end{eqnarray}
 This phenomenon is probably much more general.  In particular, we expect that it is true for a large class of real, spin-independent and attractive interactions. In the next two subsections we provide computations that seem to support this claim.

Note that the reality and spin-independence of the interactions leads to a considerable computational simplification. By an attractive interaction we mean an interaction, which in some sense, described later on, is negative definite.

Let us assume in addition that higher order terms in (\ref{non-}) are in some sense negligible. Then, \textit{formally}, $H$ is approximated by a quadratic Hamiltonian $B+\sum_\kk
D(\kk)b_\kk^{\dagger}b_\kk$ whose dispersion relation, as we will argue below,  has a strictly positive energy gap and critical velocity (recall Subsection \ref{Excitation spectrum}).

\subsubsection{Reality condition}

Let us first apply the assumption about the reality of the interaction. In this
case, it is natural to assume that the trial vector is real as well. This
means that we impose the conditions
\[\bar c_\kk=c_{-\kk}, \ \ \ \bar s_\kk =s_{-\kk}.\] This allows us to
simplify the formulas for $B$, $D(\kk)$ and $O(\kk)$:
\begin{eqnarray*}\nonumber
B&=&\sum_\kk\tau(\kk) s_\kk\bar s_\kk\nonumber\\
&&+\frac{1}{2L^d}\sum_{\kk,\kk'}
 Q(\kk,-\kk,-\kk',\kk') s_\kk\bar c_\kk
c_{\kk'}
\bar s_{\kk'}\nonumber\\ 
&&+\frac{1}{L^d}\sum_{\kk,\kk'}Q(\kk,\kk',\kk',\kk)\bar s_{\kk}\bar s_{\kk'}s_{\kk'}s_{\kk},\nonumber
\end{eqnarray*}
\begin{eqnarray*}\nonumber
O(\kk)&=&
2\tau(\kk) c_\kk\bar s_\kk\nonumber\\
&&+\frac{1}{L^d}\sum_{\kk'}
 Q(\kk,-\kk,-\kk',\kk')(c_{\kk}\bar c_{\kk}-s_\kk\bar s_\kk)
c_{\kk'}\bar
s_{\kk'}\nonumber\\ 
&&+\frac{4}{L^d}\sum_{\kk'}Q(\kk,\kk',\kk',\kk)c_{\kk}s_{\kk'}
\bar s_{\kk'}\bar
s_{\kk},\nonumber\\[3ex]
\end{eqnarray*}
\begin{eqnarray*}
D(\kk)&=&\tau(\kk)
(c_\kk\bar c_\kk- s_{\kk} \bar s_{\kk})\nonumber\\
&&+\frac{2}{L^d}\sum_{\kk'}
 Q(\kk,-\kk,-\kk',\kk')c_{\kk}\bar s_{\kk}c_{\kk'}\bar
s_{\kk'}\nonumber\\ 
&&+\frac{2}{L^d}\sum_{
\kk'}
 Q(\kk,\kk',\kk',\kk)
(c_{\kk}s_{\kk'}\bar s_{\kk'}\bar
c_{\kk}
-s_\kk s_{\kk'}\bar s_{\kk'}\bar s_\kk).
\nonumber
\end{eqnarray*}

\subsubsection{Spin $\frac12$ case}
Assume that the ``spin space'' is $\CC^2$ and the Hamiltonian is spin
independent. 
We make the BCS ansatz:
\begin{eqnarray*}
c_\kk&=&\cos\theta_\kk\left[\begin{array}{cc}
{1}&{0}\\{0}&{1}\end{array}\right],\\
s_\kk&=&\sin\theta_\kk\left[\begin{array}{cc}
{0}&{1}\\{-1}&{0}\end{array}\right],
\end{eqnarray*}
where, keping in mind the reality condition, the parameters $\theta_{\kk}$ are
real.
Then 
\begin{eqnarray*}
B&=&\sum_\kk\tau(\kk)(1-\cos2\theta_\kk)\\
&&+\frac{1}{4L^d}\sum_{\kk,\kk'}\alpha(\kk,\kk')\sin2\theta_\kk\sin2\theta_{\kk'}\\
&&+\frac{1}{4L^d}\sum_{\kk,\kk'}\beta(\kk,\kk')
(1-\cos2\theta_\kk)(1-\cos2\theta_{\kk'}),\end{eqnarray*}
where
\begin{eqnarray*}
\alpha(\kk,\kk')&:=&\frac12\bigl(
q(\kk,-\kk,-\kk',\kk')+q(-\kk,\kk,-\kk',\kk')\bigr),\\
\beta(\kk,\kk')&=&2q(\kk,\kk',\kk',\kk)
-q(\kk',\kk,\kk',\kk).
\end{eqnarray*}
Note that
\[\alpha(\kk,\kk')=\alpha(\kk',\kk),\ \ \ 
\beta(\kk,\kk')=\beta(\kk',\kk).\]
In particular, in the case of local potentials we have
\begin{eqnarray*}
\alpha(\kk,\kk')&:=&\frac12
\bigl(\hat V(\kk- \kk')+\hat V(\kk+ \kk')\bigr),
\\
\beta(\kk,\kk')&=&2\hat V(\0)-\hat V
(\kk-\kk').
\end{eqnarray*}
We further compute:
\begin{eqnarray*}
O(\kk)&= &\bigl(\delta(\kk)\cos2\theta_{\kk}+
\xi(\kk)\sin2\theta_{\kk}\bigr)\left[\begin{array}{cc}
{0}&{1}\\{-1}&{0}\end{array}\right],\\
D(\kk)& = & (\xi(\kk)\cos2\theta_\kk-\delta(\kk)\sin2\theta_{\kk}) \left[\begin{array}{cc}
{1}&{0}\\{0}&{1}\end{array}\right],
\end{eqnarray*} 
where
\begin{eqnarray*}
\delta(\kk)&=& \frac{1}{2L^d}\sum_{\kk'}\alpha(\kk,\kk')\sin2\theta_{\kk'},\\ 
\xi(\kk) & = &
\tau(\kk)+\frac{1}{2L^d}\sum_{\kk'}\beta(\kk,\kk')(1-\cos2\theta_{\kk'}).
\end{eqnarray*}

We are looking for a minimum of $B$. To this end, we first analyze critical points of $B$. We compute the derivative of $B$:
\[\partial_{2\theta_\kk}B=\delta(\kk)\cos2\theta_\kk+\xi(\kk)\sin2\theta_\kk.\]
The condition $\partial_{2\theta_\kk}B=0$, or equivalently $O(\kk)=0$,
has many solutions. We can have
\begin{eqnarray} \sin2\theta_\kk=0,&&\cos2\theta_\kk=\pm1,\label{pqi1}
\end{eqnarray}
or
\begin{eqnarray*}
\sin2\theta_\kk=-\epsilon_\kk\frac{\delta(\kk)}{\sqrt{\delta^2(\kk)+
\xi^2(\kk)}}\neq0,&&
\cos2\theta_\kk=\epsilon_\kk\frac{\xi(\kk)}{\sqrt{\delta^2(\kk)+
\xi^2(\kk)}},
\end{eqnarray*}
where $\epsilon_\kk=\pm1$. 

In particular, there are many solutions with  all $\theta_\kk$ satisfying (\ref{pqi1}). They correspond to Slater determinants and have a fixed number of particles. The solution of this kind that minimizes $B$ is called the {\em normal} or {\em Hartree-Fock solution}.

One expects that under some conditions the normal solution is not the global minimum of $B$. More precisely, one expects that a global minimum is reached by a configuration satisfying
\begin{eqnarray}
\sin2\theta_\kk=-\frac{\delta(\kk)}{\sqrt{\delta^2(\kk)+
\xi^2(\kk)}},&&
\cos2\theta_\kk=\frac{\xi(\kk)}{\sqrt{\delta^2(\kk)+
\xi^2(\kk)}},\label{resa1}
\end{eqnarray}
where at least some of $\sin 2\theta_\kk$ are different from $0$. It is sometimes called a {\em superconducting solution}. 
In such a case we  get
\begin{eqnarray}
D(\kk)=\sqrt{\xi^2(\kk)+\delta^2(\kk)}\left[\begin{array}{cc}
{1}&{0}\\{0}&{1}\end{array}\right]. \label{oddz}
\end{eqnarray}
Thus we obtain a positive dispersion relation. One can expect that it is strictly positive, since otherwise the two functions $\delta$ and $\xi$ would have  a coinciding zero, which seems unlikely. 
Thus we expect that the dispersion relation $D(\kk)$ has a positive energy gap.

If the interaction is small, then $\xi(\kk)$ is close to $\tau(\kk)$ and $\delta(\kk)$  is small. This implies that $D(\kk)$ is close to $|\tau(\kk)|$.
Thus, if $\tau(\kk)$ has a critical velocity and $D(\kk)$ has an energy gap, this implies that $D(\kk)$ also has a critical velocity.

In other words,  we expect that for a large class of interactions, if the minimum of $B$ is reached at a superconducting state, then $D(\kk)$ satisfies (\ref{resa4}).

We will not study conditions guaranteeing that a superconducting solution minimizes the energy in this  thesis. Let us only remark that such conditions involve some kind of negative definiteness of the quadratic form  $\alpha$ -- this is what we vaguely indicated by saying that the interaction is attractive. Indeed,
 multiply the definition of $\delta(\kk)$ with $\sin2\theta_\kk$ and sum it up over $\kk$. We then obtain
\begin{equation}
\sum_\kk\sin^22\theta_\kk\sqrt{\delta^2(\kk)+\xi^2(\kk)}
=-\frac{1}{2L^d}\sum_{\kk,\kk'}\sin2\theta_\kk\alpha(\kk,\kk')\sin2\theta_{\kk'}
.\label{resa9}\end{equation}
The left hand side of (\ref{resa9}) is positive.
This means that the quadratic form given by the kernel $\alpha(\kk,\kk')$ has to be negative at least at the vector given by $\sin2\theta_\kk$.

Let us also indicate why one expects that the solution corresponding to  (\ref{resa1}) is a  minimum of $B$.
We compute the second derivative:
\begin{eqnarray}\nonumber
\partial_{2\theta_\kk}\partial_{2\theta_{\kk'}}B&=&
\delta_{\kk,\kk'}\bigl(-\sin2\theta_\kk\delta(\kk)+\cos2\theta_{\kk}\xi(\kk)
\bigr)
\\\nonumber
&&+\frac{1}{2L^d}\alpha(\kk,\kk')\cos2\theta_\kk\cos2\theta_{\kk'}\\\label{resa2}
&&+\frac{1}{2L^d}\beta(\kk,\kk')\sin2\theta_\kk\sin2\theta_{\kk'}.
\end{eqnarray}
Substituting (\ref{resa1}) to the first term on the right of (\ref{resa2}) gives
\[\delta_{\kk,\kk'}\sqrt{\delta^2(\kk)+\xi^2(\kk)},\]
which is positive definite. One can hope that  the other two terms in the second derivative of $B$ do not spoil its positive definiteness, especially in the large-volume limit.

\subsubsection{Examples of energy-momentum spectrum of an interacting Fermi gas}\label{interactingfermigraphics}
In Subsection \ref{noninteractinggraphics} we have presented figures with the expected shape of the energy-momentum spectrum of a non-interacting Fermi gas. 

Here we will consider the case of an interacting Fermi gas. Calculations presented in the previous subsections, in particular equation (\ref{oddz}), suggest that the dispersion relation obtained by the HFB method is qualitatively similar to 
\begin{eqnarray*}
\omega(\kk)=\sqrt{(\kk^2-\mu)^2+\gamma^2}.\label{dispersion_oddz}
\end{eqnarray*}
Here, we have chosen the 1-body operator to be of the form
$$\tau(\kk)=\kk^2-\mu.$$
In this way we arrive at figures presented below. As in Subsection \ref{noninteractinggraphics}, in the  full and the odd cases, that is $\sp (H,P)$ and $\sp(H^-,P^-)$, the dispersion relation
$\omega$ is a singular part of the spectrum and it  is be denoted by a solid line. In the even case, $\sp(H^+,P^+)$, the dispersion relation is denoted by a dotted line. Note in particular, that the energy-momentum spectrum presented on these figures possesses properties described in Subsection \ref{twosectorproperties}.

\begin{figure}[!h]
\centering
\includegraphics{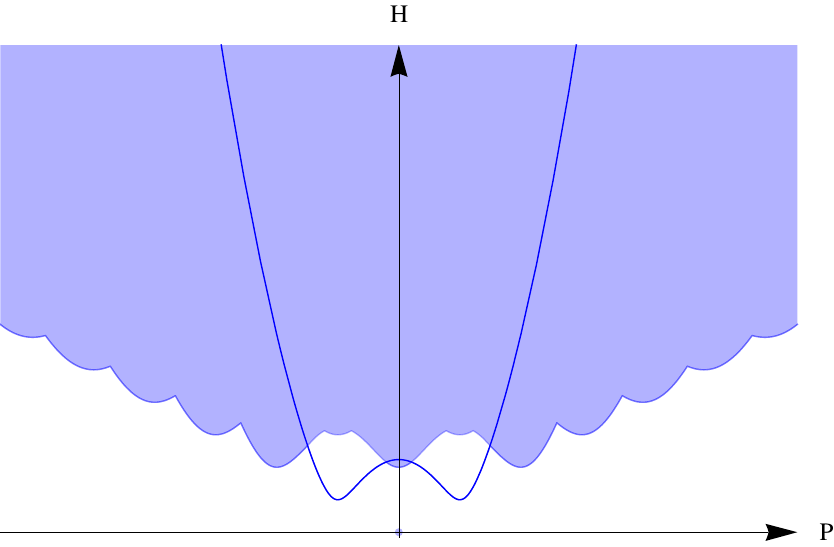}
\label{fig:5}
\caption{$\sp(H,P)$ in the interacting case,  $d=1$.}
\end{figure}

\begin{figure}[!h]
\centering
\includegraphics{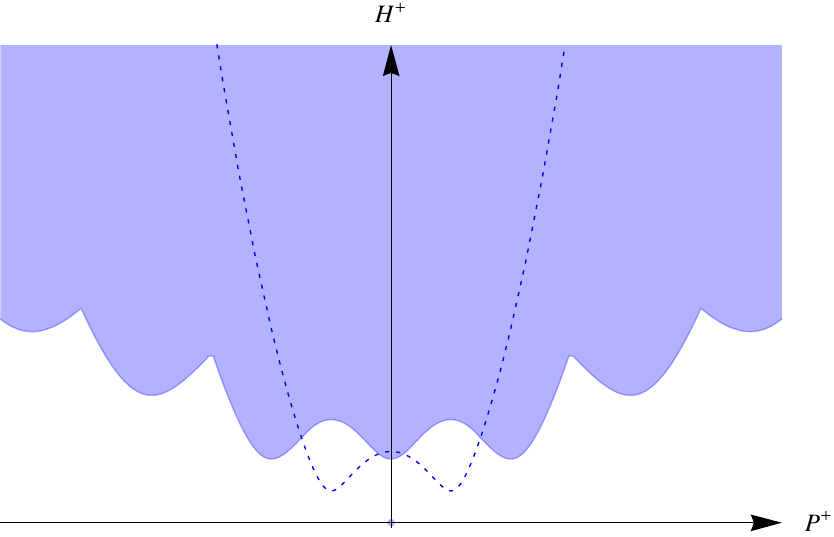}
\label{fig:6}
\caption{$\sp(H^+,P^+)$ in the interacting case, $d=1$.}
\end{figure}

\begin{figure}[!h]
\centering
\includegraphics{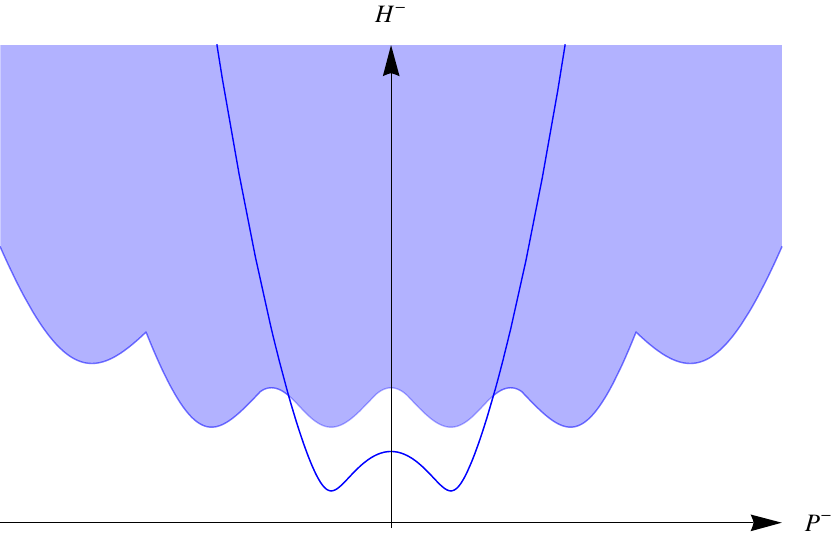}
\label{fig:7}
\caption{$\sp(H^-,P^-)$ in the interacting case, $d=1$.}
\end{figure}

\newpage
\begin{figure}[!h]
\centering
\includegraphics{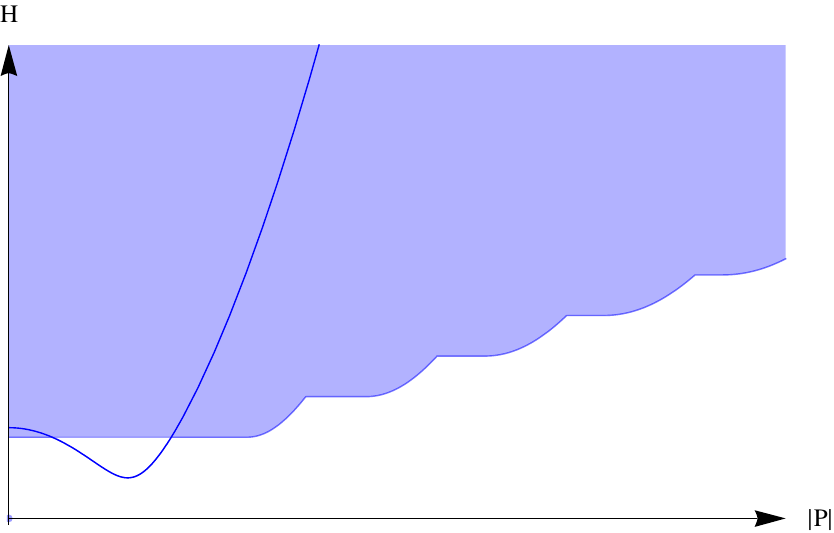}
\label{fig:8}
\caption{$\sp(H,P)$ in the interacting case, $d\geq 2$.}
\end{figure}

\begin{figure}[!h]
\centering
\includegraphics{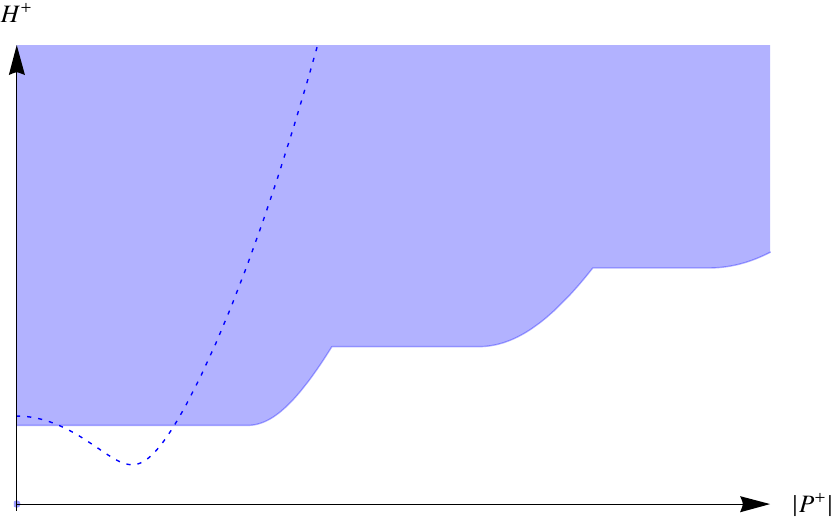}
\label{fig:9}
\caption{$\sp(H^+,P^+)$ in the interacting case, $d\geq 2$.}
\end{figure}

\begin{figure}[!h]
\centering
\includegraphics{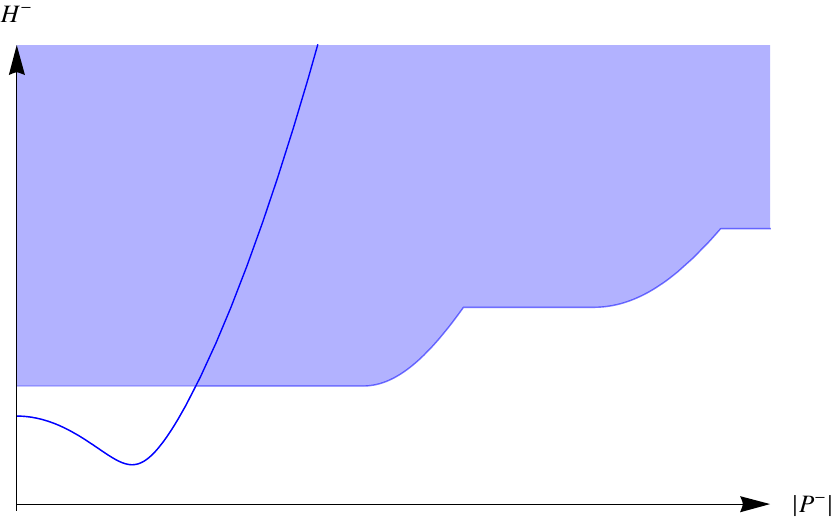}
\label{fig:10}
\caption{$\sp(H^-,P^-)$ in the interacting case, $d\geq 2$.}
\end{figure}

Again, the case $d=1$ differs from
$d\geq2$.  However, in all dimensions
the energy gap and the critical velocity are
 strictly positive.

\section{Beliaev's theorem}\label{beliaevthmsec}

In this section, which is based on \cite{DMS2013}, we shall show how the \textit{improved Bogoliubov method} described in Subsection \ref{improvedbogmethod} and \textit{Hartree-Fock-Bogoliubov approximation} presented in Section \ref{HFBBCS} fit into a more general scheme which leads to a description of the initial system in terms of approximate quasiparticles as discussed in Section \ref{Approximate versus exact quasiparticles}. 

Let us recall briefly that a typical situation when
one speaks about a system with approximate quasiparticles seems to be the following. Suppose that the Hamiltonian of a system can be written as
$H=H_0+V$, where $H_0$ is in some sense  dominant and $V$ can be  neglected. Suppose also that 
\begin{equation}
H_0=B+\sum_i \omega_ib_i^{\dagger}b_i,\label{form}
\end{equation}
where $B$ is a number, operators $b_i^{\dagger}$/$b_i$ 
satisfy the standard CCR/CAR relations and the Hilbert space contains a state annihilated by $b_i$ (the  Fock vacuum for $b_i$).
Then, after subtracting $B$, we obtain a quasiparticle quantum system (recall  Subsection \ref{Concept of a quasiparticle}). One can thus hope, that in a certain approximation (e.g. mean-field limit) the full system has quasiparticle-like excitation spectrum.

The improved Bogoliubov method and the Hartree-Fock-Bogoliubov approximation described in the previous sections led to such a decomposition. Here we shall make it a little bit more general.

Our starting point is a fairly general Hamiltonian $H$ defined on a bosonic or fermionic Fock space. For simplicity we assume that the $1$-particle space is finite
dimensional. With some technical assumptions, the whole picture should be easy to generalize to the infinite dimensional case.  
We assume that the Hamiltonian is a polynomial in creation and annihilation
operators $a_i^{\dagger}$/$a_i$, $i=1,\dots,n$.  This is a typical assumption in Many Body Quantum Physics and Quantum Field Theory. 

As we have seen in the previous sections, an important role in Many Body Quantum Physics is played by the so-called {\em
 Gaussian states}, called also {\em quasi-free states}. Gaussian states can be {\em pure} or {\em mixed}. 
The former are typical for the zero temperature, whereas the latter for positive temperatures. Here, we do not consider
mixed Gaussian states.
 
Pure Gaussian states are obtained by applying  Bogoliubov transformations to the Fock vacuum state (given by the vector 
$\Omega$ annihilated by $a_i$'s). Pure Gaussian states are especially convenient for computations. 

We minimize the expectation value of the Hamiltonian $H$ with 
respect to pure Gaussian states, obtaining a state given by a vector
$\tilde\Omega$. By applying an appropriate Bogoliubov transformation, we can replace the old creation and annihilation operators $a_i^{\dagger}$, $a_i$ by new ones  $b_i^{\dagger}$, $b_i$, which are adapted to the ``new vacuum'' $\tilde\Omega$, i.e., that satisfy $b_i\tilde\Omega=0$. We can rewrite the Hamiltonian $H$ in the new operators and Wick order them, that is, put $b_i^{\dagger}$ on the left and $b_i$ on
the right. The theorem that we prove says that
\[H=B+\sum_{ij} D_{ij}b_i^{\dagger}b_j+V,\]
where $V$ has only terms of the order greater than $2$. In  particular, $H$
does not contain terms of the type $b_i^{\dagger}$, $b_i$,  $b_i^{\dagger}b_j^{\dagger}$, or
$b_ib_j$. It is thus natural to set $H_0:=B+\sum_{ij}
D_{ij}b_i^{\dagger}b_j$. $D_{ij}$ is a hermitian matrix. Clearly, it can be diagonalized, so that $H_0$ acquires the form of (\ref{form}).

We will present several versions of this theorem. First we assume that the Hamiltonian is even. In this case it is natural to restrict the minimization to even pure Gaussian
states. In the fermionic case, we can also minimize over odd pure Gaussian states. In the bosonic case, we consider also Hamiltonians without the evenness assumption, and then we minimize with respect to all pure Gaussian states.

 The fact  that we describe below is probably very well known, at
least on the intuitive level, to many physicists, especially in condensed matter theory. One can probably say that it
 summarizes in abstract terms one of the most widely used methods of contemporary quantum physics. The earliest reference that we know to a statement similar to our main result is formulated in a paper of Beliaev \cite{Beliaev1959}. Beliaev
studied fairly general femionic Hamiltonians by what we would nowadays call the Hartree-Fock-Bogoliubov approximation. In a footnote on page 10 he writes:

\noindent
{\em The condition $H_{20}=0$ may be easily shown to be exactly equivalent to
  the requirement of a minimum ``vacuum'' energy $U$. Therefore, the ground
  state of the system in terms of new particles is a ``vacuum'' state. The
  excited states are characterized by definite numbers of new particles,
  elementary excitations.} 

\noindent Therefore, we propose to call the main result of this part of the thesis {\em Beliaev's Theorem}.

The proof of Beliaev's Theorem is not difficult, especially when it is is formulated in an abstract way,
as we do. Nevertheless, in concrete situations, when similar computations are performed, consequences 
of this result may often appear somewhat miraculous. We witnessed it two times in the previous sections. As we show, these terms have to disappear by a general
argument.

\subsection{Preliminaries}
\subsubsection{2nd quantization}

We will consider in parallel the bosonic and fermionic case. 

Let us briefly recall our notation concerning the 2nd quantization. 
We will always assume that the 1-particle space is $\CC^n$. The {\em bosonic Fock space} will be denoted $\Gamma_{\rm{s}}(\CC^n)$ and the {\em fermionic Fock space}
$\Gamma_{\rm{a}}(\CC^n)$.
We use the notation $\Gamma_{\rm{s}/\rm{a}}(\CC^n)$ for either the bosonic or fermionic
Fock space. $\Omega\in \Gamma_\sa(\CC^n)$ stands for the {\em Fock vacuum}.
If $r$ is an operator on $\CC^n$, then $\Gamma(r)$ stands for its {\em 2nd
quantization}, that is
\[\Gamma(r):=\left(\loplus_{n=0}^\infty r^{\otimes n}\right)\Big|_{\Gamma_\sa(\CC^n)}.\]
$a_i^{\dagger}$, $a_i$ denote the standard {\em creation} and {\em annihilation operators} on
$\Gamma_\sa(\CC^n)$, satisfying the usual canonical
commutation/anticommutation relations.

\subsubsection{Wick quantization} 

Consider an arbitrary  polynomial on $\CC^n$, that is a function of the form
\begin{align}
 h(\bar{z},z):=\sum_{\alpha,\beta} h_{\alpha,\beta}\bar{z}^{\alpha}z^{\beta}, \label{wielomian}
\end{align}
where $z=(z_{1},\ldots, z_{n})\in \mathbb{C}^{n}$, $\bar z$ denotes the
complex conjugate of $z$  and 
 $\alpha=(\alpha_{1},\ldots,\alpha_{n})
 \in (\mathbb{N}\cup\{0\})^{n}$  represent multiindices.
In the bosonic/fermionic case we  always assume that the coefficients
$h_{\alpha,\beta}$ are  symmetric/antisymmetric separately in the
indices of $\bar z$ and $z$.

We write
 $|\alpha|=\alpha_{1}+\cdots+\alpha_{n}.$ 
We say that $h$ is {\em even} if the sum  in  (\ref{wielomian})
is restricted to even
$|\alpha|+|\beta|$ .

{\em The Wick quantization} of (\ref{wielomian})
is  the operator on $\Gamma_\sa(\CC^n)$ defined as
\begin{align}
h(a^{\dagger},a):=\sum_{\alpha, \beta} h_{\alpha,
  \beta}(a^{\dagger})^{\alpha}a^{\beta}. \label{wick} 
\end{align}
In the fermionic case, (\ref{wick}) defines a bounded operator on
 $\Gamma_{\rm{a}}(\CC^n)$. In the bosonic case,  (\ref{wick}) can be viewed as an
operator on $\bigcap\limits_{n>0}\text{Dom} N^n\subset \Gamma_{\rm{s}}(\CC^n)$, where 
\[N=\sum_{i=1}^na_i^{\dagger}a_i\]
is the number operator.

\subsection{Bogoliubov transformations} 

We will now present some basic well known facts
about Bogoliubov transformations. For proofs and additional information we refer to \cite{Berezin1966} (see also \cite{DerGer2013}, \cite{Friedrichs1953}). We will use the 
summation convention of summing with respect to repeated indices.

Operators of the form 
\begin{align}
Q=\theta_{ij}a_{i}^{\dagger}a_{j}^{\dagger}+h_{kl}a_{k}^{\dagger}a_{l}
+\bar{\theta}_{ij}
a_{j}a_{i}\pm \frac12 h_{kk}, 
\label{quadratic}
\end{align}
where $h$ is a self-adjoint matrix, will be called \emph{quadratic Hamiltonians}.
In the bosonic/fermionic case we can always assume  that 
$\theta$ is symmetric/antisymmetric. (The term $\pm\frac12 h_{kk}$, with the
sign depending on the bosonic/fermionic case, means that $Q$ is the
Weyl quantization of the corresponding quadratic expression.)
  
   The group generated by operators of the form $\e^{\ii Q}$, where $Q$ is a
quadratic Hamiltonian, is called the \emph{metaplectic (Mp)} group in the
bosonic case and the \emph{Spin} group in the fermionic case.

In the bosonic case, the group generated by  $Mp$ together with
$\e^{\ii( y_{i}a_i^{\dagger}+\bar{y}_{i} a_i)}$, $y_{i}\in\CC$,
$i=1,\dots,n$, is called the
\emph{affine mataplectic (AMp)}
group.

In the fermionic case, the group generated by operators 
$y_i a_i^{\dagger}+\bar{y}_{i}a_i$ with $\sum |y_i|^2=1$ (which are unitary) is called
the \emph{Pin} group. $Spin$ is a subgroup of $Pin$.

In the bosonic case, consider $U\in\emph{AMp}$. It is well known that
\begin{align}
Ua_{i}U^{\dagger}=p_{ij}a_{j}+q_{ij}a_{j}^{\dagger}+\xi_i, \quad Ua^{\dagger}_{i}U^{\dagger}=\bar{p}_{ij}a^{\dagger}_{j}+\bar{q}_{ij}a_{j}+\bar\xi_i  \label{linear}
\end{align} 
for some matrices $p$ and $q$ and a vector $\xi$. 

In the fermionic case, consider $U\in\emph{Pin}$. Then
\begin{align}
Ua_{i}U^{\dagger}=p_{ij}a_{j}+q_{ij}a_{j}^{\dagger}, \quad Ua^{\dagger}_{i}U^{\dagger}=\bar{p}_{ij}a^{\dagger}_{j}+\bar{q}_{ij}a_{j}  \label{linear1}
\end{align} 
for some matrices $p$ and $q$.

The maps
(\ref{linear}) and (\ref{linear1}) 
 are often called {\em Bogoliubov transformations}.
Bogoliubov transformations can be interpreted as automorphism of the
corresponding {\em classical phase space} (see \cite{DerGer2013} for more details).

\subsection{Pure Gaussian states}
We will use  the term {\em pure state} to denote a normalized vector modulo a phase factor.  In particular, we will distinguish between a pure state and its {\em vector  representative}.

On Fock spaces we have a distinguished pure state called the  {\em (Fock) vacuum state}, corresponding to  $\Omega$. 
States given by vectors of the form $U\Omega,$ where $U\in Mp$ or $U\in Spin$, will be called \emph{even pure Gaussian states}.
 The family of even pure Gaussian states will be denoted by $\fG_{\sa,0}$.

In the bosonic case, states given by vectors
 of the form $U\Omega$ where $U\in AMp$ will be
called \emph{Gaussian pure states}. 
 The family of bosonic pure Gaussian states will be denoted by $\fG_{\s}$.

In the fermionic case, states given by vectors of the form $U\Omega$, where
$U\in Pin$ will be called \emph{fermionic pure Gaussian states}. The family of
fermionic pure Gaussian states is denoted $\fG_\a$.

Fermionic pure Gaussian states that are not even will be called {\em odd  fermionic pure  Gaussian states}. The family of odd fermionic pure Gaussian states is denoted $\fG_{\a,1}$.

One can ask whether pure  Gaussian states have {\em  natural}
 vector representatives
(that is, whether one can naturally fix the phase factor of their vector
 representatives).  In the
bosonic case this is indeed always possible. If $c=[c_{ij}]$ is a symmetric matrix satisfying $\|c\|<1$, then the vector
\begin{equation}
\det(1-c^*c)^{1/4}\e^{\frac12 c_{ij} a_i^{\dagger}a_j^{\dagger}}\Omega\label{bosgau}
\end{equation}
defines a state in $\fG_{\s,0}$ (see \cite{Ruijsenaars1978}). 
If  $\theta=[\theta_{ij}]$ is a symmetric matrix satisfying
$c=\ii\frac{\tanh\sqrt{ \theta\theta^*}}{\sqrt{\theta\theta^*}}\theta$, then
 (\ref{bosgau}) equals
\begin{equation}
 \e^{\ii X_\theta}\Omega \label{bosgau1} 
 \end{equation} 
with
$$X_\theta:=\theta_{ij}a_i^{\dagger}a_j^{\dagger}+\bar\theta_{ij} a_ja_i.$$

 Each state in  $\fG_{\s,0}$ is represented
uniquely as (\ref{bosgau}) (or equivalently as (\ref{bosgau1})).
In particular,
 (\ref{bosgau1}) provides a smooth parametrization of $\fG_{\s,0}$ by symmetric matrices.

The manifold of fermionic even pure Gaussian states is more complicated.
 We will say that a  fermionic even pure 
Gaussian state given by $\Psi$ is {\em
    nondegenerate} if $(\Omega|\Psi)\neq0$ (if it has a non-zero overlap with the vacuum).
 Every nondegenerate fermionic even pure Gaussian state can be
  represented by a vector
\begin{equation}
 \det(1+c^*c)^{-1/4}\e^{\frac12c_{ij} a_i^*a_j^*}\Omega,\label{fergau}\end{equation}
where  $c=[c_{ij}]$ is an antisymmetric matrix.
If 
$\theta=[\theta_{ij}]$ is an antisymmetric matrix satisfying
$c=\ii\frac{\tan\sqrt {\theta\theta^{\dagger}}}{\sqrt{\theta\theta^{\dagger}}}\theta$,
$\|\theta\|<\pi/2$, then
 (\ref{fergau})  equals
\begin{equation} \e^{\ii X_\theta}\Omega \label{fergau1}
\end{equation}
with
\begin{equation*} X_\theta:=\theta_{ij}a_i^{\dagger}a_j^{\dagger}+\bar\theta_{ij} a_ja_i.
\end{equation*}
Vectors (\ref{fergau}) are natural representatives of their states. It can be shown that only nondegenerate fermionic pure Gaussian states possess natural vector representatives.

Not all even fermionic pure Gaussian states are nondegenerate. 
{\em  Slater determinants} with an even non-zero 
number of particles are examples of even
Gaussian pure states that are not 
nondegenerate.

Nondegenerate pure
Gaussian states form  an open dense subset of $\fG_{\a,0}$
 containing the Fock state
(corresponding to $c=\theta=0$).
In particular, (\ref{fergau}) provides
 a smooth parametrization of a neighborhood of the Fock state in 
$\fG_{\a,0}$ by antisymmetric matrices.

The fact that each even bosonic/nondegenerate fermionic pure Gaussian state
can be represented by a vector of the form
(\ref{bosgau})/(\ref{fergau}) goes under the name of the {\em Thouless
  Theorem}. (See \cite{Thouless1960}; this name is used e.g. in the monograph by Ring and Schuck \cite{Ring1980}). The closely
related fact saying that these vectors can be represented in the form 
(\ref{bosgau1})/(\ref{fergau1}) is sometimes called the {\em Ring-Schuck
  Theorem.}

By definition, the group $AMp/Pin$ acts transitively on $\fG_{\sa}$.
In other words, for any $\tilde\Omega
\in \fG_{\sa} $ we can find $U\in AMp/Pin$ such
that $\tilde\Omega=U\Omega$. Such a $U$ is not defined uniquely -- it can be replaced
by $U\Gamma(r)$, where $r$ is unitary on $\CC^n$.

Clearly, if we set 
\begin{align}
b_{i}:=Ua_{i}U^{\dagger}, \quad b^{\dagger}_{i}:=Ua^{\dagger}_{i}U^{\dagger}, \label{proper}
\end{align} 
then $b_i\tilde\Omega=0$,
 $i=1,\dots,n$, and they satisfy the same CCR/CAR as $a_i$,
$i=1,\dots,n$. 
If $h$ is a polynomial of the form (\ref{wielomian}), then we can Wick
quantize it using the transformed operators:
\begin{align}
h(b^{\dagger},b)=\sum_{\alpha, \beta} h_{\alpha, \beta}(b^{\dagger})^{\alpha}b^{\beta}. \nonumber
\end{align}
Obviously, $Uh(a^{\dagger},a)U^{\dagger}=h(b^{\dagger},b)$.

 \subsection{Beliaev's Theorem}
  As explained above, we think that the following result should be called {\em Beliaev's Theorem}. 

\begin{thm} 
Let $h$ be a polynomial on $\CC^n$ and $H:=h(a^{\dagger},a)$ its Wick
quantization. We consider the following functions:
\begin{enumerate} 
\item (bosonic case, even pure Gaussian states)
$
\fG_{\s,0}\ni \Phi\mapsto (\Phi|H\Phi)$;
\item (bosonic case, arbitrary pure Gaussian states)
$\fG_{\s}\ni\Phi\mapsto (\Phi|H\Phi)$;
\item (fermionic case, even pure Gaussian states)
$\fG_{\a,0}\ni\Phi\mapsto (\Phi|H\Phi)$;
\item (fermionic case, odd pure Gaussian states)
$\fG_{\a,1}\ni \Phi\mapsto (\Phi|H\Phi)$.
\end{enumerate}
In (1), (3) and (4) we assume in addition that the polynomial $h$
is even. 

For a vector $\tilde\Omega$ representing a pure Gaussian state, 
let $U\in AMp/Pin$ satisfy $\tilde\Omega=U\Omega$. Set $b_i=Ua_i U^{\dagger}$ and suppose that $ \tilde h$ is the polynomial satisfying
$H=\tilde h(b^{\dagger},b)$.
Then the following statements are equivalent:
\medskip

\noindent (A)
 $\tilde\Omega$ represents a stationary point of the function defined in
(1)--(4); 

\medskip

\noindent (B)
\begin{align}
\tilde{h}(b^{\dagger},b)=B+D_{ij}b^{\dagger}_{i}b_{j}+\text{\emph{terms of higher order
    in $b$'s}}. \nonumber 
\end{align}
\label{main}\end{thm}

\begin{proof}
Let us prove the case (2), which is a little more complicated than the
remaining cases.
Let us fix $U\in AMp$ so that $\tilde\Omega=U\Omega$.
 Clearly, we can write
\begin{align}
H=\tilde{h}(b^{\dagger},b)=B+\bar{K}_{i}b_{i}+K_{i}b_{i}^{\dagger}+O_{ij}b^{\dagger}_{j}b^{\dagger}_{i}+\bar{O}_{ij}b_{i}b_{j}
+ D_{ij}b^{\dagger}_{i}b_{j} + \nonumber \\ 
+ \text{terms of higher order in
  $b$'s}. 
\label{transH}
\end{align}  
We know that in a neighborhood of $\tilde\Omega$ arbitrary pure
 Gaussian states are
parametrized by a symmetric matrix $\theta$ and a vector $y$:
\[\theta\mapsto U \e^{\ii\phi(y)}\e^{\ii X_\theta}\Omega,\] 
where $X_\theta:=\theta_{ij}a_i^{\dagger}a_j^{\dagger}+\bar\theta_{ij}a_ja_i$ and $\phi(y)= y_{i}a_i^{\dagger}+\bar y_{i} a_i$.
We get
\begin{eqnarray}
(U\e^{\ii\phi(y)}\e^{\ii X_\theta}\Omega|HU \e^{\ii\phi(y)}e^{\ii X_\theta}\Omega)&=&(\e^{\ii\phi(y)}\e^{\ii X_\theta}\Omega|U^{\dagger}\tilde h(b^{\dagger},b)U\e^{\ii\phi(y)}\e^{\ii X_\theta}\Omega)\nonumber \\
&=&(\Omega|\e^{-\ii X_\theta}\e^{-\ii\phi(y)}\tilde h(a^{\dagger},a)\e^{\ii\phi(y)}\e^{\ii X_\theta}\Omega). \label{obroty} \end{eqnarray}
Now 
\begin{eqnarray*}
\e^{-\ii X_\theta}\e^{-\ii\phi(y)}\tilde h(a^{\dagger},a)\e^{\ii\phi(y)}\e^{\ii
  X_\theta}
&=&B- \ii(\bar{\theta}_{ij}O_{ij}- \theta_{ij}\bar{O}_{ij})
-\ii(\bar{y}_{i}K_{i}-y_{i}\bar{K}_{i})\\&&+\hbox{terms containing $a_i$ or $a_i^{\dagger}$
  }+O(\|\theta\|^{2},\|y\|^{2}). \end{eqnarray*}
Therefore, (\ref{obroty}) equals
\begin{eqnarray}
B- \ii(\bar{\theta}_{ij}O_{ij}- \theta_{ij}\bar{O}_{ij})
-\ii(\bar{y}_{i}K_{i}-y_{i}\bar{K}_{i})+O(\|\theta\|^{2},\|y\|^{2}).
\label{obrot} 
\end{eqnarray}
Since vectors $y$ and matrices $\theta$ are independent variables,
 (\ref{obrot}) is stationary at $\tilde\Omega$ if and only if $[O_{ij}]$ is a
zero matrix and $[K_{i}]$ is a zero vector. This ends the proof of part (2).

To prove  (3) and (4) we note that,
 for $U\in Pin$, the neighborhood of $\tilde{\Omega}=U\Omega$  in the set of
 fermionic pure Gaussian states is
parametrized by antisymmetric matrices $\theta$:
\[\theta\mapsto U \e^{\ii X_\theta}\Omega,\] 
where again
$X_\theta:=\theta_{ij}a_i^{\dagger}a_j^{\dagger}+\bar\theta_{ij}a_ja_i$. Therefore, it
suffices to repeat the above proof with $y_i=K_i=0$, $i=1,\dots,n$.

The proof of (1) is similar. \end{proof}

\begin{prop} In addition to the assumptions of Theorem \ref{main} (2),
  suppose that
 $\tilde\Omega$ corresponds to a minimum. Then
the matrix $[D_{ij}] $ is positive.
\end{prop} 

\begin{proof}
Using that
  $O$ and $K$ are zero, we obtain
\begin{eqnarray*}
\e^{-\ii\phi(y)}\tilde h(a^{\dagger},a)\e^{\ii\phi(y)}
&=&B
+\bar{y}_{i}D_{ij}y_j\\&&+\hbox{terms containing $a_i$ or $a_i^{\dagger}$
  }+O(\|y\|^{3}). \end{eqnarray*}
Therefore, (\ref{obroty}) equals
\begin{eqnarray}
B+\bar{y}_{i}D_{ij}y_j+O(\|y\|^{3}).
\label{obrot1} 
\end{eqnarray}
Hence the matrix
$[D_{ij}]$ is positive.
\end{proof}

 Note that in cases (1), (3) and (4) the matrix $[D_{ij}]$ does not have
  to be positive.

\part{Rigorous justification of the Bogoliubov approximation}
\section{Formulation and discussion of the main result}\label{mainresult}
\label{s1}
In this part of the thesis, which is based on the article \cite{DN2014}, we shall present a rigorous result concerning the \textit{Bogoliubov approximation} presented in Subsection \ref{The original Bogoliubov approximation}.
 In the following subsections we will formulate the main theorem and explain how it is related to the ideas presented in the first part of the thesis. 
Section \ref{proofba} will be dedicated to the presentation of the proof.   
\subsection{Setup}
Let us state precisely the assumptions on the 2-body potential that we will use
from now on. Consider a real function
$\mathbb{R}^{d}\ni \x\mapsto v
(\x)$, with its Fourier transform defined  by
\[\hat{v}(\p):=\int_{\mathbb{R}^{d}}v(\x)\e^{-\ii\p\x} \d \x.\]
 We assume that $v(\x)=v(-\x)$, and that $v\in L^1(\mathbb{R}^{d})$ and $\hat v\in L^1(\mathbb{R}^{d})$. We also suppose  that the potential is positive and positive definite, i.e.
\[v(\x)\geq0,\ \ \x\in\mathbb{R}^d,\ \ \ \ \hat{v}(\p)\geq 0, \,\,\, \p \in \mathbb{R}^{d}.\]

As explained in Section \ref{themodelbose}, we will consider Bose gas on the torus  $\Lambda=]-L/2,L/2]^{d}$, that is, the $d$-dimensional cubic box of side length $L$ with periodic boundary conditions. We will always assume that $L\geq1$.

The original potential $v$ is replaced by its  {\em periodized} version
\begin{equation}
v^{L}(\x):=\frac{1}{L^d}\sum_{\p\in(2\pi/L)\mathbb{Z}^{d}}\e^{\ii\p\x}\hat{v}(\p). \label{perpot}
\end{equation}
Here, $\p\in(2\pi/L)\mathbb{Z}^{d}$ is the discrete momentum variable.
Note that $v^{L}$ is periodic with respect to the domain $\Lambda$ and that $v^{L}(\x)\rightarrow v(\x)$ as $L\to \infty$. 

Consider the Hamiltonian 
\begin{equation}
-\sum_{i=1}^{N}\Delta^{L}_{i}+\lambda\sum_{1\leq i<j\leq N} v^{L}(\x_{i}-\x_{j}) \label{hamiltonian1}
\end{equation}
acting on the space $L^{2}_{\rm{s}}(\Lambda^{N})$ (the symmetric subspace of $L^{2}(\Lambda^{N})$). The Laplacian is assumed to have periodic boundary conditions. 
\subsection{Mean-field limit}
 Let $\rho=N/L^d$ be the density of the  gas. The Bogoliubov approximation (see Subsection \ref{The original Bogoliubov approximation}) predicts that the ground state energy is
\[\frac12\lambda\rho\hat v(\0)(N-1)-\frac{1}{2}\sum_{\p\in\frac{2\pi}{L}\mathbb{Z}^{d}\setminus \{ 0\}}\left(|\p|^{2}+\rho\lambda\hat{v}(\p)-|\p|\sqrt{|\p|^{2}+2\rho\lambda\hat{v}(\p)}\right) \]
and that the low-lying excited states can  be derived from the
following {\em elementary excitation spectrum}:
\begin{equation}
|\p|\sqrt{|\p|^{2}+2\rho\lambda \hat{v}(\p)}. \label{bogol}
\end{equation}

Note that within the Bogoliubov approximation
both the ground state energy and the excitation spectrum depend on $ \rho$ and $\lambda$ only through the product $\rho\lambda$. The dependence on $L$ is very weak:
\begin{enumerate}
\item The elementary excitation spectrum  (\ref{bogol})  depends on
  $L$ 
 only through the spacing of the momentum lattice $\frac{2\pi}{L}\ZZ^d$.
\item The expression for the ground state energy
divided by the volume $L^d$ converges for $L\to\infty$  to a finite expression
\begin{equation} \frac12\rho^2\lambda\hat v(\0) -\frac{1}{2(2\pi)^d}\int\left(|\p|^{2}+\rho\lambda\hat{v}(\p)-|\p|\sqrt{|\p|^{2}+2\rho\lambda\hat{v}(\p)}\right)\d\p. \label{hammi} \end{equation}
\end{enumerate}

 We believe that it is important to understand the Bogoliubov approximation for large $L$. As pointed out in the first part of this thesis, important physical properties, such as the phonon group velocity and the description of the Beliaev damping in terms of analyticity properties of Green's functions, have an elegant description when we can view the momentum as a continuous variable, which is equivalent to taking the limit $L\to\infty$.

Note that in our problem there are three {\em a priori} uncorrelated parameters: $\lambda$, $N$ and $L$. By the {\em mean field limit} one usually understands $N\to\infty$ with $\lambda\simeq\frac1N$ and $L=\rm{const.}$ However, when  both $N$ and $L$ are large it is natural to consider a somewhat different scaling. Here the mean field limit will correspond to $N\to\infty$ with $\lambda\simeq\frac1\rho=\frac{L^d}{N}$.

Motivated by the above argument we will consider a system described by the Hamiltonian
\begin{equation}
H_N^{L}=-\sum_{i=1}^{N}\Delta^{L}_{i}+\frac{L^d}{N}\sum_{1\leq i<j\leq N} v^{L}(\x_{i}-\x_{j}). \label{hamiltonian}
\end{equation}  
It is translation invariant -- it commutes with the total momentum operator
\begin{equation}
P_N^L:=-\sum_{i=1}^{N}\ii\partial_{\x_i}^L.
\end{equation}

\subsection{Excitation spectrum in the Bogoliubov approximation}
We will denote by $E_N^L$ the ground state energy of the mean-field Hamiltonian \eqref{hamiltonian}.
 If
 $\p\in\frac{2\pi}{L}\ZZ^d\backslash\{\0\}$ let 
$K_N^{L,1}(\p),K_N^{L,2}(\p),\dots$ be the eigenvalues of $H_N^L-E_N^L$
 of total momentum $\p$ in the order of increasing values, counting the multiplicity.
The lowest eigenvalue of $H_N^L-E_N^L$ of total momentum $\p=\0$ is $0$
by general arguments (see Prop. 3.3 in \cite{CDZ2009}). Let
$K_N^{L,1}(\0),K_N^{L,2}(\0),\dots$ be the next eigenvalues of
$H_N^L-E_N^L$ of total momentum $\0$, also in the order of increasing values,  counting the multiplicity.

 We also introduce the {\em Bogoliubov  energy}
 \begin{equation}
  E_\Bog^{L}:=-\frac{1}{2}\sum_{\p\in\frac{2\pi}{L}\mathbb{Z}^{d}\setminus \{ 0\}}\left(|\p|^{2}+\hat{v}(\p)-|\p|\sqrt{|\p|^{2}+2\hat{v}(\p)}\right)  \label{energiabogoliubowa}
  \end{equation}
 and the {\em Bogoliubov elementary excitation spectrum}
\begin{gather}
e_{\p}=|\p|\sqrt{|\p|^{2}+2\hat{v}(\p)}. \label{ep}
\end{gather}
For any $\p\in\frac{2\pi}{L}\zz^d$ we consider the corresponding excitation energies with momentum $\p$:
\begin{equation*}
\left\{\sum_{i=1}^j e_{\kk_i}\ :\ 
\kk_1,\dots,\kk_j\in\frac{2\pi}{L}\ZZ^d\backslash\{\0\}, \quad \kk_1+\cdots+\kk_j=\p,\quad \quad 
j=1,2,\dots\right\}.
\end{equation*}
Let $K_\Bog^{L,1}(\p),K_\Bog^{L,2}(\p),\dots$ be these excitation
energies in the 
order of increasing values, counting the multiplicity. 
 Motivated by the discussion in Sections \ref{excandquasi} and \ref{bosegas}, we will use the term {\em excitation spectrum in the Bogoliubov approximation} to denote the set of pairs $\big(K_\Bog^{L,j}(\p),\p\big)\subset\RR\times\RR^d$.
Later on, we will see that it coincides with the joint spectrum
of commuting operators  $H_\Bog^L-E^L_\Bog$ and $P^L$
with  $(0,\0)$ removed. (See \eqref{stan} for the definition of  $H_\Bog^L$).

Below we present pictures  of the excitation spectrum of 1-dimensional Bose gas in the
Bogoliubov approximation for two potentials, $v_1$ and $v_2$.  Both
potentials are appropriately scaled Gaussians. (Note that Gaussians 
satisfy the assumptions of our main theorem).
On both pictures  the black dot at the origin corresponds to the quasiparticle vacuum, red dots
 correspond to 1-quasiparticle excitations, blue triangles correspond to 2-quasiparticles excitations,
 while green squares correspond to $n$-quasiparticles excitations with $n\geq 3$.
We also give the graphs of the Fourier transforms of both potentials. Figures \ref{rys1bosons} and \ref{rotons} can be seen as bosonic counterparts of the figures at the end of Section \ref{interactingfermigraphics}, but for a system in finite volume.

 Note that all figures are drawn in the same scale, apart from
 Figure \ref{scaledpotential} where the potential had to be scaled down because of 
 space limitations. In our units of length 
 $\frac{2\pi}{L}=\frac{15}{100}$. 

 At this point, let us make a short digression. Looking at Figure \ref{rys1bosons}, one can notice that for some given total momentum $\p$, many-quasiparticle excitation energies are lower than the elementary excitation spectrum. Physically this means that  the corresponding 1-quasiparticle excitation is not stable: it may decay to $m$-quasiparticle states, $m\geq 2$, with a lower energy. This phenomenon has been observed experimentally (\cite{HMHF2001})
and is called the \textit{Beliaev damping}. We believe it is a nice feature of the excitation spectrum (as defined in Section \ref{Excitation spectrum}) that traces of this phenomenon are visible already at this level, even without passing to the macroscopic limit (although the volume of the system has to be sufficiently large).

Recall from the discussion in Section \ref{polepropagator} that quasiparticles with finite lifetime can be described within the framework of Green's functions. If one assumes that the momentum variable is continuous, then Beliaev damping corresponds to a pole of the Green's function on a non-physical sheet of the energy complex plane. The imaginary part of the position of this pole which is responsible for the rate of decay of quasiparticles has been computed by Beliaev (\cite{Beliaev1958}), hence the name of this phenomenon. 

\begin{figure}[H]
\begin{center}
\includegraphics[scale=0.75]{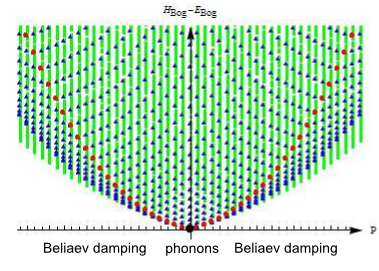}
\caption{Excitation spectrum of  1-dimensional 
homogeneous Bose gas with
potential
$v_1$ in the Bogoliubov approximation.}\label{rys1bosons}
\end{center}
\end{figure}
\begin{figure}[H]
\begin{center}
\includegraphics[scale=0.75]{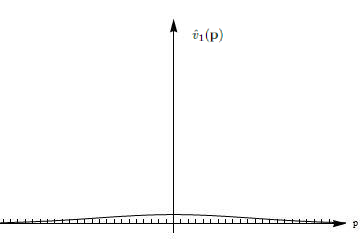}
\caption{$\hat{v}_1(\p)=\frac{\e^{-\p^2/5}}{10}$.} 
\end{center}
\end{figure}

The excitation spectrum for potential $v_2$ (see Figure \ref{rotons}) has a very different shape
-- it has  local maxima and  local minima away from the zero momentum.
 On the  picture we 
 show  traditional names of quasiparticles -- \textit{phonons} in the low momentum region, where the dispersion relation is
 approximately linear,  \textit{maxons} near the local maximum and
 \textit{rotons} near the local minimum of the elementary excitation
 spectrum. For a discussion about the physical meaning of these objects we refer to \cite{Griffin1993}. Let us only note that the definition of a qusiparticle quantum system introduced in Section \ref{Concept of a quasiparticle} is consistent with this terminology since it allows for a situation where some quasiparticles exist only for some momenta.
\begin{figure}[H]
\begin{center}
\includegraphics[scale=0.75]{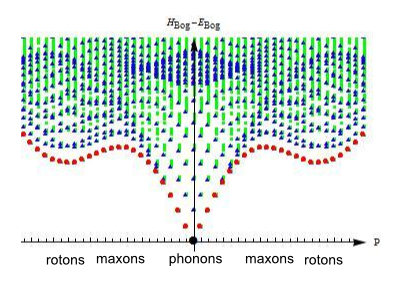}
\caption{Excitation spectrum of  1-dimensional 
homogeneous Bose gas with  potential $v_2$ in the Bogoliubov approximation.}
\label{rotons}
\end{center}
\end{figure}
\begin{figure}[H]
\begin{center}
\includegraphics[scale=0.75]{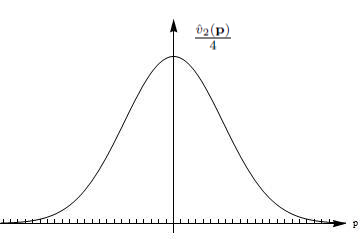}
\caption{$\hat{v}_2(\p)=\frac{15\e^{-\p^2/2}}{2}$.} 
\label{scaledpotential}
\end{center}
\end{figure} 
 
\subsection{Main result} 
 
In this Section we state our main result. It takes a slightly different form for the upper and the lower bound. From now on we drop the superscript $L$. 

\begin{thm}\label{thm}
\begin{enumerate}
\item Let $c>0$. Then there exists  $C$ such that
\begin{enumerate}
\item  if 
\begin{equation} L^{2d+2}\leq cN,\label{war1a} \end{equation}
 then 
\begin{equation} E_N\geq
\frac{1}{2}\hat{v}(\0)(N-1)+E_\Bog-CN^{-1/2}L^{2d+3}; \label{thmnier1}
\end{equation}
\item if in addition 
\begin{equation} K_N^j(\p)\leq c NL^{-d-2},\label{war1b}\end{equation} then
\begin{eqnarray}
E_N+K_N^j(\p)&\geq&
\frac{1}{2}\hat{v}(\0)(N-1)+E_\Bog+K_\Bog^j(\p)\notag \\&&
-CN^{-1/2}L^{d/2+3}\big(K_N^j(\p)+L^d\big)^{3/2}.  \label{thmnier2}
\end{eqnarray}
\end{enumerate}
\item  Let $c>0$. Then there exists  $c_1>0$ and $C$ such that
\begin{enumerate}
\item if 
\begin{eqnarray}
 L^{2d+1}&\leq &cN \label{war2a}\\
\hbox{ and }\hspace{4ex} L^{d+1}&\leq &c_1 N,\label{war2a+}
\end{eqnarray}
 then
\begin{equation} E_N\leq
\frac{1}{2}\hat{v}(\0)(N-1)+E_\Bog+CN^{-1/2}L^{2d+3/2}; \label{thmnier3}\end{equation}
\item  if in addition 
\begin{eqnarray}
 K_\Bog^j(\p)&\leq &cN L^{-d-2}\label{war2b}\\
 \hbox{ and }\hspace{4ex} K_\Bog^j(\p)&\leq &
c_1NL^{-2},\label{war2b+}\end{eqnarray}
 then
\begin{eqnarray}
E_N+K_N^j(\p)&\leq&
\frac{1}{2}\hat{v}(\0)(N-1)+E_\Bog+K_\Bog^j(\p)\notag \\&&
+CN^{-1/2}L^{d/2+3}(K_\Bog^j(\p)+L^{d-1})^{3/2}.
\label{thmnier4}
\end{eqnarray}
\end{enumerate}
\end{enumerate}
\end{thm}

 Let us stress that the constants $C$ and $c_1$ that appear in the theorem  depend only the potential $v$, the dimension $d$, and the constant $c$, but {\em do not} depend on $N$, $j$ and $L$. Note also that both in the first (1), resp. (2) part of the theorem we can deduce (a) from (b) by setting $K_N^j(\p)=0$, resp. $K_\Bog^j(\p)=0$.

Theorem \ref{thm}  expresses the idea that the Bogoliubov approximation becomes exact for large $N$ and $L$ provided that the volume does not grow too fast. This may appear not very transparent, since the error terms in the theorem depend on two parameters $L$ and $N$ as well as on the excitation energy. Therefore, we present some consequences of our theorem, where the error term depends only on $N$. They  generalize the corresponding remarks of \cite{Seiringer2011}.

\begin{cor}\label{coro}
Let $b>1$, $-1-\frac{1}{2d+1}< \alpha \leq 1$ and $L^{4d+6}\leq b N^{1-\alpha}$. Then there exists $M$ such that if $N>M$, then \vspace{2ex}
\begin{enumerate}
\item \hspace{7ex} $E_N=\frac{1}{2}\hat v(\0)(N-1)+E_\Bog+O(N^{-\alpha/2})$;
\vspace{2ex}
\item if $\min\left(K_N^j(\p),K_\Bog^j(\p)\right)\leq (bN^{1-\alpha}L^{-d-6})^{1/3}$, then
\[K_N^j(\p)= K_\Bog^j(\p)+O(N^{-\alpha/2});\]
\item 
  if  $0<\alpha\leq1$ and
 $\min\left(K_N^j(\p),K_\Bog^j(\p)\right)\leq bN^{1-\alpha}L^{-d-6}$, then
 \[ K_N^j(\p)= K_\Bog^j(\p)+\left(1+K_\Bog^j(\p)\right)O(N^{-\alpha/2}).\]
\end{enumerate}
\end{cor}
 The proof that Thm \ref{thm} implies  Cor. \ref{coro} will be given in Subsection \ref{proof of coro}.

\begin{remark}
\begin{enumerate}
\item The case $\alpha=1$, $L=1$ of Corollary \ref{coro}
corresponds directly to the result of \cite{Seiringer2011}.
\item In part (3) of Corollary \ref{coro} one can also include the case $\alpha=0$ provided that $L$ is sufficiently large.
\end{enumerate}
\end{remark}

Thus, for large $N$ within a growing range of the volume, the low-lying
energy-momentum spectrum of the homogeneous Bose gas is well described
by the Bogoliubov approximation. 

In the infinite-volume limit  momentum 
 becomes a continuous variable, which is important when we want to consider the so-called \textit{critical velocity} and \textit{phase velocity} introduced by Landau. They play a crucial role in his theory of superfluidity (\cite{Landau1941}, \cite{Landau1947}, see also \cite{CDZ2009},\cite{ZagBru2001}).
 

 Mathematically, the Bogoliubov approximation has been studied mostly in the context of the ground state energy (\cite{LS2001}, \cite{LS2004}, \cite{ESY2008}, \cite{GuSei2009}, \cite{Sol2006} \cite{YY2009}, see also \cite{LSSY2005}). This makes the work of Seiringer (\cite{Seiringer2011}), Grech-Seiringer (\cite{GreSei2013}) and more recently by Lewin, Nam, Serfaty and Solovej (\cite{LNSS2014}) even more notable, since they are devoted to a rigorous study of the excitation spectrum of a Bose gas.

  In \cite{Seiringer2011} Seiringer proves that for a system of $N$ bosons on a flat unit torus $\mathbb{T}^{d}$ which interact via a  two-body potential $v(\x)/(N-1)$, the excitation spectrum up to an energy $\kappa$ is formed by elementary excitations of momentum $\p$ with a corresponding energy of the form
(\ref{bogol})
up to an error term  of the order $O(\kappa^{3/2}N^{-1/2})$. Also in \cite{GreSei2013} and \cite{LNSS2014} the authors are concerned with  finite systems in the large particle number limit.

Our result can be considered as an extension of Seiringer's result to systems of arbitrary volume. Motivated by the discussion in the first part of this thesis, the ultimate goal would be to prove similar results in the thermodynamic limit with a fixed coupling constant. Since at the moment this is out of reach, we consider other limits, which involve convergence of the volume to infinity.

\subsection{Outline of the proof}

The rest of this thesis is devoted to the presentation of the proof of Theorem \ref{thm} and Corollary \ref{coro}. Let us briefly outline the main ideas behind that proof.

There are two main difficulties that one needs to overcome when proving the main result of this thesis. The first one refers to the fact that one needs to compare spectral properties of two Hamiltonians which are a priori defined on different spaces: the Hamiltonian \eqref{secondbogapprox} obtained by Bogoliubov within his approximation "lives and acts" in the full Fock space (because it contains terms which are not particle-number-conserving) while the mean-field Hamiltonian \eqref{hamiltonian} is particle-number-conserving and thus acts only on fixed-particle-number subspaces of the Fock space.

The second main problem is related to the \textit{c-number substitution} described in Subsection \ref{The original Bogoliubov approximation}. According to this substitution one neglects the non-commutativity of the operators $a^{\dagger}_0$ and $a_0$.

We deal with the second problem by introducing the so-called  {\em extended Fock space} which contains non-physical states with a negative number of zero-momentum modes (see Section \ref{extendedspacesection}). This allows us to treat the zero-momentum mode in a special way.

To deal with the first problem, motived by Bogoliubov's introduction of the operators \eqref{boperators}, we introduce their counterparts adapted to the extended Fock space. In particular, they allow us to simplify the
algebraic computations involved in the proof.

 Our proof uses partly methods presented in \cite{Seiringer2011}. Note, however, that naive mimicking leads to a much weaker result, which involves assuming that $N\geq C \e^{cL^{d/2}}$ to ensure that the error terms tend to zero when taking the infinite-volume limit. This can be easily seen by looking for example at equation (24) of \cite{Seiringer2011}. In this equation one of the constants is given by the expression $\e^{C_{2}}$ where $C_{2}$ is given by 
$$\sqrt{\frac{64N}{N-1}\sum_{\p\neq\0}\beta^{2}_{\p}}.$$
In the infinite-volume limit the sum in the above expression can be replaced by a Riemann integral at the cost of a factor $L^{d/2}.$  This leads to a factor $\e^{c L^{d/2}}$ in the estimates.

Our method  leads to relatively simple
algebraic calculations, which is helpful when we want to control the
volume dependence. Note also  that our 
method yields the same results as in \cite{Seiringer2011} if one takes $L=1$. 

Interestingly, we have never seen the method of the extended space in the literature. We believe it might be useful for further applications.


\section{Proofs of the theorem and corollary}\label{proofba}
\subsection{Miscellanea}

Let us describe some notation and basic facts from operator theory
used below.

If $A$, $B$ are operators, then the following inequality will be often used:
\begin{equation*}
 -A^\dagger A-B^\dagger B\leq A^\dagger B+B^\dagger A\leq A^\dagger A+B^\dagger B.
 \end{equation*}

We will write $A+\hc$ for $A+A^\dagger $.

If $A$ is a self-adjoint operator and $\Omega$ a Borel subset of the spectrum of $A$, then $\one_\Omega(A)$ will denote the spectral projection of $A$ onto  $\Omega$.

Let $A$ be a bounded from below self-adjoint operator on  Hilbert space $\cH$.  For simplicity, let us assume that it has only discrete spectrum.
We define
\[\vecsp(A):=(E_1,E_2,\dots),\]
where $E_1,E_2,\dots$ are the eigenvalues of $A$ in the  order of
increasing values, counting the multiplicity. If $\dim\cH=n$, then we set $E_{n+1}=E_{n+2}=\dots=\infty$.

We will use repeatedly two consequences of the {\em min-max principle} \cite{RS1978}: 
\[A\leq B \ \hbox{ implies }\vecsp(A)\leq \vecsp(B),\]
and
the so-called {\em Rayleigh-Ritz principle}:
If $\cK$ is a closed subspace of $\cH$, let
 $P_\cK$ be the projection onto $\cK$. Then we have 
\begin{equation*}
\vecsp(A)\leq\vecsp\Big(P_\cK AP_\cK\Big|_\cK\Big).
\end{equation*}

\subsection{Second quantization}

The Hamiltonian $H_N$ is defined on the $N$-particle bosonic space
\[\cH_N:
=L^{2}_{\s}(\Lambda^{N}).\] 
We will work most of the time in the momentum representation, in which the 1-particle space $L^2(\Lambda)$ 
is represented as $l^2\bigl(\frac{2\pi}{L}\ZZ^d\bigr)$, thus
\[\cH_N\simeq
\otimes_{\s}^Nl^2\bigl(\frac{2\pi}{L}\ZZ^d\bigr).\] 

It is convenient to consider simultanously the direct sum of the
$N$-particle spaces, the {\em bosonic Fock space}
\begin{equation}
\mathcal{H}:=\mathop\oplus\limits_{N=0}^\infty\cH_N
=\Gamma_\s\Big(l^2\bigl(\frac{2\pi}{L}\ZZ^d\bigr)\Bigr).
\label{physicalspace1}
\end{equation}

The direct sum of the Hamiltonians $H_N$ 
will be denoted $H$. Using  the notation of the second quantization it can be written (recall \eqref{hamboseboxpedpart2}) as
\[H:=\mathop\oplus\limits_{N=0}^\infty H_N=\sum_\p \p^2a_\p^\dagger a_\p+\frac{1}{2N}
\sum_{\p,\q,\kk}\hat{v}(\kk)a_{\p+\kk}^\dagger a_{\q-\kk}^\dagger a_\q a_\p.\]


Let us introduce some special notation for various operators and their 2nd quantization. 

Let $P$ be the projection onto the constant function in $L^{2}(]L/2,L/2]^d)$, and $Q=1-P$. The operator that counts the number of particles in, resp. outside the zero momentum mode will be denoted by $N_0$, resp. $N^{>}$, i.e.
\begin{gather}
N_0=\sum_{i=1}^{N}P_{i},\ \ \ 
N^{>}=\sum_{i=1}^{N}Q_{i}. \label{particlenoop}
\end{gather}
In the 2nd quantization notation,
\[N_0=a_\0^\dagger a_\0,\ \ \ N^>=\sum_{\p\neq\0}a_\p^\dagger a_\p.\]

For $N$-particle bosonic wave functions $\Psi,\Phi$, due to their symmetry, we have
\begin{eqnarray}
\langle\Psi|N^{>}|\Phi\rangle&=&N\langle\Psi|Q_{1}|\Phi\rangle,\label{Q1}\\ 
\langle\Psi|N^{>}(N^>-1)|\Phi\rangle&=&N(N-1)\langle\Psi|Q_{1}Q_2|\Phi\rangle.  \label{Q1Q2}
\end{eqnarray}

The symbol $T$ will denote the kinetic energy of the system: $T=-\sum\limits_{i=1}^{N}\Delta_{i}$. For further reference, note that
\begin{equation*}
\langle\Psi|N^{>}|\Psi\rangle \leq \frac{L^{2}}{(2\pi)^{2}}\langle\Psi|T|\Psi\rangle
. \label{lemma1b-}  
\end{equation*}

\subsection{Bounds on interaction}

The potential $v$ can be interpreted as an operator of multiplication
by $v(\x_1-\x_2)$ on $L_\s^2(\Lambda^2)$.
Following \cite{Seiringer2011}, we would like to estimate this 2-body operator by simpler, 1-body operators. As a preliminary step we record the following bound:

\begin{lemma} Let $\epsilon>0$. Then
\begin{eqnarray*}
v&\geq& P\otimes PvP\otimes P+P\otimes PvQ\otimes Q+Q\otimes QvP\otimes P \\&&+(1-\epsilon)(P\otimes Q+Q\otimes P)v(P\otimes Q+Q\otimes P)\\&&+(1-\epsilon^{-1})Q\otimes QvQ\otimes Q,\\
v&\leq& P\otimes PvP\otimes P+P\otimes PvQ\otimes Q+Q\otimes QvP\otimes P \\&&+(1+\epsilon)(P\otimes Q+Q\otimes P)v(P\otimes Q+Q\otimes P)\\&&+(1+\epsilon^{-1})Q\otimes QvQ\otimes Q.
\end{eqnarray*} \label{rhs}
\end{lemma}
\begin{proof} Using the translation invariance of $v$ we obtain
\begin{eqnarray*}
v&=& (P\otimes P+Q\otimes Q)v (P\otimes P+Q\otimes Q)
 \\&&+(P\otimes Q+Q\otimes P)v(P\otimes Q+Q\otimes P)
\\&&+(P\otimes Q+Q\otimes P)vQ\otimes Q
+Q\otimes Qv(P\otimes Q+Q\otimes P).
\end{eqnarray*} Then we apply the Schwarz inequality to the last two terms.
\end{proof}

Let us now identify the second quantization of various terms on the r.h.s of the estimates of Lemma \ref{rhs}. For instance, let us have a look at the term involving $P\otimes QvQ\otimes P$. To this end recall the notion of the second quantization of certain 2-body operators. More precisely, let $w$ be an operator on the symmetrized 2-particle space. Then by its second quantization we will mean the operator that restricted to the 
$N$-particle space equals
\[\sum_{1\leq i<j\leq N} w_{ij}.\]
If $w$ is an operator on the unsymmetrized 2-particle space, then we can also speak about its second quantization, but now its restriction to
the 
$N$-particle space equals
\[\frac12\sum_{1\leq i\neq j\leq N} w_{ij}.\]
In the momentum
 basis  this operator written in the 2nd quantized language equals
\[\frac12\sum_{\p_1,\p_2,\p_3,\p_4}\langle \p_1,\p_2|w|\p_3,\p_4\rangle
a_{\p_1}^\dagger
a_{\p_2}^\dagger
a_{\p_3}
a_{\p_4}.
\]
Taking for $w$ the operator $P\otimes QvQ\otimes P$ we calculate
\begin{eqnarray*}
&&\frac12\sum_{\p_1,\p_2,\p_3,\p_4}\langle \p_1,\p_2|P\otimes QvQ\otimes P|\p_3,\p_4\rangle a_{\p_1}^\dagger
a_{\p_2}^\dagger
a_{\p_3}
a_{\p_4} = \\
&&\sum_{\p_1,\p_2,\p_3,\p_4}\int_{\Lambda}\int_{\Lambda}(P\e^{-\ii \p_1 \x})(Q\e^{-\ii \p_2 \y})  \frac{v(\x-\y)}{2L^{2d}}(Q\e^{-\ii \p_3 \y})(P\e^{-\ii \p_4 \x})\d \x \d \y a_{\p_1}^\dagger
a_{\p_2}^\dagger
a_{\p_3}
a_{\p_4} =\\
&&\sum_{\p,\p_1,\p_2,\p_3,\p_4}\int_{\Lambda}\int_{\Lambda}(P\e^{-\ii \p_1 \x})(Q\e^{-\ii \p_2 \y}) \e^{\ii\p(\x-\y)}\frac{\hat{v}(\p)}{2L^{3d}}(Q\e^{-\ii \p_3 \y})(P\e^{-\ii \p_4 \x})\d \x \d \y a_{\p_1}^\dagger
a_{\p_2}^\dagger
a_{\p_3}
a_{\p_4}, 
\end{eqnarray*}
where in the last step we used \eqref{perpot}. Since $P$ is a projection onto the constant function, this implies $\p_1=\p_4=\0$. Since $Q=1-P$, we have $\p_2\neq \0$ and $\p_3\neq \0$. Thus we obtain 
\begin{eqnarray*}
\frac12\sum_{\p_1,\p_2,\p_3,\p_4}\langle \p_1,\p_2|P\otimes QvQ\otimes P|\p_3,\p_4\rangle a_{\p_1}^\dagger
a_{\p_2}^\dagger
a_{\p_3}
a_{\p_4} = \\
\sum_{\p,\p_2\neq\0,\p_3\neq \0}\int_{\Lambda}\int_{\Lambda}\e^{-\ii \p_2 \y} \e^{\ii\p(\x-\y)}\frac{\hat{v}(\p)}{2L^{3d}}\e^{-\ii \p_3 \y}\d \x \d \y a_{\0}^\dagger
a_{\p_2}^\dagger
a_{\p_3}
a_{\0}.
\end{eqnarray*}
Using $\frac{1}{L^d}\int_{\Lambda}\e^{\ii \p \x}=\delta_{\p,\0}$ it leads to
\begin{eqnarray*}
\frac12\sum_{\p_1,\p_2,\p_3,\p_4}\langle \p_1,\p_2|P\otimes QvQ\otimes P|\p_3,\p_4\rangle a_{\p_1}^\dagger
a_{\p_2}^\dagger
a_{\p_3}
a_{\p_4} = \frac{1}{2L^d}\sum_{\p\neq\0}\hat{v}(\p)a_{\p}^\dagger
a_{\p}
a_{\0}^\dagger
a_{\0}. 
\end{eqnarray*}
We thus identify the second quantization of various terms on the r.h.s of the estimates of Lemma \ref{rhs}:
\begin{eqnarray*}
 P\otimes PvP\otimes P&&
\frac{1}{2L^d}\hat v(\0)a_\0^\dagger a_\0^\dagger a_\0a_\0=
\frac{1}{2L^d}\hat v(\0)N_0(N_0-1),\\
P\otimes PvQ\otimes Q&&
\frac{1}{2L^d}\sum_{\p\neq\0} \hat v(\p)a_0^\dagger a_0^\dagger a_\p a_{-\p},\\
Q\otimes QvP\otimes P&&
\frac{1}{2L^d}\sum_{\p\neq\0} \hat v(\p)a_\p^\dagger a_{-\p}^\dagger a_0 a_0,\\
P\otimes QvQ\otimes P,\ Q\otimes PvP\otimes Q
&&
\frac{1}{2L^d}\sum_{\p\neq\0} \hat v(\0)a_\p^\dagger  a_{\p}N_0,\\
P\otimes QvP\otimes Q,\ Q\otimes PvQ\otimes P
&&
\frac{1}{2L^d}\sum_{\p\neq\0} \hat v(\p)a_\p^\dagger  a_{\p}N_0
.
\end{eqnarray*} 
The second quantization of $ Q\otimes QvQ\otimes Q$
 can be bounded  from above by
\[v(\0)\sum_{1\leq i<j\leq  N}Q_iQ_j=v(\0)\frac12N^>(N^>{-}1).\]

Introduce the family of {\em estimating Hamiltonians}
\begin{eqnarray*}
H_{N,\epsilon}&:= &\frac{1}{2}\hat v(\0)(N-1)+\sum_{\p\neq\0}\big(|\p|^2+\hat v(\p)\big)
a_\p^\dagger a_\p\\
&&+\frac{1}{2N}\sum_{\p\neq\0}\hat v(\p)\Big(a_\0^\dagger a_\0^\dagger a_\p a_{-\p}+a_\p^\dagger a_{-\p}^\dagger a_\0 a_\0 \Big)\\
&&-\frac{1}{N}\sum_{\p\neq\0}\big(\hat v(\p)+\frac{\hat v(\0)}{2}\big)a_\p^\dagger a_\p N^>+\frac{\hat v(\0)}{2N}N^>\\
&&+\frac{\epsilon}{N}\sum_{\p\neq\0}\big(\hat v(\p)+\hat v(\0)\big)a_\p^\dagger a_\p N_0\\
&&+(1+\epsilon^{-1})\frac{1}{2N} v(\0)L^d N^>(N^>-1).
\end{eqnarray*}

The operators $H_{N,\epsilon}$ preserve the $N$-particle sectors. By the above calculations  we obtain the following estimates on the Hamiltonian:
\begin{eqnarray}
H_N&\geq &H_{N,-\epsilon}, \ \  0<\epsilon\leq1;\label{hamiltonianboundlower}  \\
H_N&\leq &H_{N,\epsilon},\ \ 0<\epsilon. \label{hamiltonianboundupper}
\end{eqnarray}

\subsection{Extended space}\label{extendedspacesection}
So far we used the {\em physical  Hilbert space} (\ref{physicalspace1}). By the exponential property of Fock spaces we have the identification
\begin{equation}
\mathcal{H}\simeq
\Gamma_\s(\CC)\otimes\Gamma_\s\Big(l^2\bigl(\frac{2\pi}{L}\ZZ^d\backslash\{\0\}\bigr)\Bigr).
\label{physicalspace}
\end{equation}
Let us embed the space of  {\em zero modes} $\Gamma_\s(\CC)=l^2(\{0,1,\dots\})$ in a larger space $l^2(\ZZ)$. Thus we obtain the {\em extended Hilbert space}
\begin{equation}
\mathcal{H}^{\text{ext}}:= l^2(\ZZ)\otimes \Gamma_\s\Big(l^2\bigl(\frac{2\pi}{L}\ZZ^d\backslash\{\0\}\bigr)\Bigr).
\label{fullspace}
\end{equation} 
The physical space (\ref{physicalspace}) is spanned by vectors of the
form $|n_0\rangle\otimes \Psi^>$, where
 $|n_0\rangle$ represents $n_0$  zero modes  ($n_0\geq0$)
and $ \Psi^>$ represents a vector outside the zero mode.

The space (\ref{fullspace}) is also spanned by vectors  of this  form, where now 
 the relation $n_0\geq0$ is not imposed. 
The orthogonal complement of $\cH$ in $\cH^\ext$ will be denoted by $\cH^\nph$ (for ``non-physical''). 

 On $\cH^\ext$ we have a self-adjoint operator $N_0^\ext$ such that $N_0^\ext|n_0\rangle\otimes \Psi^>=n_0|n_0\rangle\otimes \Psi^>$. Its spectrum equals $\ZZ$. 
 Clearly
\[N_0^\ext\Big|_\cH=N_0,\ \ \ \cH=\Ran \one_{[0,\infty[}(N_0^\ext),\  \ 
\cH^\nph=\Ran \one_{]-\infty,0[}(N_0^\ext).\]
If $N\in\ZZ$, we will write $\cH_N^\ext$ for the subspace of $\cH^\ext$ corresponding to $N^>+N_0^\ext=N$.

We have also a unitary operator
\begin{gather*}
U|n_0\rangle\otimes\Psi^>=|n_0-1\rangle\otimes\Psi^>.
\intertext{Notice that both $U$ and $U^{\dagger}$ commute with both $a_{\p}$ and $a^{\dagger}_{\p}$ with $\p\neq\0$.  We now define for $\p\neq\0$ the following operator on $\mathcal{H}^{\text{ext}}$:}
{b}_{\p}:=a_{\p}U^{\dagger}. 
\end{gather*}
Operators $b_\p$ and $b_\q^\dagger$ satisfy the same CCR as $a_\p$ and $a_\q^\dagger$.

The extended space is useful in the study of $N$-body Hamiltonians.
To illustrate this, on $\cH_N^\ext$ let
us introduce the {\em extended Hamiltonian}
\begin{eqnarray*}
H_N^\ext&=&\sum_{\p\neq\0}\Big(\p^2+\frac{N_0^\ext}{N}(\hat{v}(\p)+\hat{v}(\0)\Big)b_\p^\dagger
b_\p\\&&
+
\frac{1}{2}\sum_{\p\neq\0}\Big(\hat{v}(\p)\frac{\sqrt{N_0^\ext(N_0^\ext
    -1)}}{N}
b_\p
b_{-\p}
+\hc\Big)\\
&&+\frac{1}{N}\sum_{\kk,\p\neq\0}\hat{v}(\kk)b_{\kk}^\dagger b_{\p-\kk}^\dagger
b_\p\sqrt{\max(N_0^\ext,0)}+\hc\Big)\\
&&+\frac{1}{2N}
\sum_{\p,\q,\kk\neq\0}\hat{v}(\kk)b_{\p+\kk}^\dagger b_{\q-\kk}^\dagger
b_\q b_\p.
\end{eqnarray*}
It is easy to see that $H_N^\ext$ preserves the $N$-particle physical
space $\cH_N$ and on $\cH_N$  it coincides with $H_N$.

Here we will use the {\em extended estimating Hamiltonian}, which is
the following operator on $\cH_N^\ext$:
\begin{eqnarray*}
H_{N,\epsilon}^\ext&:= &\frac{1}{2}\hat v(\0)(N-1)+\sum_{\p\neq\0}\big(|\p|^2+\hat v(\p)\big)
b_\p^\dagger b_\p\\
&&+\frac{1}{2}\sum_{\p\neq\0}\hat v(\p)\Big(\frac{\sqrt{(N_0^\ext-1)N_0^\ext}}{N}b_\p b_{-\p}+\hc
\Big)\\&&-\frac{1}{N}\sum_{\p\neq\0}\big(\hat v(\p)+\frac{\hat v(\0)}{2}\big)b_\p^\dagger b_\p N^>+\frac{\hat v(\0)}{2N}N^>\\
&&+\frac{\epsilon}{N}\sum_{\p\neq\0}\big(\hat v(\p)+\hat v(\0)\big)b_\p^\dagger b_\p N_0^\ext\\
&&+(1+\epsilon^{-1})\frac{1}{2N} v(\0)L^d N^>(N^>-1).
\end{eqnarray*}
Note that $H_{N,\epsilon}^\ext$ 
preserves $\cH_N$ and restricted to $\cH_N$ coincides with $H_{N,\epsilon}$.

\subsection{Bogoliubov Hamiltonian}

Consider the  operator
\begin{eqnarray*}
&&\sum_{\p\neq\0}\big(|\p|^2+\hat{v}(\p)\big)
b_\p^\dagger b_\p+\frac{1}{2}\sum_{\p\neq\0}\hat{v}(\p)\Big(b_\p b_{-\p}+b_\p^\dagger b_{-\p}^\dagger\Big).
\end{eqnarray*} 
acting on $\cH^\ext$.
It commutes with $N_0+N^>$ and $U$. 
In particular, it preserves
$\cH_N^\ext$. Its restriction to $\cH_N^\ext$ will be denoted
$H_{\Bog,N}$.

We can write
\begin{eqnarray}
H_{N,\epsilon}^\ext&=&\frac{1}{2}\hat{v}(\0)(N-1)+H_{\Bog,N}+R_{N,\epsilon}, \label{hamextepsilon}\\
R_{N,\epsilon}&:=
&\frac{1}{2}\sum_{\p\neq\0}\hat{v}(\p)\Big(\Big(\frac{\sqrt{(N_0^\ext-1)N_0^\ext}}{N}
-1\Big)b_\p b_{-\p} +\hc \Big) \notag\\
&&-\frac{1}{N}\sum_{\p\neq\0}\Big(\hat{v}(\p)+\frac{\hat{v}(\0)}{2}\Big)
b_{\p}^{\dagger} b_{\p} N^>+\frac{\hat{v}(\0)}{2N}N^>\notag\\
&&+\frac{\epsilon}{N}\sum_{\p\neq\0}\Big(\hat{v}(\p)+\hat{v}(\0)\Big)b_{\p}^\dagger b_\p N_0^\ext \notag\\
&&+(1+\epsilon^{-1})\frac{1}{2N} v(\0)L^d N^>(N^>-1). \label{rnepsilon}
\end{eqnarray}

Clearly, all $H_{\Bog,N}$ are unitarily equivalent to one another:
$UH_{\Bog,N} U^\dagger=H_{\Bog,N-1}$. It is easy to see that they are
all unitarily equivalent to what we can call the {\em standard Bogoliubov Hamiltonian}:
\begin{eqnarray}
H_\Bog&=&\sum_{\p\neq\0}\big(|\p|^2+\hat{v}(\p)\big)
a_{\p}^\dagger a_\p+\frac{1}{2}\sum_{\p\neq\0}\hat{v}(\p)\Big(a_\p a_{-\p}+a_{\p}^\dagger a_{-\p}^\dagger\Big).\label{stan}
\end{eqnarray}
$H_\Bog$ acts on $\Gamma_\s\Big(l^2\bigl(\frac{2\pi}{L}\ZZ^d\backslash\{\0\}\bigr)\Bigr)$.

We would now like to find a unitary transformation diagonalizing $H_\Bog$. To this end set
\begin{equation*}
A_{\p}:=|\p|^{2}+\hat{v}(\p), \,\,\,\,\,\,B_{\p}:= \hat{v}(\p).
\end{equation*}
Introduce also $\alpha_\p$, $\beta_\p$, $c_\p$ and $s_\p$ by
\begin{eqnarray*}
\alpha_{\p}&=&\frac{1}{B_\p}\Big(A_\p-\sqrt{A_\p^2-B_\p^2}\Big)\ =\ \tanh(2\beta_\p),\\
c_\p&=&\frac{1}{\sqrt{1-\alpha_\p^2}}\ =\ \cosh(2\beta_\p),\\
s_\p&=&\frac{\alpha_\p}{\sqrt{1-\alpha_\p^2}}\ =\ \sinh(2\beta_\p).
\end{eqnarray*}
Now let $S=\e^{-X}$,
where
\begin{gather}
X=\sum_{\p\neq\0}\beta_{\p}\left(a_{\p}^{\dagger}a_{-\p}^{\dagger}-a_{\p}a_{-\p}\right). \label{X} 
\end{gather} 
Then using the Lie formula
\begin{gather*}
\e^{-X}a_{\q}\e^X=\sum_{j=0}^\infty\frac{(-1)^j}{j!}\underset{\hbox{ $j$ times\hskip 3ex}}{[X,...[X,a_{\q}]\dots]} \\
=1+2\beta_{\q}a^{\dagger}_{-\q}+\frac12 4\beta^{2}_{\q}a_{\q}+\ldots
\end{gather*} 
we get 
\begin{gather}
Sa_{\q}S^{\dagger}=c_\q a_{\q}+s_\q a^{\dagger}_{-\q}. \label{sas}
\end{gather}
Therefore,
\begin{eqnarray}
H_{\Bog}&=&\sum_{\p\neq 0}\frac{1}{2}\left(A_{\p}(a_{\p}^{\dagger}a_{\p}+a_{-\p}^{\dagger}a_{-\p})+B_{\p}(a_{\p}^{\dagger}a_{-\p}^{\dagger}+a_{\p}a_{-\p})\right) \notag \\
&=&-\frac{1}{2}\sum_{\p\neq 0}\left(A_{\p}-\sqrt{A_{\p}^{2}-B_{\p}^{2}}\right)\\&&+
\sum_{\p\neq 0}\sqrt{A_{\p}^{2}-B_{\p}^{2}}
\big(c_\p a_\p^\dagger+s_\p a_{-\p}\big)\big(c_\p a_\p+s_\p a_{-\p}^\dagger\big)\\
&=&E_\Bog+S\big(\sum_{\p\neq 0}e_{\p}a_{\p}^{\dagger}a_{\p}\big)S^\dagger,
\label{hbogoliubov}
\end{eqnarray}
where $e_{\p}$ and $E_{\Bog}$ are defined by \eqref{ep} and \eqref{energiabogoliubowa} respectively.
Thus the spectrum of $H_\Bog-E_\Bog$  equals 
\[
\Big\{\sum_{i=1}^j e_{\kk_i}\ :\ 
\kk_1,\dots,\kk_j\in\frac{2\pi}{L}\ZZ^d\backslash\{\0\},\ \ 
j=0,1,2,\dots\Big\}.\]

For further reference note the identities:
\begin{eqnarray}
\alpha_\p&=&\frac{\hat v(\p)}{|\p|^2+\hat v(\p)+
|\p|\sqrt{2\hat v(\p)+|\p|^2}},\notag\\
(c_\p-s_\p)^2=\frac{1-\alpha_\p}{1+\alpha_\p}&=&
\frac{|\p|}
{\sqrt{|\p|^2+2\hat v(\p)}}, \label{(cp-sp)2} \\
s_\p(c_\p-s_\p)=\frac{\alpha_\p}{1+\alpha_\p}&=&
\frac{\hat v(\p)}
{|\p|^2+2\hat v(\p)+|\p|\sqrt{|\p|^2+2\hat v(\p)}}, \label{s(c-s)}\\
2s_\p c_\p(c_\p-s_\p)^2=\frac{\alpha_\p}{(1+\alpha_\p)^2}&=&
\frac{\hat v(\p)}{|\p|^2+2\hat v(\p)}.\notag
\end{eqnarray}
We note also an alternative formula for the Bogoliubov energy:
\[E_\Bog
=
-\frac{1}{2}\sum_{\p\in\frac{2\pi}{L}\mathbb{Z}^{d}\setminus \{ 0\}}\frac{\hat v(\p)^2}{|\p|^{2}+\hat{v}(\p)+|\p|\sqrt{|\p|^2+2\hat{v}(\p)|\p|}}. 
\]

\subsection{Lower bound}

In this section we prove the lower bound part of Theorem \ref{thm}. Using the notation introduced in the previous sections it follows from the following statement:

\begin{thm} Let $c>0$. Then  there exists $C$ such that for any $\kappa\geq 0$ with
 \begin{equation}
  L^{d+2}(L^d+\kappa)\leq cN\label{condit1}
  \end{equation}
 we have
\begin{eqnarray*}
\vecsp\big(\one_{[0,\kappa]}(H_N-E_N)H_N\big)&\geq&
\frac{1}{2}\hat{v}(\0)(N-1)+\vecsp\big(H_\Bog\big)\\&&-CN^{-1/2}L^{d/2+3}(\kappa+L^d)^{3/2}.
\end{eqnarray*}
\label{lower}\end{thm}

The proof of the lower bound starts with estimates analogous to Lemmas 1 and 2 of \cite{Seiringer2011}. Note that in these estimates all operators involve the physical Hilbert space.

\begin{lemma}\label{lemma1}
The ground state energy $E_N$ of $H_N$ satisfies the bounds
\begin{equation}
 0\geq E_N - \frac{1}{2}\left(N-1\right)\hat{v}(\0)\geq \frac{1 }{2}\big(\hat v(\0)-L^d v(\0)\big). \label{lemma1a}
\end{equation}
 \label{lemma1b}  
\end{lemma}
\begin{proof}
The upper bound to the ground state energy follows by using a constant trial wave function $\Psi=L^{-Nd/2}$, which gives
\begin{equation}
\frac{1}{2}(N-1)\hat v(\0)\geq E_N.\label{lemma1acalc2}
\end{equation}

Using $\hat{v}(\p)\geq 0$ for every $\p\in\frac{2\pi}{L}\ZZ^{d}$ we obtain
 $\sup\limits_\x v(\x)=v(\0)$. Moreover, 
\begin{equation*}
\frac{1}{2L^d}\sum_{\p\in\frac{2\pi}{L}\ZZ^{d}\setminus\{\0\}}\hat{v}(\p)\left|\sum_{i=1}^{N}\e^{\ii\p\x_{j}}\right|^{2}\geq 0.
\end{equation*}
This is equivalent to
\begin{equation}
\sum_{1\leq i<j\leq N}v(\x_{i}-\x_{j})\geq\frac{N^{2}}{2L^d}\hat{v}(\0)-\frac{N}{2}v(\0). \label{lemma1acalc}
\end{equation}
Hence,
\begin{equation}
H_N\geq T + \frac{ L^d}{N}  \left(\frac{N^{2}}{2L^d}\hat{v}(\0)-\frac{N}{2}v(\0)\right),
\label{lemma1acalc1}
\end{equation}
and so
\[E_N\geq \frac{ L^d}{N}   \left(\frac{N^{2}}{2L^d}\hat{v}(\0)-\frac{N}{2}v(\0)\right).\]
\end{proof} 

  Let $\kappa\geq0$. 
For brevity we  introduce the following notation for  the spectral
projection onto  the spectral subspace of $H_N$ corresponding to the
energy less than or equal to $E_N+\kappa$:
\[\one_{\kappa}^N:=\one_{[0,\kappa]}(H_N-E_N).\]
 $\one_{\kappa}^N$ can be  understood as a projection acting on the extended space with range in the physical space.

\begin{lemma} There exists $C$ such that
\begin{equation}
 N^>\leq CL^{2}(H_N-E_N+L^d).\label{pafa}
 \end{equation}
Consequently,
\begin{eqnarray}
\label{eq6a}
\one_{\kappa}^NN^{>}\one_{\kappa}^N&\leq &C L^{2}\left(L^d+\kappa\right).
\end{eqnarray}
\end{lemma}

\begin{proof}
Using first (\ref{lemma1acalc1}) and   
(\ref{lemma1acalc2}) we obtain
\begin{eqnarray*}
T&\leq &H_N-E_N-\frac{1}{2}\hat v(\0)+\frac{1}{2}L^d v(\0)\\
&\leq& C(H_N-E_N+L^d).
\end{eqnarray*}
By  (\ref{lemma1b-}) this implies (\ref{pafa}).
\end{proof}

\begin{lemma}\label{lemma2}
We have
\begin{eqnarray}
\label{eq6}
\one_{\kappa}^N(N^{>})^2\one_{\kappa}^N&\leq &C L^{4}\left(L^d+\kappa\right)^{2}.
\end{eqnarray}
\end{lemma}
\begin{proof}
Let $\one_{\kappa}^N\Psi=\Psi$.
As in \cite{Seiringer2011},
\begin{eqnarray}
\langle\Psi|N^{>}T|\Psi\rangle&=&\langle\Psi|N^{>}(H_N-E_N-\frac{1}{2}\kappa)|\Psi\rangle\label{eq1}\\
&&+
N\Big\langle \Psi|Q_1\Big(E_N+\frac12\kappa-\frac{L^d}{N}\sum_{2\leq i<j\leq N}v(\x_{i}-\x_{j})\Big)|\Psi\Big\rangle\label{eq2}\\
&&
-L^d\Big\langle \Psi|Q_1\sum_{2\leq j\leq N}v(\x_{1}-\x_{j})|\Psi\Big\rangle\label{eq3}.
\end{eqnarray}
Using  Schwarz's inequality, the first term can be bounded as
\begin{eqnarray*}
|(\ref{eq1})|&\leq&\|N^>\Psi\|\,\Big\|H_N-E_N-\frac12\kappa\Big\|\\
&\leq&\frac\kappa2\langle\Psi|(N^>)^2\Psi\rangle^{1/2}.
\end{eqnarray*}
Let us estimate the second term. Using (\ref{lemma1acalc}) we get
\begin{eqnarray*}
E_N-\frac{L^d}{N}\sum_{2\leq i<j\leq N}v(\x_{i}-\x_{j})&\leq&\frac{1}{2}(N-1)\hat v(\0)+\frac{L^d}{2N}(N-1)v(\0)-\frac{1}{2N}(N-1)^2\hat v(\0)\\
&=&\frac{1}{2}\frac{N-1}{N}\big(\hat v(\0)+L^d v(\0)\big).
\end{eqnarray*}
Hence,
\begin{eqnarray*}
(\ref{eq2})&\leq&
\Big(\frac\kappa2+\frac12\frac{N-1}{N}\big(\hat v(\0)+L^dv(\0)\big)\Big)N\langle\Psi|Q_1|\Psi\rangle\\
&\leq &
\Big(\frac\kappa2+\frac12\big(\hat v(\0)+L^dv(\0)\big)\Big)\langle\Psi|N^>|\Psi\rangle
\end{eqnarray*}
Finally, let us consider the third term:
\begin{eqnarray*}
\langle\Psi|Q_1 v(\x_1-\x_2)|\Psi\rangle&=&
\langle \Psi|Q_1Q_2 v(\x_1-\x_2)|\Psi\rangle+
\langle \Psi|Q_1P_2 v(\x_1-\x_2)Q_2|\Psi\rangle\\
&&+\langle \Psi|Q_1P_2 v(\x_1-\x_2)Q_2|\Psi\rangle.
\end{eqnarray*}
Using Schwarz's inequality we obtain
\begin{eqnarray*} 
|\langle\Psi|Q_1 Q_2v(\x_1-\x_2)|\Psi\rangle|&\leq&
v(\0)\langle \Psi|Q_1Q_2|\Psi\rangle^{1/2},\\
|\langle\Psi|Q_1 P_2v(\x_1-\x_2)Q_2|\Psi\rangle|&\leq&
v(\0)\langle \Psi|Q_1|\Psi\rangle.
\end{eqnarray*}
Furthermore, denoting $Q_1\psi=\tilde{\psi}$, we have
\begin{eqnarray*}
\langle\Psi|Q_1 P_2v(\x_1-\x_2)P_2|\Psi\rangle= \langle\tilde{\Psi}| P_2v(\x_1-\x_2)P_2|\Psi\rangle= \\
\int\int\int\int \overline{\tilde{\Psi}}(\x_1,\x_2)v(\x_1-\x'_2)\Psi(\x_1,\x''_2)\d \x''_2 \d \x'_2 \d \x_2 \d \x_1 = 
\int v(\x)\d\x \langle \Psi|Q_1P_2|\Psi\rangle.
\end{eqnarray*}
Since $Q_1$ and $P_2$ are commuting orthogonal projections, we have 
$$\langle \Psi|Q_1P_2|\Psi\rangle=\langle \Psi|Q_1P_2 Q_1P_2|\Psi\rangle=\|Q_1P_2\Psi\|\geq 0.$$
Thus
$$\langle\Psi|Q_1 P_2v(\x_1-\x_2)P_2|\Psi\rangle\geq 0.$$
Therefore, using (\ref{Q1}) and (\ref{Q1Q2})
\begin{eqnarray*}
|(\ref{eq3})|&\leq&
v(\0)L^d\Big(\sqrt{\frac{N-1}{N}}\langle\Psi|(N^>{-}1)N^>|\Psi\rangle^{1/2}
+\frac{N-1}{N}\langle\Psi|N^>|\Psi\rangle\Big)\\&\leq&
v(\0)L^d
\Big(\langle\Psi|(N^>)^2\Psi\rangle^{1/2}+
\langle\Psi|N^>|\Psi\rangle\Big).
\end{eqnarray*}
Now 
\begin{eqnarray}
\langle \psi|(N^{>})^{2}|\psi\rangle
&\leq & \frac{L^{2}}{(2\pi)^{2}}\langle \psi|N^{>}T|\psi\rangle.\label{eq4}
\end{eqnarray}
We can add the three estimates, use (\ref{eq4}) and obtain
\begin{eqnarray*}
\langle\Psi|N^>T|\Psi\rangle&
\leq&C(\kappa +L^d)\Big(\langle\Psi|(N^>)^2|\Psi\rangle^{1/2}+
\langle\Psi|N^>|\Psi\rangle\Big)\\
&\leq &CL^{2}(\kappa +L^d)^{2}\\&&+
CL(\kappa +L^d)
\langle\Psi|N^>T|\Psi\rangle^{1/2}.
\end{eqnarray*}
Setting
 $X:=\langle \psi|N^{>}T|\psi\rangle^{1/2}$ we can rewrite this as $X^{2}<c+aX$
in the obvious notation. Solving this inequality we get that
\[X^{2}\leq\frac{a^{2}}{2}+c+\sqrt{a^{2}+4c}.\]
This implies
\begin{eqnarray}\label{eq5}
\one_{\kappa}^NN^{>}T\one_{\kappa}^N&\leq &CL^{2}\left(L^d+\kappa\right)^{2}.\end{eqnarray}
If in addition we use  (\ref{eq4}), we obtain
 (\ref{eq6}).
\end{proof}

\begin{lemma} 
\begin{equation}
\sup_{0<\epsilon\leq1}\one_{\kappa}^NR_{N,-\epsilon}\one_{\kappa}^N
\geq -CN^{-1/2}L^{d/2+3}(L^d+\kappa)^{3/2}.
\end{equation}
\label{lami}\end{lemma}

\begin{proof}
\begin{eqnarray}\one_{\kappa}^NR_{N,-\epsilon}\one_{\kappa}^N
&\geq&\one_{\kappa}^N\frac{1}{2}\sum_{\p\neq\0}\hat v(\p)\Big(\Big(\frac{\sqrt{(N_0-1)N_0}}{N}
-1\Big)b_\p b_{-\p}+\hc
\Big)\one_{\kappa}^N \notag\\
&&-\one_{\kappa}^N\frac{1}{N}\sum_{\p\neq\0}\big(\hat v(\p)+\frac{\hat v(\0)}{2}\big)b_\p^\dagger b_\p N^>\one_{\kappa}^N \notag\\
&&-\epsilon\one_{\kappa}^N\frac{1}{N}\sum_{\p\neq\0}\big(\hat v(\p)+\hat v(\0)\big)b_\p^\dagger b_\p N_0\one_{\kappa}^N\notag\\
&&-\epsilon^{-1}\one_{\kappa}^N\frac{1}{2N} v(\0)L^d (N^>)^2 \one_{\kappa}^N.
\label{pada} 
\end{eqnarray}
Note that the range of  $\one_{\kappa}^N$ is inside the physical space, so whenever possible we replaced $N_0^\ext$ by $N_0$.
It is easy to estimate from below
 various terms on the right of (\ref{pada}) by expressions involving $N^>$. 
The first term requires more work than the others. We have
\begin{eqnarray*}
N-\sqrt{(N_0-1)N_0}&=&\frac{2NN^>-(N^>)^2+N-N^>}
{N+\sqrt{(N-N^>-1)(N-N^>)}}\\
&\leq&2N^>+1.
\end{eqnarray*}
Then we use
\begin{eqnarray*}
\big(\sqrt{(N_0-1)N_0}-N\big)\sum_{\p\neq\0}\hat v(\p)b_\p b_{-\p}+\hc&\geq
&-\Big(\sum_{\p\neq\0}\hat v(\p)b_\p b_{-\p}\Big)^\dagger
\sum_{\p\neq\0}\hat v(\p)b_\p b_{-\p}\\
&&-\big(\sqrt{(N_0-1)N_0}-N\big)^2\\
&\geq&-C(N^>)^2-(2N^>+1)^2\\& \geq&
 -C_1\big((N^>)^2+1\big).
\end{eqnarray*}

To bound the third term 
we use   $N_0\leq N$. We obtain 
\begin{eqnarray*}
\one_{\kappa}^NR_{N,-\epsilon}\one_{\kappa}^N&\geq&-C\one_{\kappa}^N\frac{(N^>)^2+1}{N}\one_{\kappa}^N\\
&&-C\one_{\kappa}^N\frac{(N^>)^2}{N}\one_{\kappa}^N\\
&&-\epsilon C\one_{\kappa}^N N^>\one_{\kappa}^N\\
&&-\epsilon^{-1}C\one_{\kappa}^NL^d\frac{(N^>)^2}{N}\one_{\kappa}^N.
\end{eqnarray*}
Using that $0\leq \epsilon\leq1$ and $L\geq1$, we can partly absorb the first two terms in the fourth:
\begin{eqnarray*}
&\geq&-\frac{C}{N}\one_{\kappa}^N-\epsilon C\one_{\kappa}^N N^>\one_{\kappa}^N-\epsilon^{-1}C\one_{\kappa}^NL^d\frac{(N^>)^2}{N}\one_{\kappa}^N.
\end{eqnarray*}
By (\ref{pafa}) and (\ref{eq6}), this can be estimated by
 \begin{eqnarray}
&\geq&-CN^{-1} -\epsilon CL^{2}(L^d+\kappa)-
\epsilon^{-1}CN^{-1}L^{d+4}(L^d+\kappa)^2.
\label{errorlowerbound}
\end{eqnarray}
 Setting  $\epsilon= c^{-1/2}L^{d/2+1}(L^d+\kappa)^{1/2}N^{-1/2}$ in (\ref{errorlowerbound}), which by Condition (\ref{condit1}) is less than $1$,  
we bound it by
\[\geq\ -CN^{-1}-CN^{-1/2}L^{d/2+3}(L^d+\kappa)^{3/2}.\]
Using $L\geq1$, we can absorb the first term in the second. \end{proof}

\begin{proof}[Proof of Thm \ref{lower}]
 Recall inequality (\ref{hamiltonianboundlower}), which
implies for $0<\epsilon\leq1$
\begin{equation*}
\one_{\kappa}^N H_N\one_{\kappa}^N\geq
\one_{\kappa}^N
\left(\frac{1}{2}\hat v(\0)(N-1)+H_{\Bog,N}+R_{N,-\epsilon}\right)\one_{\kappa}^N.
\end{equation*}
Thus it suffices to apply Lemma \ref{lami} and  the min-max principle.
\end{proof}

\begin{proof}[Proof of Thm \ref{thm} (1)]
First set $\kappa=0$. Then Condition (\ref{condit1}) becomes Condition (\ref{war1a}) and we obtain Thm \ref{thm} (1a).

Next set $\kappa=K_N^j(\p)$. Then Condition (\ref{condit1}) is equivalent to the conjunction of Conditions (\ref{war1a}) and (\ref{war1b}). We obtain Thm \ref{thm} (1b).
\end{proof}

\subsection{Upper bound}
In this section we prove the following theorem, which implies the upper bound of Theorem \ref{thm}:
\begin{thm}
Let $c>0$. Then there exist $c_1>0$ and $C$ such that if $\kappa\geq0$ and \begin{eqnarray}
L^{d+2}(\kappa+L^{d-1})&\leq &cN,\label{condit2}\\
L^{2}(\kappa+L^{d-1})&\leq &c_1N\label{condit2a}
\end{eqnarray}
then 
\begin{eqnarray*}
\vecsp\big(H_N\big)&\leq&
\frac{1}{2}\hat{v}(\0)(N-1)+\vecsp\Big(\one_{[0,\kappa]}(H_\Bog-E_\Bog)H_\Bog
\Big)
\\&&+CN^{-1/2}L^{d/2+3}(\kappa+L^{d-1})^{3/2}.
\end{eqnarray*}
\label{upper}\end{thm}

For brevity, we set
\[\one_\kappa^\Bog:=\one_{[0,\kappa]}(H_{\Bog,N}-E_\Bog).\]
 From now on, to simplify the notation we will also write $H_\Bog$ instead of $H_{\Bog,N}$, even though this is an abuse of notation. ($H_{\Bog,N}$ is  unitarily equivalent, but strictly speaking distinct from (\ref{stan})).

We also set
\[d_\p:=Sb_\p S^\dagger\]
 where $S$ is defined as in \eqref{X}  with operators $a$'s replaced by $b$'s.
Clearly,
\begin{eqnarray*}
d_\p&=\ c_\p b_{\p}+s_\p b_{-\p}^\dagger,\ \ \ d_\p^\dagger&=\ c_\p b_\p^\dagger+s_{\p}b_{-\p}.
\end{eqnarray*}

\begin{lemma}
There exist $C_1,C_2$ such that
\begin{equation} H_\Bog-E_\Bog\geq C_1L^{-2}N^>-C_2L^{d-1}.\label{qe1}
\end{equation}
Consequently,
\begin{equation}
\one_\kappa^\Bog  N^>\one_\kappa^\Bog\leq CL^{2}(L^{d-1}+\kappa).
\label{qe1a}
\end{equation}
\label{lemi1}\end{lemma}
\proof
Using \eqref{hbogoliubov} we have that 
\begin{eqnarray*}H_{\Bog}-E_\Bog&=&\sum_{\p\neq 0}e_{\p}Sb_{\p}^{\dagger}b_{\p}S^{\dagger}\\
&\geq&\sum_{\p\neq 0} \frac{\pi\sqrt{8\hat v(\0)}}{L}Sb_{\p}^{\dagger}b_{\p}S^\dagger=
\frac{\pi\sqrt{8\hat v(\0)}}{L}SN^{>}S^{\dagger}.\end{eqnarray*}
Now
\begin{eqnarray*}SN^{>}S^{\dagger}\ =\ \sum_{\pm\p\neq\0} \big(d_\p^\dagger d_\p+d_{-\p}^\dagger d_{-\p}\big)&=&
\sum_{\pm\p\neq\0}\Big( (c_\p^2+s_\p^2)\big(b_{\p}^\dagger b_{\p}+b_{-\p}^\dagger b_{-\p}\big)\notag\\
&&+ 2c_\p s_\p\big(b_{\p}^\dagger b_{-\p}^\dagger+b_{\p} b_{-\p}\big)+ 2s_\p^2\Big).
\end{eqnarray*}
(When we write $\pm \p$ under the summation symbol, we sum over all pairs $\{\p,-\p\}$). Using
\[b_{\p}^\dagger b_{-\p}^\dagger+b_{\p} b_{-\p}\geq
-\big(b_{\p}^\dagger b_{\p}+b_{-\p}^\dagger b_{-\p}+1\big)
\]we obtain
\begin{eqnarray} &&\sum_{\pm\p\neq\0}\big(d_\p^\dagger d_\p+d_{-\p}^\dagger d_{-\p}\big)\notag\\&\geq&\sum_{\pm\p\neq\0}\Big(
(c_\p-s_\p)^2\big(b_\p^\dagger b_\p+b_{-\p}^\dagger b_{-\p}\big)-2s_\p(c_\p-s_\p)\Big).\label{sns}
\end{eqnarray}
By \eqref{(cp-sp)2} we know that $\inf\limits_{\p\neq 0}(c_{\p}-s_{\p})^2\geq\frac{\sqrt{2}\pi}{\sqrt{\hat{v}(0)}L}$. Also, \eqref{s(c-s)} yields
\[
\frac{1}{L^d}\sum_{\pm\p\neq\0}s_{\p}(c_{\p}-s_{\p})<\infty,\]
uniformly in $L$.  Thus
\begin{eqnarray*}H_{\Bog}-E_\Bog &\geq& \frac{C}{L}SN^{>}S^{\dagger}\\&\geq& \frac{C_1}{L^2}\sum_{\pm\p\neq\0}\big(b_\p^\dagger b_\p+b_{-\p}^\dagger b_{-\p}\big)-\frac{C_2 L^{d-1}}{L^d}\sum_{\pm\p\neq\0}2s_\p(c_\p-s_\p)\\&
=&C_1L^{-2}N^>-C_2L^{d-1}.
\end{eqnarray*}
This proves (\ref{qe1}), which can be rewritten as
\begin{equation}
 N^>\leq C_1^{-1}L^{2}(H_\Bog-E_\Bog+C_2L^{d-1}).\label{qe1+}
 \end{equation}
By the definition of $\one_\kappa^\Bog$ we have
$$\one_\kappa^\Bog(H_\Bog - E_\Bog)\one_\kappa^\Bog \leq \kappa$$ 
which implies
(\ref{qe1a}). \qed
\\
\\
We shall now prove a similar lemma, but for quartic operators. The calculations are therefore a little bit more burdensome.
\begin{lemma}
Set
\begin{eqnarray*}
M&:=&\sum_{\p\neq\0}(c_\p-s_\p)^2b_\p^\dagger b_\p,\\
A_1&:=&\sum_{\p\neq\0}2s_\p(c_\p-s_\p),\\
A_2&:=&\sum_{\p\neq0}4(c_\p-s_\p)^2s^2_\p.\end{eqnarray*}
Then
\begin{equation}
(SN^>S^\dagger+A_1)^2\geq M^2-A_2.
\label{lemik}
\end{equation}
\end{lemma}

\proof
Recall that when we write $\pm \p$ under the summation symbol, we sum over all pairs $\{\p,-\p\}$. We have
\begin{eqnarray}&&\Bigg(\sum_{\pm\p\neq\0}\big(
d_\p^\dagger d_\p+d_{-\p}^\dagger d_{-\p}+2s_\p(c_\p-s_\p)\big)\Bigg)^2\notag\\
&=&
\sum_{\pm\p,\pm \q\neq\0}\Big(
d_\p^\dagger \big(d_\q^\dagger d_\q+d_{-\q}^\dagger
d_{-\q}+2s_\q(c_\q-s_\q)\big)d_\p\notag\\&&\ \ 
+d_{-\p}^\dagger \big(d_\q^\dagger d_\q+d_{-\q}^\dagger
d_{-\q}+2s_\q(c_\q-s_\q)
\big)d_{-\p}\Big)\notag
\\
&&+\sum_{\pm\p,\pm \q\neq\0}2s_\p(c_\p-s_\p)\Big(d_\q^\dagger
d_\q+d_{-\q}^\dagger d_{-\q}+2s_\q(c_\q-s_\q)\Big)\notag\\
&&+
\sum_{\pm\p\neq\0}\big(
d_\p^\dagger d_\p+d_{-\p}^\dagger d_{-\p}\big).
\notag
\end{eqnarray}
Using \eqref{sns} we bound this from below by
\begin{eqnarray}&&\sum_{\pm\p,\pm \q\neq\0}(c_\q-s_\q)^2\Big(d_\p^\dagger
 \big(b_\q^\dagger b_\q+b_{-\q}^\dagger b_{-\q}\big)d_\p+
 d_{-\p}^\dagger\big(b_\q^\dagger b_\q+b_{-\q}^\dagger b_{-\q}\big)d_{-\p}\Big)\notag \\
&&+\sum_{\pm\p,\pm \q\neq\0}2s_\p(c_\p-s_\p)(c_\q-s_\q)^2\big(b_\q^\dagger b_\q+b_{-\q}^\dagger b_{-\q}\big)\notag
\\
&&+
\sum_{\pm\p\neq\0}\Big((c_\p-s_\p)^2\big(
b_\p^\dagger b_\p+b_{-\p}^\dagger b_{-\p}\big)-2s_\p(c_\p-s_\p)\Big).
\notag\\
&=&\sum_{\pm\p,\pm \q\neq\0}(c_\q-s_\q)^2\Big(b_\q^\dagger
 \big(d_\p^\dagger d_\p+d_{-\p}^\dagger d_{-\p}\big)b_\q+
 b_{-\q}^\dagger\big(d_\p^\dagger d_\p+d_{-\p}^\dagger
 d_{-\p}\big)b_{-\p}\Big)\notag \\
&&+\sum_{\pm\p\neq\0}(c_\p-s_\p)^2\Big(
2s_\p^2\big(b_\p^\dagger b_\p+b_{-\p}^\dagger b_{-\p}\big)+2c_\p
s_\p\big(b_\p^\dagger b_{-\p}^\dagger+b_\p b_{-\p}\big)+2s_\p^2\Big)
\notag\\
&&+\sum_{\pm\p,\pm \q\neq\0}2s_\p(c_\p-s_\p)(c_\q-s_\q)^2\big(b_\q^\dagger b_\q+b_{-\q}^\dagger b_{-\q}\big)\notag
\\
&&+
\sum_{\pm\p\neq\0}\Big((c_\p-s_\p)^2\big(
b_\p^\dagger b_\p+b_{-\p}^\dagger b_{-\p}\big)-2s_\p(c_\p-s_\p)\Big).
\notag\\&=&\sum_{\pm\p,\pm \q\neq\0}(c_\q-s_\q)^2\Big(b_\q^\dagger
 \big(d_\p^\dagger d_\p+d_{-\p}^\dagger
 d_{-\p}+2s_\p(c_\p-s_\p)\big)b_\q\notag\\
&&\ \ \ \ +
 b_{-\q}^\dagger\big(d_\p^\dagger d_\p+d_{-\p}^\dagger
 d_{-\p}+2s_\p(c_\p-s_\p)\big)b_{-\p}\Big)\notag \\
&&+\sum_{\pm\p\neq\0}(c_\p-s_\p)^2\Big(
\big(2s_\p^2+1\big)\big(b_\p^\dagger b_\p+b_{-\p}^\dagger b_{-\p}\big)+2c_\p
s_\p\big(b_\p^\dagger b_{-\p}^\dagger+b_\p b_{-\p}\big)\Big)
\notag
\\
&&+
\sum_{\pm\p\neq\0}\big((c_\p-s_\p)^22s_\p^2-2s_\p(c_\p-s_\p)\big).
\notag\end{eqnarray}
Using \eqref{sns} one more time, we bound this from below by
\begin{eqnarray}
&&
\sum_{\pm\p,\pm \q\neq\0}(c_\q-s_\q)^2(c_\p-s_\p)^2\Big(b_\q^\dagger
 \big(b_\p^\dagger b_\p+b_{-\p}^\dagger b_{-\p})b_\q+
 b_{-\q}^\dagger\big(b_\p^\dagger b_\p+b_{-\p}^\dagger b_{-\p}\big)b_{-\q}\Big)\notag\\
&&+\sum_{\pm\p\neq\0}(c_\p-s_\p)^2(2s_\p^2-2c_\p s_\p+1)
\big(b_\p^\dagger b_\p+b_{-\p}^\dagger b_{-\p}\big)\notag\\
&&+\sum_{\pm\p\neq\0}\big((-2c_\p
s_\p+2s_\p^2)(c_\p-s_\p)^2-2s_\p(c_\p-s_\p)\big)\notag\\
&=&\Bigg(\sum_{\pm\p\neq0}(c_\p-s_\p)^2\big(b_\p^\dagger
b_{\p}+b_{-\p}^\dagger b_{-\p}\big)\Bigg)^2-\sum_{\pm\p\neq0}4(c_\p-s_\p)^2 c_\p s_\p.\notag
\end{eqnarray}
\qed

\begin{lemma}
There exist $C_1,C_2$ such that
\begin{equation} \big(H_\Bog-E_\Bog\big)^2\geq C_1L^{-4}(N^>)^2-C_2L^{2d-2}.\label{qe2}
\end{equation}
Therefore,
\begin{equation}
\one_\kappa^\Bog ( N^>)^2\one_\kappa^\Bog \leq CL^{4}(L^{d-1}+\kappa)^2.
\label{qe2a}
\end{equation}
\end{lemma}

\proof As in the proof of Lemma \ref{lemi1},
\begin{eqnarray}\big(H_{\Bog}-E_\Bog\big)^2&\geq&
\frac{\big(\pi\sqrt{8\hat v(\0)}\big)^2}{L^{2}}\big(SN^{>}S^{\dagger}\big)^2.\label{pqo}
\end{eqnarray}
For any $\delta>0$, Lemma \ref{lemik} implies
\[(1+\delta)(SN^>S^\dagger)^2+(1+\delta^{-1})A_1^2\geq M^2-A_2.\]
Moreover, the limits $\lim\limits_{L\to\infty}\frac{A_1}{L^d}$ and 
$\lim\limits_{L\to\infty}\frac{A_2}{L^d}$ exist.
Therefore,

\begin{eqnarray*}
\big(SN^{>}S^{\dagger}\big)^2&\geq&M^2-CL^{2d}.
\end{eqnarray*}
Using (\ref{pqo}) and $M\geq C_1L^{-1}N^>$, we easily conclude that (\ref{qe2}) holds. Hence
\begin{equation*}
(N^>)^2\leq C_2^{-1}L^{4}\big((H_\Bog-E_\Bog)^2+C_3L^{2d-2}\big),
\end{equation*}
which easily implies  (\ref{qe2a}).
\qed

Suppose now that
$G$ is a smooth non-negative function on $[0,\infty[$ such that
\begin{equation}G(s)=
\begin{cases}
 1, & \text{if }s\in[0,\frac{1}{3}]\\
 0, & \text{if }s\in[1,\infty[.
\end{cases} \label{F}
\end{equation}
Set
\[A_N:=G(N^>/N),\ \ A_N^\nph:=\one-A_N.\]
The operator $A_N$ will serve as a smooth approximation to the projection onto the physical space. 
Set \[Y_\kappa:=\one_\kappa^\Bog A_N.\]
\begin{lemma}
We have
\[\one_\kappa^\Bog-Y_\kappa Y_\kappa^\dagger=O\big(
L^{2}(\kappa+L^{d-1})N^{-1}\big).\]
\label{lle1}\end{lemma}

\proof We have
\begin{eqnarray*}
&&\one_\kappa^\Bog-Y_\kappa Y_\kappa^\dagger\ =\ \one_\kappa^\Bog\big(1-G(N^>/N)^2\big)
\one_\kappa^\Bog \\
&=&
\one_\kappa^\Bog (N^>/N)^{1/2}
\Big(\big(1-G(N^>/N)^2\big)(N^>/N)^{-1}\Big)
(N^>/N)^{1/2}
\one_\kappa^\Bog
.\end{eqnarray*}
But by the spectral theorem
\[\|\big(1-G(N^>/N)^2\big)(N^>/N)^{-1}\|=\sup\limits_s\{|(1-G(s)^2)s^{-1}|\}<\infty,\] and by (\ref{qe1a})
\[(N^>/N)^{-1/2}
\one_\kappa^\Bog=O\big(
L(\kappa+L^{d-1})^{1/2}N^{-1/2}\big).\]
\qed

Let $0<c_0<1$. If 
\begin{equation} 
\|\one_\kappa^\Bog-Y_\kappa Y_\kappa^\dagger\|\leq c_0, \label{pqi}
\end{equation}
then
$Y_\kappa Y_\kappa^\dagger$ is invertible on $\Ran \one_\kappa^\Bog$.
 We will denote by
$\big(Y_\kappa Y_\kappa^\dagger\big)^{-1}$ the corresponding inverse. We set
\[X_\kappa:=
\big(Y_\kappa Y_\kappa^\dagger\big)^{-1/2}.\]
On the orthogonal complement of $\Ran\one_\kappa^\Bog$ we extend it by
$0$.

 By Lemma \ref{lle1} and Condition \eqref{condit2a} with a sufficiently small $c_1$, we can guarantee that \eqref{pqi} holds with, say, $c_0\leq1/2$. Therefore, in what follows $X_\kappa$ is well defined.

\begin{lemma}
\begin{equation}
 \one_\kappa^\Bog-X_\kappa=O\big(
L^{2}(\kappa+L^{d-1})N^{-1}\big).\label{pqo1}
\end{equation}
\end{lemma}

\proof
We know already that $\big(Y_\kappa Y_\kappa^\dagger\big)^{-1}$ exists. By the convergent Neumann series
$$\big(Y_\kappa Y_\kappa^\dagger\big)^{-1}=\sum_{j=0}^{\infty}\left(\one-Y_\kappa Y_\kappa^\dagger\right)^{j}$$
and thus
$$\|\big(Y_\kappa Y_\kappa^\dagger\big)^{-1}\|\leq (1-c_0)^{-1}.$$
This implies
\[\|\one_\kappa^\Bog-\big(Y_\kappa Y_\kappa^\dagger\big)^{-1}\|\leq c_0(1-c_0)^{-1}\]
which is $O\big(
L^{2}(\kappa+L^{d-1})N^{-1}\big)$. This implies (\ref{pqo1}) by the spectral theorem. \qed

\begin{lemma}
\[X_\kappa[A_N,[A_N,H_\Bog]]X_\kappa=O\big(N^{-2} L^{2}(\kappa+L^{d-1})\big).\]
\label{???}\end{lemma}

\proof We have
\[[N^>,[N^>,H_\Bog]]=2\sum \hat v(\p)(b_\p b_{-\p}+b_\p^\dagger b_{-\p}^\dagger).\]
Using
\[
-b_\p^\dagger b_\p-b_{-\p}^\dagger b_{-\p}-1
\leq
b_\p b_{-\p}+b_\p^\dagger b_{-\p}^\dagger\leq b_\p^\dagger b_\p+b_{-\p}^\dagger b_{-\p}+1
\] we obtain
\[-C(N^>+L^d)\leq
[N^>,[N^>,H_\Bog]]\leq
C(N^>+L^d).\]
This implies
\begin{equation}
\Big \|(N^>+L^d)^{-1/2}\Big[N^>,[N^>,H_\Bog]\Big](N^>+L^d)^{-1/2}\Big\|\leq C. 
\label{powq}
\end{equation}
Now notice that for operators $S$ and $T$ one has 
\begin{gather*}
\left[\e^{\ii tS},T\right]=\int_{0}^{t}\frac{\d}{\d u}\e^{\ii uS}T\e^{-\ii uS}\d u\e^{\ii tS}=\int_{0}^{t}\ii\e^{\ii uS}[S,T]\e^{\ii (t-u)S}\d u  
\intertext{which together with the representation mentioned above yields}
[A_N,H_\Bog]=\frac{1}{2\pi}\int_{\RR}\hat{G}(t)\int_{0}^{t}\ii\e^{\ii u\frac{N^>}{N}}\left[\frac{N^>}{N},H_\Bog\right] \e^{\ii (t-u)\frac{N^>}{N}}\d u \d t.
\intertext{We can repeat this calculation to obtain}
\Big[A_N,[A_N,H_\Bog]\Big]= \\
\frac{-1}{4\pi^2}\int_{\RR}\hat{G}(p)\int_{0}^{p}\e^{\ii s\frac{N^>}{N}}\frac{1}{2\pi}\int_{\RR}\hat{G}(t)\int_{0}^{t}[\frac{N^>}{N},\e^{\ii u\frac{N^>}{N}}[\frac{N^>}{N},H_\Bog] \e^{\ii (t-u)\frac{N^>}{N}}]\d u \d t\e^{\ii (p-s)\frac{N^>}{N}}\d s \d p \\
=\frac{-1}{4\pi^2}\int_{\RR}\int_{0}^{p}\int_{\RR}\int_{0}^{t}\hat{G}(p)\hat{G}(t)\e^{\ii (s+u)\frac{N^>}{N}}[\frac{N^>}{N},[\frac{N^>}{N},H_\Bog]]\e^{\ii (t-u+p-s)\frac{N^>}{N}}\d u \d t \d s \d p \\
=\frac{-1}{2\pi^2 N^2}\int_{\RR}\int_{0}^{p}\int_{\RR}\int_{0}^{t}\hat{G}(p)\hat{G}(t)\e^{\ii (s+u)\frac{N^>}{N}}\left(\sum \hat v(\q)(b_\q b_{-\q}+b_\q^\dagger b_{-\q}^\dagger)\right)\e^{i(t-u+p-s)\frac{N^>}{N}}\d u \d t \d s \d p.
\intertext{Thus, because all integrals are finite (and uniform with respect to $n$ and $L$), we only need to bound the expression}
\Big\langle\psi\Big|X_\kappa \e^{\ii a\frac{N^>}{N}}\left(\sum \hat {v}(\q)(b_\q b_{-\q}+b_\q^\dagger b_{-\q}^\dagger)\right)\e^{\ii c\frac{N^>}{N}} X_{\kappa}\Big|\psi\Big\rangle 
\end{gather*}
for some $a,c\in\RR$.\\
We can now insert the identity operator $\one=(N^>+L^d)^{1/2}(N^>+L^d)^{-1/2}$ between $X_\kappa$ and $\e^{\ii a\frac{N^>}{N}}$ and use Schwarz inequality to obtain 
\begin{gather*}
\Big\langle\psi\Big|X_\kappa \e^{\ii a\frac{N^>}{N}}\left(\sum \hat {v}(\q)(b_\q b_{-\q}+b_\q^\dagger b_{-\q}^\dagger)\right)\e^{\ii c\frac{N^>}{N}} X_{\kappa}\Big|\psi\Big\rangle \leq \Big\langle X_{\kappa}\psi \Big|(N^>+L^d)\Big|X_{\kappa}\psi\Big\rangle^{1/2} \\
\times \Big\langle X_{\kappa}\psi\Big|\e^{-\ii c\frac{N^>}{N}}(N^>+L^d)^{1/2}A^{\dagger}A(N^>+L^d)^{1/2}\e^{\ii c\frac{N^>}{N}}\Big|X_{\kappa}\psi\Big\rangle^{1/2}
\intertext{where}
A= (N^>+L^d)^{-1/2}\left(\sum \hat{v}(\p)(b_\p b_{-\p}+b_\p^\dagger b_{-\p}^\dagger)\right)(N^>+L^d)^{-1/2}.
\intertext{The self-adjoint operator $A$ is bounded. This follows from the fact that if}
-\langle\phi|B|\phi\rangle\leq\langle\phi|D|\phi\rangle\leq\langle\phi|B|\phi\rangle
\intertext{for some self-adjoint operators $B$ and $D$ (with $B>0$), then by setting $\psi=B^{1/2}\phi$ one obtains}
-\langle\psi|\psi\rangle\leq\langle\psi|B^{-1/2}DB^{-1/2}|\psi\rangle\leq\langle\psi|\psi\rangle
\intertext{and thus}
-\one\leq B^{-1/2}DB^{-1/2}\leq \one.
\end{gather*}
Now apply this for $B=(N^>+L^d)$ and $D=\sum \hat{v}(\p)(b_\p b_{-\p}+b_\p^\dagger b_{-\p}^\dagger)$. Now we use the operator inequality 
\begin{gather*}
\e^{-\ii c\frac{N^>}{N}}(N^>+L^d)^{1/2}A^{\dagger}A(N^>+L^d)^{1/2}\e^{\ii c\frac{N^>}{N}}\leq \|A\|^2 (N^>+L^d)
\end{gather*}
which together with Lemma \ref{lemi1} yields the proof.
\qed

We define
\begin{equation}
Z_\kappa:=X_\kappa A_N=
\big(\one_\kappa^\Bog A_N^2\one_\kappa^\Bog
\big)^{-1/2}
 A_N.\label{zk}
 \end{equation}
Clearly, $Z_\kappa$ is a partial isometry with initial space $\Ran (A_{N}\one_{\kappa}^{\Bog})$ and final space $\Ran (\one_{\kappa}^{\Bog})$.

\begin{lemma}\label{lemma4}
\begin{eqnarray*}
\one_\kappa^\Bog (H_\Bog-E_\Bog) \one_\kappa^\Bog&=&Z_\kappa (H_\Bog-E_\Bog) Z_\kappa^\dagger\\
&&+ O\big(L^{2}(L^{d-1}+\kappa)\kappa N^{-1}\big)\\
&&+
 O\big(L^{2}(L^{d-1}+\kappa)N^{-2}\big)
.\end{eqnarray*}
\end{lemma}
\begin{proof}
We have
\begin{eqnarray}
\one_\kappa^\Bog (H_\Bog-E_\Bog) \one_\kappa^\Bog
&=&
\big(\one_\kappa^\Bog -X_\kappa\big)
(H_\Bog-E_\Bog) \one_\kappa^\Bog\label{tag1}
\\
&&+X_\kappa
(H_\Bog-E_\Bog) 
\big(\one_\kappa^\Bog -X_\kappa\big)\label{tag2}\\
&&+X_\kappa
(H_\Bog-E_\Bog) 
X_\kappa;\notag
\end{eqnarray}
\begin{eqnarray}
X_\kappa
(H_\Bog-E_\Bog) 
X_\kappa&=&-
X_\kappa A_N^\nph
(H_\Bog-E_\Bog)  A_N^\nph
X_\kappa\notag\\
&&+
X_\kappa 
(H_\Bog-E_\Bog)  A_N^\nph
X_\kappa\label{tag3}\\
&&+
X_\kappa A_N^\nph
(H_\Bog-E_\Bog)  
X_\kappa\label{tag4}\\
&&+
X_\kappa A_N
(H_\Bog-E_\Bog)  A_NX_\kappa;\notag
\end{eqnarray}
\begin{eqnarray}
-X_\kappa A_N^\nph
(H_\Bog-E_\Bog)  A_N^\nph
X_\kappa
&\!\!\!\!{=}&\!\!\!\!\!\!
{-}\frac12
X_\kappa (A_N^\nph)^2
(H_\Bog-E_\Bog) 
X_\kappa\label{tag5}\\
&&\!\!\!\!\!\!{-}\frac12X_\kappa 
(H_\Bog-E_\Bog)  (A_N^\nph)^2
X_\kappa\label{tag6}\\
&&\!\!\!\!\!\!{+}\frac12
X_\kappa \big[A_N^\nph,[A_N^\nph,
H_\Bog]\big]
X_\kappa.\label{tag7}
\end{eqnarray}
The error term in the lemma equals the sum of (\ref{tag1}),...,(\ref{tag7}).
By (\ref{pqo1}),
\[ (\ref{tag1}) ,(\ref{tag2})= O\big(L^{2}(L^{d-1}+\kappa)\kappa N^{-1}\big).\]
By (\ref{qe2a}),
\[(\ref{tag3}),\dots, (\ref{tag6})= O\big(L^{2}(L^{d-1}+\kappa)\kappa N^{-1}\big).\]
By Lemma \ref{???},\[ (\ref{tag7})= O\big(L^{2}(L^{d-1}+\kappa) N^{-2}\big).\]
\end{proof}

\begin{lemma}\label{lemma89}
Assume (\ref{condit2}). Then
\begin{eqnarray}
\inf_{0<\epsilon\leq1}Z_\kappa R_{N,\epsilon} Z_\kappa^{\dagger}&\leq&
C L^{d/2+3}(L^{d-1}+\kappa)^{3/2}N^{-1/2}
\label{epsilo}
\end{eqnarray}
\end{lemma}

\begin{proof}
From the definition \eqref{zk} of $Z_{\kappa}$ it follows that the final space for $Z_{\kappa}^{\dagger}$ is $\Ran \one_\kappa^\Bog A_N$. Thus 
\begin{eqnarray*}
Z_\kappa R_{N,\epsilon} Z_\kappa^{\dagger}&\leq&
Z_\kappa \frac{1}{2}\sum_{\p\neq\0}\hat v(\p)\Big(\Big(\frac{\sqrt{(N_0-1)N_0}}{N}
-1\Big) b_\p b_{-\p}
+\hc\Big)  Z_\kappa^{\dagger}\notag\\&&+Z_\kappa \frac{\hat v(\0)}{2N}N^> Z_\kappa^{\dagger}\notag\\
&&+\epsilon Z_\kappa \frac{1}{N}\sum_{\p\neq\0}\big(\hat v(\p)+\hat v(\0)\big)b_\p^\dagger b_\p N_0^\ext Z_\kappa^{\dagger} \notag\\
&&+(1+\epsilon^{-1})Z_\kappa 
\frac{1}{2N} v(\0)L^d N^>(N^>-1) Z_\kappa^{\dagger}\notag\\
&\leq&\one_\kappa^\Bog C\frac{(N^>)^2+1}{N}\one_\kappa^\Bog
\notag \\
&&+\one_\kappa^\Bog C\frac{N^>}{N}\one_\kappa^\Bog \notag \\
&&+\epsilon\one_\kappa^\Bog CN^>\one_\kappa^\Bog \notag\\
&&+(1+\epsilon^{-1})\one_\kappa^\Bog C\frac{L^d(N^>)^2}{N}\one_\kappa^\Bog . 
\end{eqnarray*}
Using $\epsilon\leq1$, we can simplify the bound as follows:
\begin{eqnarray}
&\leq &\one_\kappa^\Bog \frac{C}{N}
{+}\epsilon\one_\kappa^\Bog CN^>\one_\kappa^\Bog 
{+}\epsilon^{-1}\one_\kappa^\Bog
C\frac{L^d(N^>)^2}{N}\one_\kappa^\Bog,
 \label{rnepsilon+1}
\end{eqnarray}
By (\ref{qe1a}) and (\ref{qe2a}).
this can be estimated by
\begin{equation*}
 CN^{-1}+\epsilon C L^{2}(L^{d-1}+\kappa)+
\epsilon^{-1} C L^{d+4}(L^{d-1}+\kappa)^2N^{-1}.
\end{equation*}
Setting  $\epsilon=c^{-1/2}N^{-1/2}L^{d/2+1}(L^{d-1}+\kappa)^{1/2}$, which is less than $1$ by   Condition (\ref{condit2}),
we obtain
\begin{equation}
 CN^{-1}+ C L^{d/2+3}(L^{d-1}+\kappa)^{3/2} N^{-1/2}.
\label{rnepsilon+}
\end{equation}
By changing $C$, the second term can obviously absorb $CN^{-1}$.
\end{proof}

\begin{proof}[Proof of Thm \ref{upper}]
$Z_\kappa$ is a partial isometry with the initial space contained in the physical space and the final projection $\one_\kappa^\Bog$. 
Therefore,
\begin{eqnarray*}
\vecsp H_N&\leq&\vecsp\Bigg( Z_\kappa^\dagger Z_\kappa H_N
 Z_\kappa^\dagger Z_\kappa\Big|_{\Ran Z_\kappa^\dagger} \Bigg)\\
&=&\vecsp \Bigg(
Z_\kappa H_N
 Z_\kappa^\dagger\Big|_{\Ran\one_\kappa^\Bog}\Bigg).
\end{eqnarray*}

\begin{eqnarray}Z_\kappa H_N
 Z_\kappa^\dagger&\leq&Z_\kappa H_{N,\epsilon}
 Z_\kappa^\dagger\notag\\
&=&
\frac{1}{2}\hat v(\0)(N-1)\one_\kappa^\Bog
+H_\Bog\one_\kappa^\Bog\notag\\
&&+Z_\kappa (H_\Bog -E_\Bog)Z_\kappa^\dagger-(H_\Bog-E_\Bog)\one_\kappa^\Bog
\label{lab1}\\
&&
+Z_\kappa R_{N,\epsilon}  Z_\kappa^\dagger.
\label{lab2}
\end{eqnarray}
By Lemma \ref{lemma4},
\begin{eqnarray}
(\ref{lab1})&\leq&
CL^{2}(L^{d-1}+\kappa)\kappa N^{-1}\label{szac1}\\
&&+
CL^{2}(L^{d-1}+\kappa)N^{-2}.\label{szac2}
\end{eqnarray}
Using $\kappa<\kappa+L^{d-1}$ and later \eqref{condit2} we have
\begin{eqnarray*}
(\ref{szac1})&\leq &CL^2(L^{d-1}+\kappa)^{2} N^{-1}\\
&\leq&CL^{-d/2+1}(L^{d-1}+\kappa)^{3/2} N^{-1/2}.
\end{eqnarray*}
Thus (\ref{szac1}) can be absorbed in
$O(L^{d/2+3}(L^{d-1}+\kappa)^{3/2} N^{-1/2})$.  

We easily check that the same is true in the case of  (\ref{szac2}).  To bound \eqref{lab2} we use Lemma \ref{lemma89}.
\end{proof}

\begin{proof}[Proof of Thm \ref{thm} (2)]
First set $\kappa=0$. Then Condition (\ref{condit2})
 becomes Condition (\ref{war2a}) and 
 Condition (\ref{condit2a})
 becomes Condition (\ref{war2a+}). 
We obtain Thm \ref{thm} (2a).

Next set $\kappa=K_\Bog^j(\p)$. Then Condition (\ref{condit2}) is equivalent to the conjunction of Conditions (\ref{war2a}) and (\ref{war2b}).
Condition (\ref{condit2a}) is equivalent to the conjunction of Conditions (\ref{war2a+}) and (\ref{war2b+}).
This shows Thm \ref{thm} (2b).
\end{proof}

\subsection{Proof of Corollary \ref{coro}}\label{proof of coro}

The proof of the corollary is based on the following lemma:

\begin{lemma}\label{lemmaappendix}
\begin{enumerate}
\item Let $b>1$, $-1-\frac{1}{d+1}\leq \alpha \leq 1$ and $L^{4d+6}\leq b N^{1-\alpha}$.
 Then \vspace{2ex}
\begin{enumerate}
\item \hspace{8ex} $\frac{1}{2}\hat v(\0)(N-1)+E_\Bog\leq E_N+O(N^{-\alpha/2})$;
\item if $K_N^j(\p)\leq (bN^{1-\alpha}L^{-d-6})^{1/3}$, then
\[ \frac{1}{2}\hat v(\0)(N-1)+E_\Bog +K_\Bog^j(\p)\leq E_N+ K_N^j(\p)+O(N^{-\alpha/2});\]
\item if $0\leq\alpha\leq1$ and $K_N^j(\p)\leq bN^{1-\alpha}L^{-d-6}$, then
 \begin{eqnarray*}
 \frac{1}{2}\hat v(\0)(N-1)+E_\Bog+K_\Bog^j(\p)&\leq& E_N+ K_N^j(\p)\\&&+\left(1+K_N^j(\p)\right)O(N^{-\alpha/2}).\end{eqnarray*}
\end{enumerate}
\item
Let $b>1$,  $-1-\frac{1}{2d+1}< \alpha \leq 1$ and $L^{4d+3}\leq bN^{1-\alpha}$. Then there exists $M$ such that if $N>M$, then \vspace{2ex}
\begin{enumerate}\item
\hspace{8ex}
$ E_N\leq \frac{1}{2}\hat v(\0)(N-1)+E_\Bog+O(N^{-\alpha/2})$;
\item if $K_\Bog^j(\p)\leq (bN^{1-\alpha}L^{-d-6})^{1/3}$, then
\[ E_N+K_N^j(\p)\leq \frac{1}{2}\hat v(\0)(N-1)+E_\Bog+K_\Bog^j(\p)+O(N^{-\alpha/2});\]
\item  if  $0<\alpha\leq1$ 
and  $K_\Bog^j(\p)\leq bN^{1-\alpha}L^{-d-6}$, then
\begin{eqnarray*}
 E_N+K_N^j(\p)&\leq&\frac{1}{2}\hat v(\0)(N-1)+E_\Bog+ K_\Bog^j(\p)\\
&&+\left(1+K_\Bog^j(\p)\right)O(N^{-\alpha/2}).\end{eqnarray*}
\end{enumerate}\end{enumerate}
\end{lemma}

\begin{proof}
To prove (1), resp. (2) we use Thm \ref{thm} (1), resp. (2). We give a proof of the latter part, since it is slightly more involved (because of the parameter $c_1$).

(2a): First we check  Condition (\ref{war2a}):
\begin{eqnarray}
L^{2d+1}&=&\big(L^{4d+3}\big)^{\frac{2d+1}{4d+3}}\leq \big(bN^{1-\alpha}\big)^{\frac{2d+1}{4d+3}}.
\label{noto9}\end{eqnarray}
For  $-1-\frac{1}{2d+1}\leq \alpha$ we have (\ref{noto9})$\leq cN$.

Next,
\begin{eqnarray}
L^{d+1}&=&\big(L^{4d+3}\big)^{\frac{d+1}{4d+3}}\leq\big(bN^{1-\alpha}\big)^{\frac{d+1}{4d+3}}.
\label{noto9+}\end{eqnarray}
We have (\ref{noto9+})$\leq cN N^{\frac{-3d-2-\alpha(d+1)}{4d+3}}$. Therefore, for  $-1-\frac{1}{2d+1}\leq \alpha$,  
  Condition (\ref{war2a+}) is satisfied for large enough $N$.

Then we apply Thm \ref{thm} (2)(a), using
\[N^{-1/2}L^{2d+\frac32}\leq N^{-1/2}(bN^{1-\alpha})^{1/2}=O(N^{-\alpha/2}).\]

(2b): We check  Condition (\ref{war2b}):
\begin{eqnarray}\notag
K_{\Bog}^j(\p)&\leq&\big(bN^{1-\alpha}L^{-d-6}\big)^{1/3}\\\notag
&\leq&\big(bN^{1-\alpha}L^{-d-6}\big)^{1/3}\big(bN^{1-\alpha}L^{-4d-3}\big)^{\frac{2d}{12d+9}}\\
&=&C N^{\frac{(1-\alpha)(2d+1)}{12d+9}}L^{-d-2}.\label{noto}
\end{eqnarray}
For $ -1-\frac{1}{2d+1}\leq\alpha$ we have (\ref{noto})$\leq CNL^{-d-2}$.

Also
\begin{eqnarray*}
\ref{noto}=O(N^{\frac{-(2+\alpha)(4d+3)+2d(1-\alpha)}{12d+9}}) NL^{-2-d}
\end{eqnarray*}
which implies
\begin{eqnarray}\notag
K_{\Bog}^j(\p)
&\leq& O(N^{\frac{-(2+\alpha)(4d+3)+2d(1-\alpha)}{12d+9}}) NL^{-2}
.\notag
\end{eqnarray}
Therefore, if $ -1-\frac{1}{2d+1}<\alpha$,  Condition (\ref{war2b+}) is satisfied for large enough $N$.

We clearly have
\begin{eqnarray} N^{-1/2}L^{d/2+3}\big(K_{\Bog}^j(\p)+L^{d-1}\big)^{3/2}&\leq&
2^{3/2}N^{-1/2}L^{d/2+3}K_{\Bog}^j(\p)^{3/2}\label{noto0}\\&&+2^{3/2}N^{-1/2}L^{2d+\frac32}.
\label{noto1}\end{eqnarray} 
We already know that   (\ref{noto1}) is $O(N^{-\alpha/2})$. Thus to apply Thm \ref{thm} (2b) we need only to bound (\ref{noto0}):
\[N^{-1/2}L^{d/2+3}K_{\Bog}^j(\p)^{3/2}\leq
N^{-1/2}L^{d/2+3}\big(bN^{1-\alpha}L^{-d-6}\big)^{1/2}=O(N^{-\alpha/2}).
\]

(2c):  Condition (\ref{war2b}) is trivially satisfied, since for $L\geq1$, $N\geq1$ and $\alpha > 0$
\[K_{\Bog}^j(\p)\leq bN^{1-\alpha}L^{-d-6}\leq bNL^{-d-2}.
\]

We have
\begin{eqnarray}\notag
K_{\Bog}^j(\p)
&\leq&bN^{1-\alpha}L^{-d-6}
L^{d+4}\\
&=&O(N^{-\alpha}) NL^{-2}
.\notag
\end{eqnarray}
Therefore Condition (\ref{war2b+}) is satisfied for large enough $N$.

To apply Thm \ref{thm} (2b) we  bound  (\ref{noto0}):
\begin{eqnarray*}
N^{-1/2}L^{d/2+3}K_{\Bog}^j(\p)^{3/2}&=&
b^{1/2}N^{-\alpha/2}\big(b^{-1}N^{-(1-\alpha)}L^{d+6}\big)^{1/2}K_{\Bog}^j(\p)^{3/2}\\
&\leq& O(N^{-\alpha/2})K_{\Bog}^j(\p).\end{eqnarray*}
\end{proof}

\noindent{\em Proof of Corollary \ref{coro}} Part (1) follows directly from Lemma \ref{lemmaappendix} (1a) and (2a).

Let us prove (2). To simplify notation we drop $\p$ from
 $K_N^j(\p)$ and $K_\Bog^j(\p)$.

Assume first that $K_N^j\leq K_\Bog^j$. By Lemma \ref{lemmaappendix} (1b) for some $C>0$ 
\begin{eqnarray*}
\frac{1}{2}\hat v(\0)(N-1)+E_\Bog+K_\Bog^j&\leq &E_N+K_N^j+CN^{-\alpha/2}\\
&\leq &E_N+K_\Bog^j+CN^{-\alpha/2}.\end{eqnarray*}
Thus
\begin{eqnarray*}\frac{1}{2}\hat v(\0)(N-1)+E_\Bog-E_N+K_\Bog^j-CN^{-\alpha/2}&\leq& K_N^j\leq K_\Bog^j.
\end{eqnarray*}
By Lemma \ref{lemmaappendix} (2a), \[-CN^{-\alpha/2}\leq
\frac12 \hat{v}(\0)(N-1)+E_\Bog-E_N.\]
Hence the statement follows.

Assume now that $K_\Bog^j\leq K_N^j$. Then we use Lemma \ref{lemmaappendix} (2b) and obtain
\begin{gather*}
K_\Bog^j\leq K_N^j\leq \frac{1}{2}\hat v(\0)(N-1)+E_\Bog-E_N+K_\Bog^j+CN^{-\alpha/2}.
\end{gather*}
By Lemma \ref{lemmaappendix} (1a),
\[\frac{1}{2}\hat v(\0)(N-1)+E_\Bog-E_N\leq CN^{-\alpha/2}.\]
 The statement follows again. This ends the proof of part (2).

The proof of part (3) is similar, except that one  uses Lemma \ref{lemmaappendix} (1c) and (2c).
\qed

\section{Summary and outlook}
\subsection{Summary}
Let us give here a brief recap of what has been presented and proven in this thesis.
\subsubsection{Quasiparticles and excitation spectrum}
In the first part of the thesis we proposed different concepts and definitions related to the notion of a quasiparticle and excitation spectrum of \textit{translation invariant quantum systems}. The presented approach relies on a spectral and Hamiltonian-based point of view.

In Section \ref{excandquasi} we have defined the \textit{energy-momentum spectrum} and the \textit{excitation spectrum} of a translation invariant quantum system. Based on these two concepts, we have introduced the \textit{critical velocity} and the \textit{energy gap}. Having done that we introduced the notion of a \textit{quasiparticle quantum system} and described the properties of the excitation spectrum of such system. In particular, Theorem \ref{spectralpropquantusys} provides a description of the critical velocity and energy gap for these systems.

Because the notion of a quasiparticle quantum system may be too idealistic in real situations (e.g. in condensed matter physics), we introduced also the notion of \textit{quasiparticle-like systems} and a system which has \textit{quasiparticle-like excitation spectrum}. At the end of the first part we have also analyzed properties of translation invariant quantum systems which posses the fermionic parity superselection rule. We concluded the first part of the thesis by analyzing the energy-momentum spectrum of an non-interacting Fermi gas.  

\subsubsection{Mean-field methods and calculation of excitation spectrum}
In the second part of the thesis we presented approximate methods which allow us to calculate the excitation spectrum of a translation invariant interacting Bose and Fermi gas.

In the context of the Bose gas, we presented the well known \textit{Bogoliubov approximation} and the so-called \textit{improved Bogoliubov method}. The latter one was based on the procedure of \textit{minimization of the Hamiltonian over pure Gaussian states}. This method, together with the mean-field approach which neglects certain higher order terms in the transformed Hamiltonian, suggests that the excitation spectrum of a homogeneous Bose gas is quasiparticle-like.

Later, we applied the same approach to the homogeneous Fermi gas. It implies - again under certain assumptions - that the excitation spectrum of a homogeneous Fermi gas is quasiparticle-like. We also presented figures which show how the conjectured excitation spectrum should look like.

At the end of the second part we summarized the procedure of minimization of Hamiltonians over pure Gaussian states in more abstract terms. This led to Theorem \ref{main}, which we call \textit{Beliaev's Thoerem}. This theorem explains why this minimization scheme leads to a picture of approximate quasiparticles for systems with quite general Hamiltonians.  
 
\subsubsection{Rigorous justification of the Bogoliubov approximation}
In the last part of the thesis we presented a rigorous result concerning the Bogoliubov approximation. This result (Theorem \ref{thm}), which can be seen as the main result of this thesis, expresses the idea that the Bogoliubov approximation becomes exact for the mean-field Hamiltonian
$$H_N^{L}=-\sum_{i=1}^{N}\Delta^{L}_{i}+\frac{L^d}{N}\sum_{1\leq i<j\leq N} v^{L}(\x_{i}-\x_{j})$$ 
for large $N$ (number of particles) and $L$ (size of the system) provided that the volume does not grow too fast. We call this Hamiltonian mean-field because of the coupling term $\frac{L^d}{N}=\frac{1}{\rho}$ located in front of the interaction term. In the particular limit of large $N$ and $L$ for which the Bogoliubov approximation becomes exact, the coupling term tends to zero. 

The proof of this result involves the concept of the extended Fock space. In the extended Fock space we allow for a "negative number of particles" in the zero momentum state. The benefit of using the extended Fock space is that one can introduce operators which simplify certain calculations.

\subsection{Outlook}

Finally, let us mention a few directions we would like to follow in our future work and which arise from the considerations described in this thesis.

\subsubsection{Rigorous results in the fermionic case}
One of the possible directions to follow concerns a rigorous analysis of the excitation spectrum of a homogeneous Fermi gas. There are only few results concerning this system, mostly related to the so-called \textit{BCS functional} and \textit{gap equation} at positive temperature (e.g. \cite{HHSS2008,BHS2013,HS2008PRB}).

The Hartree-Fock-Bogoliubov approximation presented in Subsection \ref{HFBBCS} suggests in particular that a homogeneous Fermi gas at zero temperature has a positive energy gap and positive critical velocity. It suggests also that the excitation spectrum is quasiparticle-like. The hope is that these conjectures could be proven in the same spirit as in the case of the rigorous justification of the Bogoliubov approximation, in particular involving some kind of mean-field limit. 

Recall that arguments supporting the above, rather vague statements involve the assumption that the interaction is (at least partially) attractive. Note that rigorous methods used in part three of this thesis  use the positivity of the potential and its Fourier transform. Thus, other methods will be needed in the fermionic case.

\subsubsection{Beliaev damping}   
The term \textit{Beliaev damping}, as explained in Section \ref{s1},   describes a process of dissipation due to collisions between quasiparticles of an interacting Bose gas. Using perturbation theory  Beliaev calculated that the pole $e_{p}$ of the Green's function of an interacting Bose gas contains an imaginary part with
\begin{equation*}
\mathfrak{I}  e_{p} \approx |p|^{5}.
\end{equation*}
It took a long time before Beliaev's prediction was confirmed experimentally . The experimental realisation of Bose-Einstein Condensation in 1995 (\cite{Ketterle}) gave new impetus to theoretical work on Beliaev damping. One should mention here the work by Giorgini (\cite{G1998}) who rederived Belieav's damping term starting with the Gross-Pitaevskii equation and using the random phase approximation.

One might expect that a mathematically more rigorous understanding of this phenomenon could be obtained by introducing a Hamiltonian model which would rely on the already proven quasiparticle-like structure of the excitation spectrum. This approach could involve the so-called Friedrichs model (\cite{Do1965}) and an application of the Fermi Golden Rule. 

\subsubsection{Correlation functions}
Our proof concerning the justification of the Bogoliubov approximation in the mean-field limit involves the determination of an approximate ground state. Since correlation functions are given by expectation values of products of creation and annihilation operators in the ground state, one can hope that one could prove results concerning correlation functions in a more rigorous way, maybe using methods developed in this thesis.

Perhaps a similar procedure could be also applied to the rigorous study of the so-called \textit{van Hove form factor} (\cite{Hove1954}). 

An analysis of correlation and Green's functions could also offer an insight into the so-called Hugenholtz-Pines theorem (\cite{HP59}, see also \cite{GN64}).

\clearpage

\addcontentsline{toc}{section}{References}
\bibliographystyle{acm}
\bibliography{doktorat}

\end{document}